\documentclass[12pt]{article}
\usepackage[slantedGreek]{mathpazo}

\usepackage[utf8]{inputenc}
\usepackage{color}
\usepackage{float}
\usepackage{amsmath}
\usepackage{amssymb}
\usepackage{setspace}
\usepackage[sort, authoryear]{natbib}
\usepackage{chngcntr}

\makeatletter

\floatstyle{ruled}
\newfloat{algorithm}{tbp}{loa}
\providecommand{\algorithmname}{Algorithm}
\floatname{algorithm}{\protect\algorithmname}

\usepackage{url}
\usepackage{booktabs}
\usepackage{array}
\usepackage{multirow}
\usepackage{float}

\usepackage{graphicx}
\usepackage{amsthm}
\usepackage{mathrsfs}

\usepackage[toc,page]{appendix}

\usepackage{enumitem}
\setlist[enumerate]{topsep=0pt,itemsep=-1ex,partopsep=1ex,parsep=1ex}

\usepackage{parskip}
\usepackage{verbatim}
\usepackage{xcolor}
\usepackage{subcaption}

\definecolor{darkblue}{rgb}{0.0,0.0,0.55}
\RequirePackage[colorlinks,citecolor=darkblue,urlcolor=darkblue,linkcolor=darkblue]{hyperref} 

\usepackage[format=plain, labelfont=bf, singlelinecheck=off]{caption} 

\numberwithin{equation}{section}
\usepackage{algorithm}
\usepackage{algpseudocode}
\newlength{\continueindent}
\usepackage{etoolbox}

\newcommand*{\ALG@customparshape}{\parshape 2 \leftmargin \linewidth \dimexpr\ALG@tlm+\continueindent\relax \dimexpr\linewidth+\leftmargin-\ALG@tlm-\continueindent\relax}
\apptocmd{\ALG@beginblock}{\ALG@customparshape}{}{\errmessage{failed to patch}}

\algnewcommand{\algorithmicgoto}{\textbf{go to}}%
\algnewcommand{\Goto}[1]{\algorithmicgoto~\ref{#1}}%

\numberwithin{equation}{section}

\newcommand{\calX}{{\cal X}}

\usepackage{amsthm}
\newtheorem{theorem}{\textsc{{Theorem}}}

\newtheorem{lemma}{\textsc{{Lemma}}}
\newtheorem{proposition}{\textsc{{Proposition}}}

\theoremstyle{remark}

\newcommand{\bfb}{\mathbf{b}}
\newcommand{\bfc}{\mathbf{c}}
\newcommand{\bfh}{\mathbf{h}}

\newcommand{\bff}{\mathbf{f}}

\newcommand{\bfr}{\mathbf{r}}

\newcommand{\bfv}{\mathbf{v}}
\newcommand{\bfw}{\mathbf{w}}
\newcommand{\bfx}{\mathbf{x}}
\newcommand{\bfy}{\mathbf{y}}

\newcommand{\bfA}{\mathbf{A}}
\newcommand{\bfB}{\mathbf{B}}

\newcommand{\bfF}{\mathbf{F}}

\newcommand{\bfH}{\mathbf{H}}

\newcommand{\bfK}{\mathbf{K}}

\newcommand{\bfQ}{\mathbf{Q}}
\newcommand{\bfR}{\mathbf{R}}

\newcommand{\bfT}{\mathbf{T}}

\newcommand{\Ell}{\boldsymbol \ell}

\newcommand{\bfbeta}{\boldsymbol \beta}

\newcommand{\bfxi}{\boldsymbol \xi}

\newcommand{\bfmu}{\boldsymbol \mu}
\newcommand{\bfphi}{\boldsymbol \phi}

\newcommand{\bfgamma}{\boldsymbol \gamma}
\newcommand{\bfrho}{\boldsymbol \rho}
\newcommand{\bfsigma}{\boldsymbol \sigma}

\newcommand{\bfSigma}{\boldsymbol \Sigma}

\newcommand{\bfOmega}{\boldsymbol \Omega}

\newcommand{\diag}{\text{diag}}

\expandafter\def\expandafter\normalsize\expandafter{%
    \normalsize
    \setlength\abovedisplayskip{6pt}
    \setlength\belowdisplayskip{6pt}
    \setlength\abovedisplayshortskip{6pt}
    \setlength\belowdisplayshortskip{6pt}
}

\makeatother

\defcitealias{NOAA}{NOAA2019}

\usepackage{geometry}
\usepackage{setspace}
\begin{document}

\title{\vspace{-1cm} \baselineskip=20pt  \bf Multifidelity Computer Model Emulation with High-Dimensional Output: An Application to Storm Surge
}



\date{}
\singlespacing
\maketitle
\baselineskip=15pt
\vspace{-1.75cm}
\begin{center}
  Pulong Ma\footnote{{\it Address for correspondence}: Pulong Ma, Assistant Professor in School of Mathematical and Statistical Sciences at Clemson University; 220 Parkway Drive, Clemson, SC 29634; Email: plma@clemson.edu}
\\
  {\it School of Mathematical and Statistical Sciences, Clemson University, Clemson, SC, USA}\\
  Georgios Karagiannis \\
  {\it Durham University, Durham, UK}\\
  Bledar A. Konomi\\
  {\it University of Cincinnati, Cincinnati, OH, USA}\\
  Taylor G. Asher \\
  {\it University of North Carolina at Chapel Hill, Chapel Hill, NC, USA}\\
  Gabriel R. Toro \\
  {\it Lettis Consultants International, Inc., USA} \\
  and Andrew T. Cox\\
  {\it Oceanweather, Inc., USA}
  \hskip 5mm\\
\end{center}

\begin{abstract}
Hurricane-driven storm surge is one of the most deadly and costly natural disasters, making precise quantification of the surge hazard of great importance. Surge hazard quantification is often performed through physics-based computer models of storm surges. Such computer models can be implemented with a wide range of fidelity levels, with computational burdens varying by several orders of magnitude due to the nature of the system. The threat posed by surge makes greater fidelity highly desirable, however such models and their high-volume output tend to come at great computational cost, which can make detailed study of coastal flood hazards prohibitive. These needs make the development of an emulator combining high-dimensional output from multiple complex computer models with different fidelity levels important. We propose a parallel partial autoregressive cokriging model to predict highly-accurate storm surges in a computationally efficient way over a large spatial domain. This emulator has the capability of predicting storm surges as accurately as a high-fidelity computer model given any storm characteristics over a large spatial domain. 
\end{abstract}

\noindent%
{\it Keywords:}  Autoregressive cokriging; High-dimensional output; Multifidelity computer model; Storm surge; Uncertainty quantification 


\singlespacing
\section{Introduction}

Storm surge is one of the most severe natural disasters that can lead to significant flooding in coastal areas that brings multi-billion dollar damages and is responsible on average for half of lives lost from hurricanes \citep{rappaport2014fatalities}. On average, inflation-normalized direct economic damages to the U.S. (1900-2005) are estimated at \$10 billion per year and increasing as a result of storm surges \citep{pielke2008normalized}. Since 2005, there have been 12 hurricanes whose total U.S. damages exceeded \$10 billion \citep{NOAA2019}. For instance, Hurricane Katrina (2005) caused over 1500 deaths and total estimated damages of \$75 billion in the New Orleans area and along the Mississippi coast as a result of storm surge \citep{FEMA2006}. To mitigate these impacts, studies are carried out to evaluate the probabilistic hazard \citep[e.g.,][]{niedoroda2010analysis,Cialone2017} and risk \citep[e.g.,][]{fischbach2016bias} from coastal flooding through a synthesis of computer modeling, statistical modeling, and extreme-event probability computation. Here computer modeling is used to predict the storm surge hazard initialized by hurricanes, statistical modeling is used to determine the distribution of hurricane characteristics, and extreme-event probability is used to assess the flood hazard. These studies support development and application of flood insurance rates, building codes, land use planning/development, infrastructure design and construction, and related goals by providing hazard levels at a range of frequencies \citep[e.g.,][]{aerts2014evaluating}. Similarly, forecast simulations are used to support a wide array of operational needs, most notably disaster mitigation and evacuation planning/preparedness \citep[e.g.,][]{blanton2018integrated,georgas2016stevens}. 

The ADCIRC ocean circulation model \citep{Luettich2004,westerink2008basin} is the primary computer model in the U.S. to predict storm surges in coastal areas. It was certified by the Federal Emergency Management Agency (FEMA) for coastal flood hazard study and has been successfully used in a large number of applications, including FEMA flood hazard map updates \cite[e.g.,][]{FEMA2008, niedoroda2010analysis, Jensen2012, Hesser2013} and in support of United States Army Corps of Engineers (USACE) projects \citep[e.g.,][]{Wamsley2013,Cialone2017}. These studies develop surge and wave hazard elevations corresponding to a range of frequencies, from the 50\% annual exceedance level to the 0.01\% annual exceedance level. 

In the risk assessment of coastal flood hazard, ADCIRC needs to be run for a large number of storm characteristics.
ADCIRC can be run at different levels of accuracy due to the sophistication of the physics incorporated in mathematical models, accuracy of numerical solvers, and resolutions of meshes. Although ADCIRC can be run efficiently in parallel on large supercomputers \citep{Tanaka2011}, its computational cost scales with the cube of the spatial resolution, meaning that very high fidelity models are several orders of magnitude more expensive than lower fidelity ones.  By incorporating the physics of ocean waves, ADCIRC can generate storm surges with even greater fidelity, but this adds another order of magnitude to the run time as compared to the uncoupled ADCIRC model \citep{Dietrich2012}. For instance, a single high-resolution, coupled simulation in Southwestern Florida takes roughly 2000 core-hours on a high-performance supercomputer \citep{xsede}. As a result, research often uses coarser models without wave effects, even though the importance of more advanced, detailed models in estimating storm surge has been well-demonstrated \citep[e.g.,][]{marsooli2018numerical,yin2016coupled}. 

An important scientific demand is to develop an \emph{emulator} - a fast probabilistic approximation of a simulator - that can produce highly accurate storm surges over a large spatial domain. The development of an emulator is needed in probabilistic flood studies dealing with climate change where a broad host of variables need to be considered. For instance, risk studies with relatively simple surge models estimate the current coastal flooding risk to New York City alone at over 100 million U.S. dollars per year \citep{aerts2014evaluating}. But these estimates are highly sensitive to underlying assumptions that remain to be explored \citep[e.g.,][]{de2011effect,fischbach2016bias}. Similarly, high-fidelity studies of coastal flooding changes associated with sea level rise and/or climate change have shown complex, nonlinear changes in flooding patterns that simpler studies cannot uncover \citep{liu2019physical}. 
 
The main scientific goal in this article is to develop an emulator that can not only predict highly accurate storm surges but also can be run very quickly  over a large spatial domain. The development of an emulator for the high-fidelity storm surge model directly can be computationally prohibitive, since the high-fidelity storm surge model requires a tremendous amount of computing resources just to obtain a single run. An alternative is to develop an emulator that can combine a limited number of highly accurate simulations from a high-fidelity but expensive surge model and a larger number of less accurate simulations from a low-fidelity but cheaper surge model. Combining simulations from different fidelity levels relies on the idea that, after quantifying discrepancy between models of different fidelity levels, information from the low-fidelity surge model can facilitate prediction at high fidelity.  

Several statistical works have been proposed to combine output from simulators at different fidelity levels based on a well-known geostatistical method called \emph{cokriging} \citep[see Chapter 3 of][]{Cressie1993}. The idea to emulate multifidelity computer models is originated in \cite{Kennedy2000} using an autoregressive cokriging model - a  cokriging model with an Markov assumption. The work in \cite{Kennedy2000} has been extended in several ways. For instance, \cite{Qian2008} propose a Bayesian hierarchical formulation. \cite{Gratiet2013} devises an efficient Bayesian approach to estimate model parameters. \cite{konomikaragiannisABTCK2019} introduce nonstationarity by partitioning the input space via a Bayesian tree structure. \cite{Ma2020OBayes} develops objective Bayesian analysis for an autoregressive cokriging model. A typical assumption in univariate autoregressive cokriging models is the hierarchically nested design, which is not always preferred as discussed in \cite{Qian2008}. All these works focus on univariate computer model output without addressing the high-dimensionality challenge and possibly non-nested design in the storm surge application.

To address challenges due to high-dimensionality in the output space and non-nested design, we propose the parallel partial (PP) cokriging emulation methodology that can deal with high-dimensional output over a large spatial domain and non-nested design. In particular, the PP cokriging emulator is an extension of the PP kriging emulator \citep{Gu2016} on different levels of fidelity code with massive output.  To allow non-nested design, we develop a data-augmentation technique for parameter estimation via Monte Carlo expectation-maximization (MCEM) algorithm. For prediction, we develop a sequential prediction approach in which prediction at a higher fidelity level requires prediction to be made at a lower fidelity level. The proposed PP cokriging emulator explicitly introduces nonstationarity through spatially-varying the mean parameters and variance parameters in Gaussian processes at each fidelity level and also allows fast computations in an empirical Bayesian framework. 


The remainder of the article is organized as follows. Section~\ref{sec: storm surge} introduces the storm surge application with two storm surge simulators. In Section~\ref{sec: model formulation}, we present the proposed methodology to handle high-dimensional output and non-nested design.   In Section~\ref{sec: application}, an analysis of storm surge simulators is performed with the proposed methodology. Section~\ref{sec: Discussion} is concluded with discussions and possible extensions.

\section{Application of interest: Storm Surge} \label{sec: storm surge}
This section describes the storm surge simulators with their discrepancy highlighted and the simulation design that is used for this study.
\subsection{Storm Surge Simulators}
 ADCIRC is a hydrodynamic circulation numerical model that solves the shallow-water equations for water levels and horizontal currents using the finite-element method over an unstructured gridded domain representing bathymetric and topographic features \citep{Luettich2004}. Information about ADCIRC can be accessed at \url{https://adcirc.org}. In what follows, we will refer to ADCIRC as the low-fidelity simulator, meanwhile we will refer to the coupled ADCIRC + SWAN model as the high-fidelity simulator. The latter incorporates the Simulating WAves Nearshore (SWAN) wave model \citep{Booij1999, Zijlema2010} in order to enhance system physics and accuracy. This is achieved by tightly coupling the ADCIRC and SWAN models, simulating them on the same unstructured mesh \citep{Dietrich2011, Dietrich2012}. This coupling is important for accurate prediction of waves in the nearshore, which ride on top of storm surge and bring substantial destructive power to coastal structures and defenses.

\begin{figure}[htbp]
\begin{subfigure}{.7\textwidth}
\centering
\makebox[\textwidth][c]{ \includegraphics[width=1.0\textwidth, height=0.15\textheight]{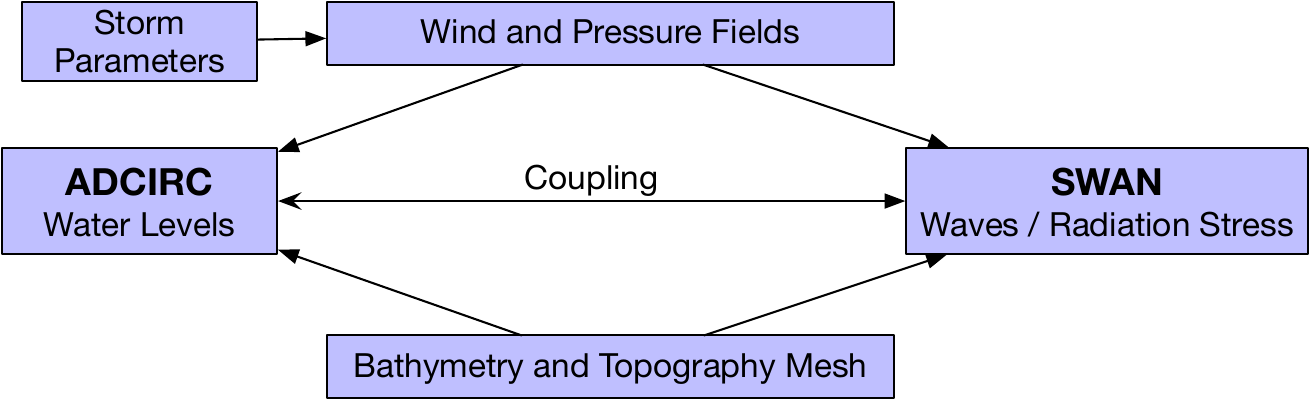}}
\end{subfigure}%
\begin{subfigure}{.3\textwidth}
  \centering
  \includegraphics[width=.9\linewidth, height=0.20\textheight]{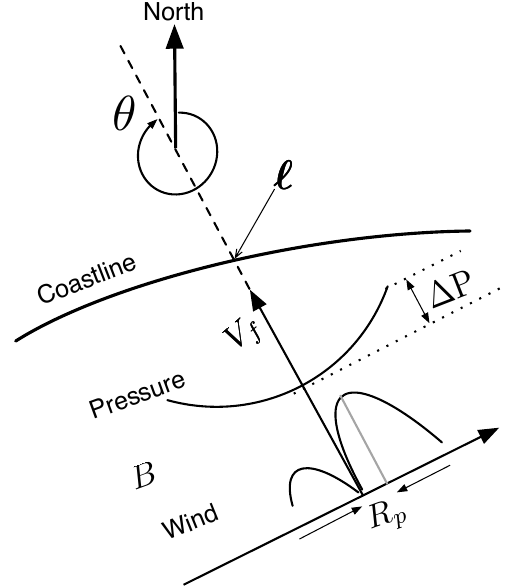}
\end{subfigure}
\caption{Diagram of storm surge models based on ADCIRC and storm parameters. The left panel shows ADCIRC and its coupling with SWAN. Then right panel shows storm parameters when a storm approaches the coastline \citep{Toro2010}. ADCIRC has bathymetry and topography mesh and wind and pressure fields as inputs. Storm parameters are used to derive the wind and pressure fields.}
\label{fig: ADCIRC + SWAN}
\end{figure}

Figure~\ref{fig: ADCIRC + SWAN} shows a basic diagram for the ADCIRC simulator and the ADCIRC + SWAN simulator. In this study, we focus on six input parameters to characterize the storm: $\Delta P$, $R_p$, $V_f$, $\theta$, $B$, $\mathbf{\Ell}$, with their physical meaning given in Table~\ref{table: input parameters}. These parameters will be treated as inputs in the ADCIRC simulator.  Although the behavior of hurricanes is much more complex than this characterization, using this simplified storm parameterization is acceptable for the probabilistic characterization of future storms since no practical or robust model exists to represent these effects for surge-frequency calculations for future storms. In this application, the response variable of interest is the peak surge elevation (PSE) for each landfalling hurricane simulated from these surge models, where the peakness is taken across time over the course of one storm. 

\begin{table}[htbp]
\centering
\normalsize
   \caption{Storm characteristics parameters.}
  {\resizebox{1.0\textwidth}{!}{%
  \setlength{\tabcolsep}{3.5em}
   \begin{tabular}{c l } 
   \toprule \noalign{\vskip 1.5pt} 
Input variables    &  Physical meaning \\ \noalign{\vskip 1.5pt} 
\midrule \noalign{\vskip 2.5pt} 
$\Delta P$  & central pressure deficit of the storm (mb) \\ \noalign{\vskip 2.5pt} 
$R_p$       & scale pressure radius in nautical miles \\ \noalign{\vskip 2.5pt} 
$V_f$       & storm's forward speed (m$/$s) \\ \noalign{\vskip 2.5pt} 
$\theta$    & storm's heading in degrees clockwise from north \\ \noalign{\vskip 2.5pt} 
$B$         & Holland's $B$ parameter (unitless) \\ \noalign{\vskip 2.5pt} 
$\mathbf{\Ell}$       & landfall location in latitude and longitude \\  \noalign{\vskip 2.5pt}  \bottomrule
   \end{tabular}%
   }}
   \label{table: input parameters}
\end{table}

\subsection{Model Validation}
 
In this work, our goal is to develop an emulator for this coupled ADCIRC + SWAN model for storm surge prediction. The detailed validation of the coupled ADCIRC + SWAN model has been performed against a variety of data sources (tide harmonic constituent data, surge measurements, water level gages, and wave buoy data) and storm events (Hurricane Charley 2004, Tropical storm Gabrielle 2001, Hurricane Donna 1960) in FEMA coastal flood hazard studies \citep[e.g.,][]{FEMA2017}. The coupled ADCIRC + SWAN model has been validated to observed surges in several studies for historical storms, and shows good model performance with typical errors below 0.3 meters \citep{Dietrich2012,Dietrich2011,FEMA2017,Cialone2017}.

From a physics perspective, ADCIRC + SWAN explicitly incorporates ubiquitous wave effects on water levels and currents through the SWAN model. Before modern computing, wave setup effects were determined via approximate equations and added onto surge hazard estimates.  Now, model coupling is standard practice for both researchers and practitioners to more correctly represent the physical processes, with clear benefits \citep[e.g.,][]{Dietrich2011}.  ADCIRC + SWAN has been used in all regions of the U.S. Gulf and Atlantic coasts (including our study region) by both FEMA and USACE in their flood hazard studies.  Flooding forecasting work still struggles to utilize coupled models because of the computational cost, and our work also has the potential to aid such efforts, partly thanks to the speed of emulator construction. There is a substantial need to develop an efficient emulator for the high-fidelity simulator ADCIRC + SWAN to aid flood hazard studies and forecasting work.

\subsection{Model Simulation Setup} \label{sec: model simulation}
In FEMA coastal flood hazard studies \citep[e.g.,][]{FEMA2017}, storm surge hazard assessment is accomplished via the annual exceedance probability (AEP) for hurricane-prone areas. The quantification of AEP requires large-scale numerical simulations from ADCIRC. Statistical modeling is used to develop characteristics of synthetic storms based on historical tropical cyclones. The wind and pressure fields are used as inputs into hydrodynamic models such as ADCIRC to predict storm surges. In this application, the ADCIRC simulator and the ADCIRC + SWAN simulator are run on the same mesh with 148,055 nodes (spatial points).  The mesh and simulation characteristics were constructed for a FEMA coastal flood hazard study in Southwest Florida \citep{FEMA2017} and these simulations were carried out using the same standards and methods documented in that study. The peak surge hazard estimates produced by that study are considered to represent current conditions (i.e. sea level and climate in the years around when the study is done), and the joint probability distribution of tropical cyclone parameters is constructed from regional historical data. 

We primarily focus on peak storm surges at $N=9,284$ spatial locations in the Cape Coral subregion of the study area, since Cape Coral is a study region in the FEMA Region IV's mission. Section~\ref{sec: nodes} of the Supplementary Material gives a brief description of coastal flood study and shows the full mesh of ADCIRC and the selected spatial locations. To design the experiment, we select 50 unique combinations of storm parameters $(\Delta P, R_p, V_f, \theta, B)$ based on the maximin Latin hypercube design (LHD) in the input domain $[30, 70]\times [16, 39] \times [3, 10] \times [15, 75] \times [0.9, 1.4]$. This parameter range corresponds to a core region of major surge hazards of interest in the FEMA coastal study \citep{FEMA2017}. For each combination of these 5 storm parameters, the landfall location $\Ell$ is repeated with one $R_p$ spacing along the coastline in Cape Coral. For each of the 5 storm parameters, the initial position of landfall location is randomly chosen, meaning that no two storms make landfall at the same location, and the number of landfalls for each of the 50 parameter combinations varies, with a higher number of smaller storms; this reflects the smaller spatial scale of storm surge for smaller storms, which necessitates higher sampling to capture the more localized response. The meteorological forcing in both the ADCIRC simulator and the ADCIRC + SWAN simulator is produced by a single group (Oceanweather, Inc.) using the work in \cite{Cardone2009}. The full simulation can be found from the Coastal Flood Modeling Database - Southwest Florida \citep{Asher2022}. 

In total, we obtained 226 inputs, meaning that on average, each of the 50 parameter combinations is used to generate 4.5 storms. These 226 inputs will be referred to as $\calX_0$. We randomly selected 60 inputs from the 226 inputs to run the ADCIRC + SWAN simulator, which will be referred to as $\calX_2$. Here 60 model runs are selected because the folklore for Gaussian process emulation is to use $10d$ model runs. Then we randomly selected 150 inputs from the the remaining 166 inputs to run the ADCIRC simulator, which will be referred to as $\calX_1^1$. To characterize the difference between the ADCIRC simulator and the ADCIRC + SWAN simulator, we also randomly chose 50 inputs from the 60 inputs to run the ADCIRC simulator, which will be referred to as $\calX_1^2$. Let $\calX_1:=\calX_1^1 \cup \calX_1^2$ be the collection of 200 inputs from $\calX_1^1$ and $\calX_1^2$. Notice that only 50 inputs in the ADCIRC + SWAN simulator are nested within the 200 inputs in the ADCIRC simulator.  

Figure~\ref{fig: storm surges at two input settings} shows peak surge elevations over $9,284$ spatial locations from the ADCIRC simulator and the ADCIRC + SWAN simulator at two different input settings $\bfx_1=(48.30,$ $20.48, 6.187, 62.28, 1.260, -82.08, 26.59)^\top$ and $\bfx_2=(68.85, 34.34, 8.778, 55.57, 1.066, -82.15,$ $26.69)^\top$. The output surfaces also have very different variations in different regions at each fidelity level. This indicates that a spatially-varying mean function or a spatially-varying variance function may help capture the spatial variations in the output space. The third column shows that PSEs have their maximum difference less than 0.3 meters between the low and high fidelity simulators, and also shows that the discrepancy has different spatial structures. At some specific regions such as the Fort Myers Beach (on the top right panel), very sharp changes can be detected. Physically, the changes in the surge elevation arise from differences in the spatial and temporal structures of the surge-only versus the wave response to storm forcing. For instance, in the top panel, the addition of wave-driven water level setup leads to greater overtopping (a situation when waves are higher than the dunes or structures they encounter) of coastal barrier islands, bringing greater water into the semi-protected bays in the southeastern portion of the figure. It can be difficult to detect these sorts of patterns without modeling wind-driven wave effects. In what follows, we develop a cokriging-based emulator to approximate the high-fidelity simulator by combining simulations from a limited number of high-fidelity runs and a larger number of low-fidelity simulation runs.

\begin{figure}
\begin{subfigure}{.333\textwidth}
  \centering
\makebox[\textwidth][c]{ \includegraphics[width=1.0\linewidth, height=0.12\textheight]{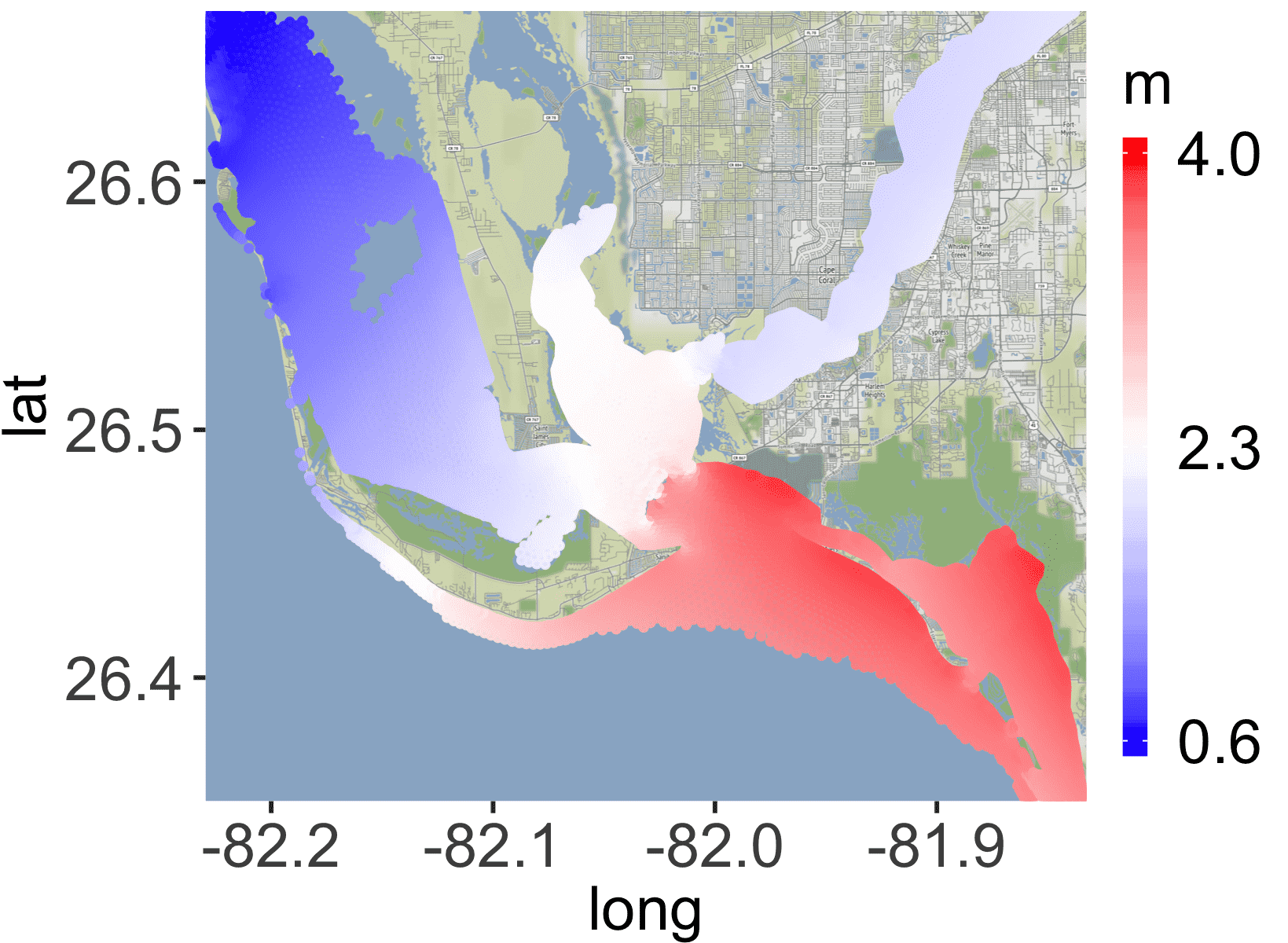}}
  \caption{ADCIRC: $\bfx_1$.}
\end{subfigure}%
\begin{subfigure}{.333\textwidth}
  \centering
  \includegraphics[width=1.0\linewidth, height=0.12\textheight]{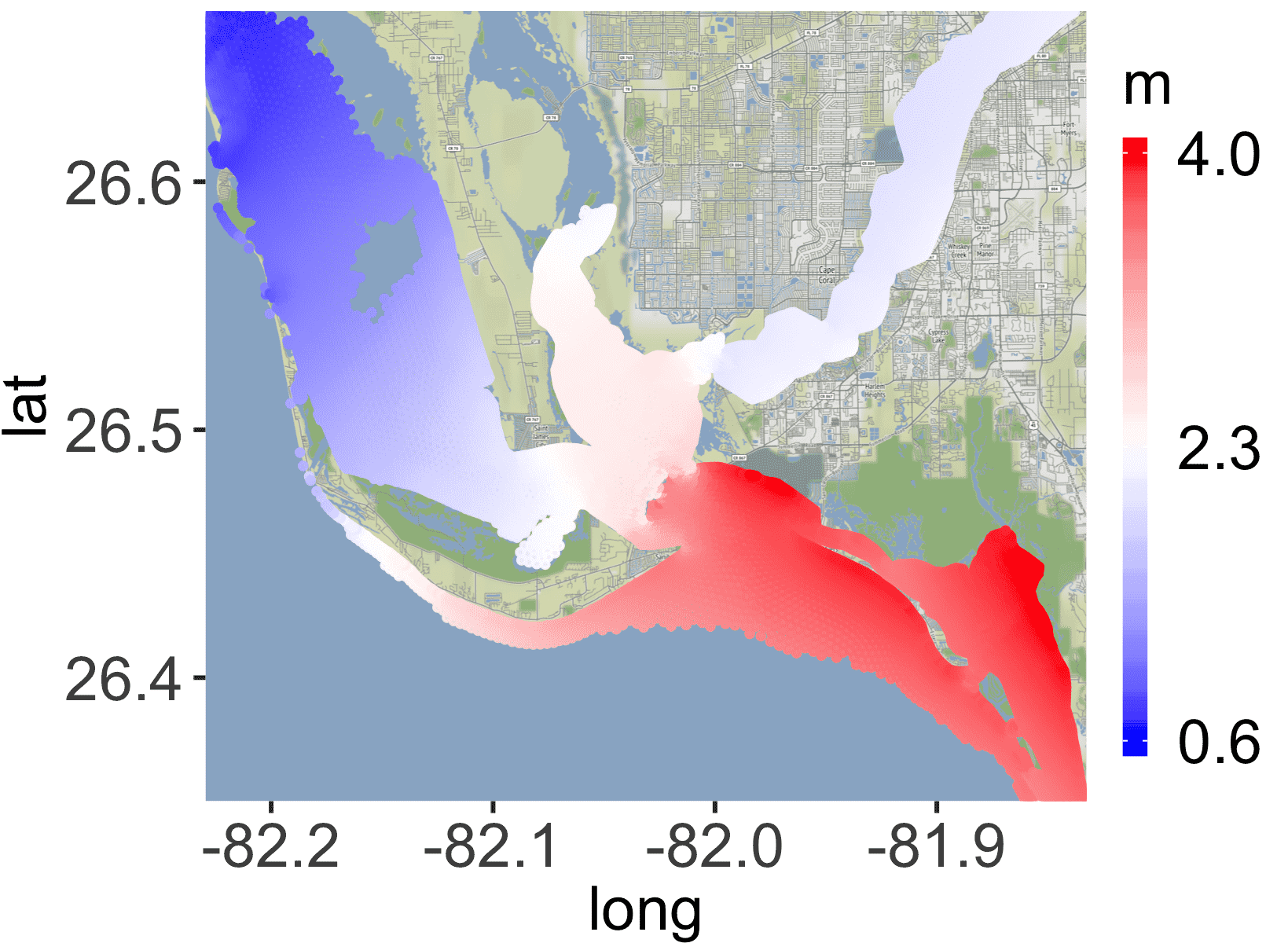}
  \caption{ADCIRC+SWAN: $\bfx_1$.}
\end{subfigure}%
\begin{subfigure}{.333\textwidth}
  \centering
  \includegraphics[width=1.0\linewidth, height=0.12\textheight]{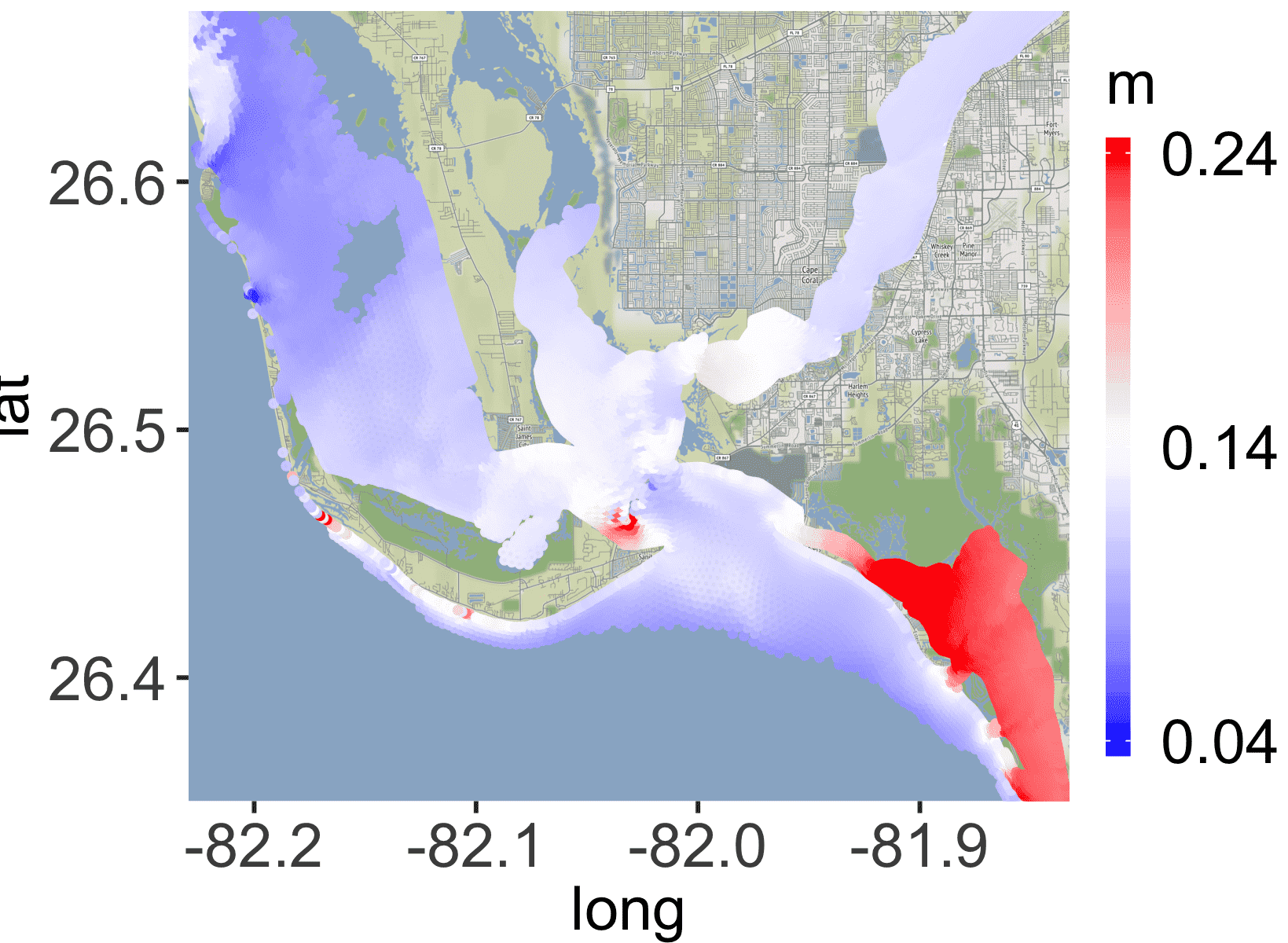}
  \caption{Difference at $\bfx_1$.}
\end{subfigure}
 
 \begin{subfigure}{.333\textwidth}
  \centering
\makebox[\textwidth][c]{ \includegraphics[width=1.0\linewidth, height=0.12\textheight]{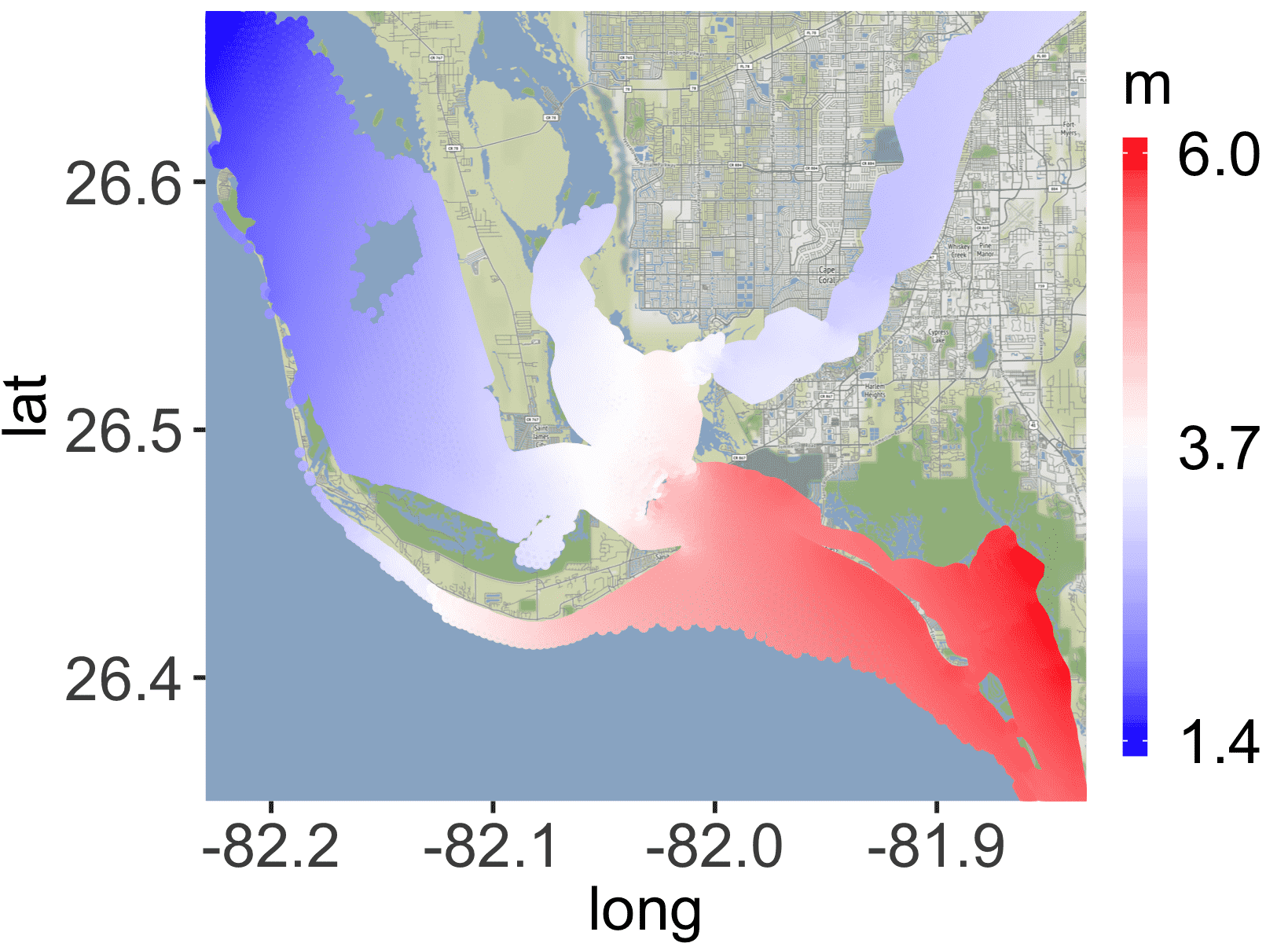}}
  \caption{ADCIRC: $\bfx_2$.}
\end{subfigure}%
\begin{subfigure}{.333\textwidth}
  \centering
  \includegraphics[width=1.0\linewidth, height=0.12\textheight]{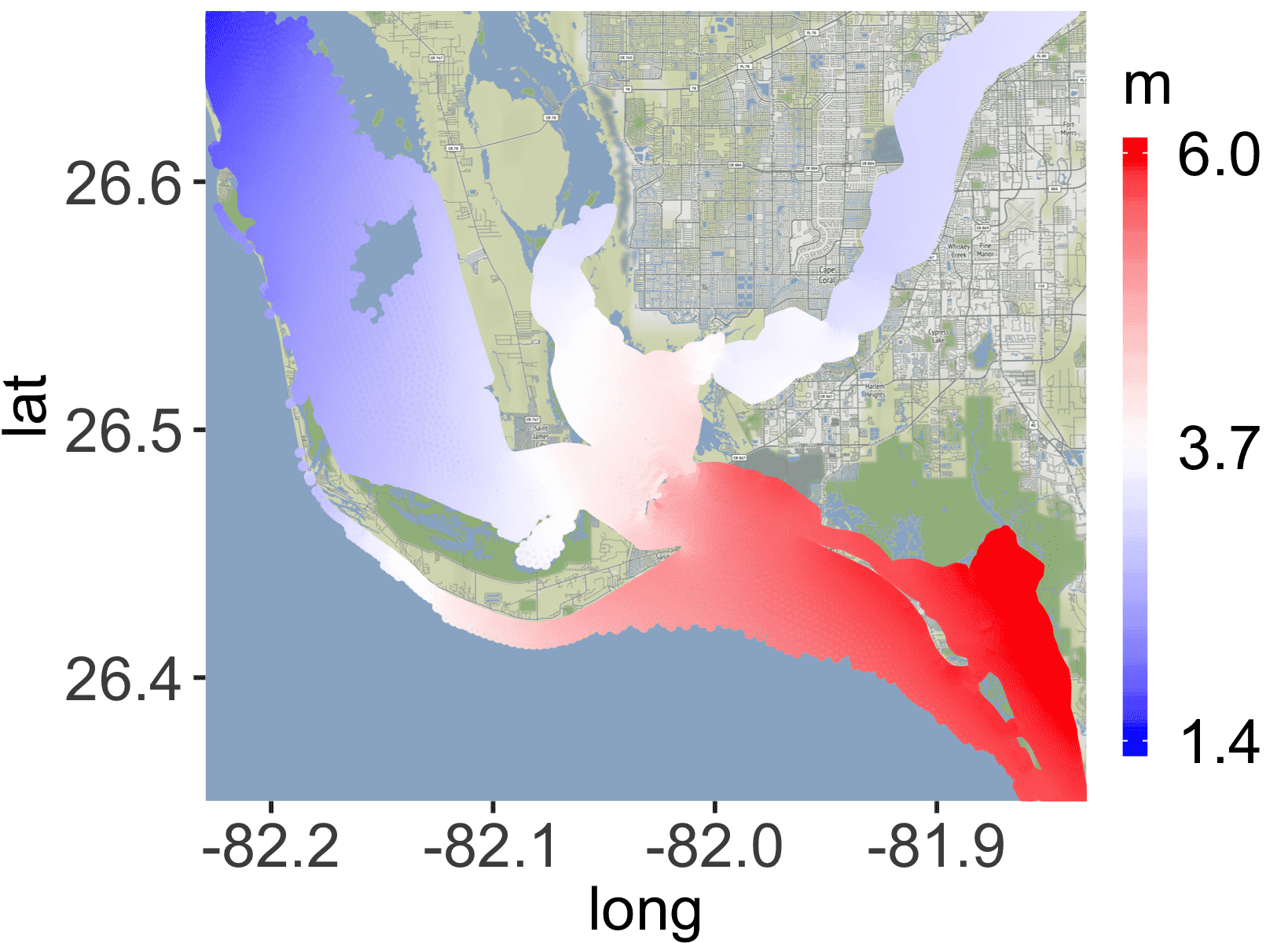}
  \caption{ADCIRC+SWAN: $\bfx_2$.}
\end{subfigure}%
\begin{subfigure}{.333\textwidth}
  \centering
  \includegraphics[width=1.0\linewidth, height=0.12\textheight]{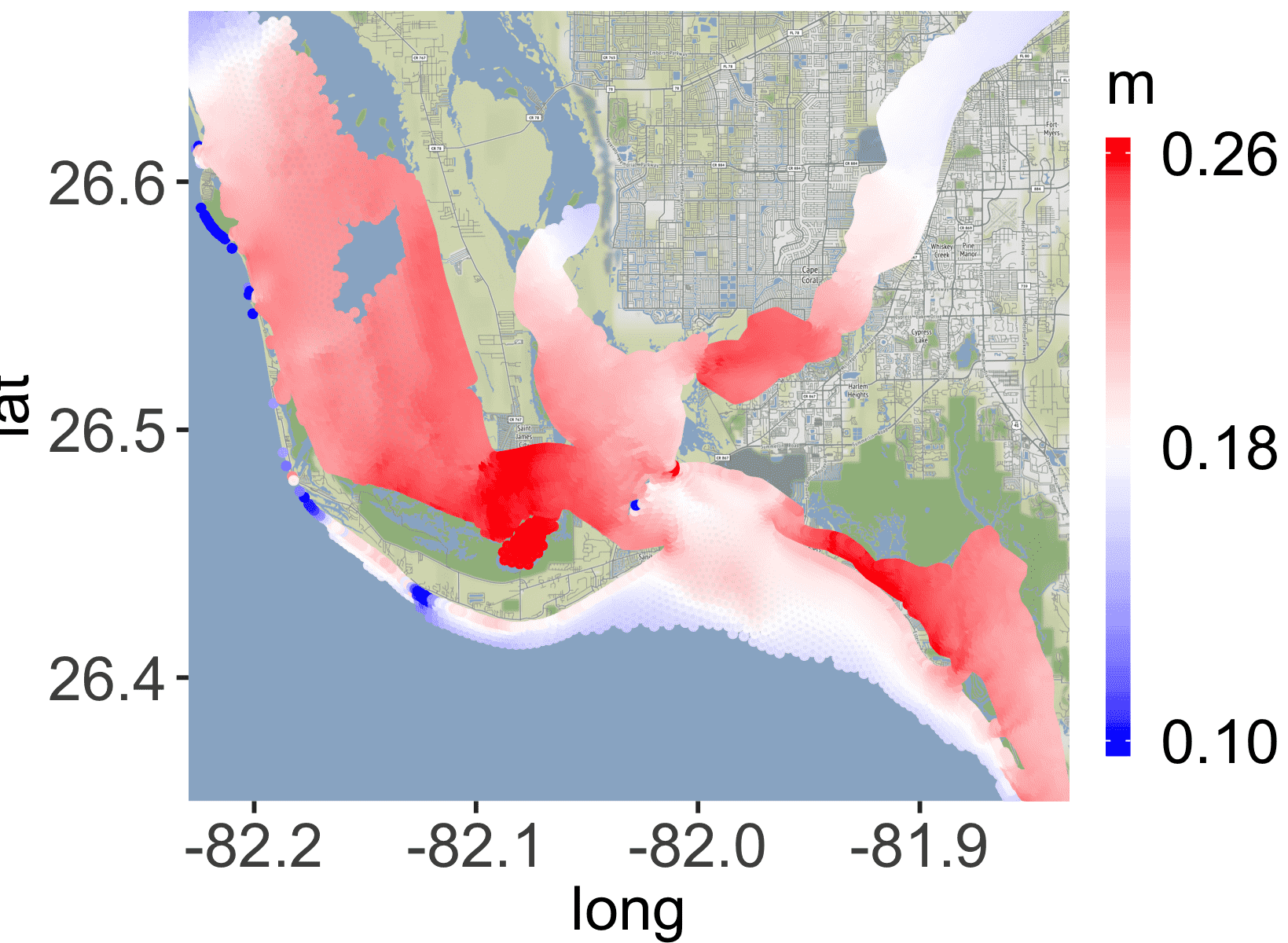}
  \caption{Difference at $\bfx_2$.}
\end{subfigure}

 \caption{Comparison of model runs from the ADCIRC simulator and the ADCIRC + SWAN simulator at two input settings. The first and second columns show the model runs from the low-fidelity simulator and the high-fidelity simulator. The third column shows the difference between the high-fidelity run and the low-fidelity run.}
\label{fig: storm surges at two input settings}
\end{figure}


\section{Multifidelity Computer Model Emulation} \label{sec: model formulation}

Section~\ref{sec: univariate model} gives a brief introduction of the general autoregressive cokriging framework, and Section~\ref{sec: multivariate model} presents the proposed methodology called \emph{parallel partial cokriging}.  

\subsection{Background on Autoregressive Cokriging Modeling}\label{sec: univariate model}

Assume that the computer model can be run at $s$ levels of sophistication corresponding to output functions $y_{1}(\cdot),\ldots,y_{s}(\cdot)$. The computer model associated to $y_{t}(\cdot)$ is assumed to be more accurate than the one associated to $y_{t-1}(\cdot)$ for $t=2,\ldots,s$. Let $\calX$ be a compact subset of $\mathbb{R}^{d}$, which is assumed to be the input space of computer model. Further assume that the computer model $y_{t}(\cdot)$ is run at a set of input values denoted by $\calX_{t}\subset\calX$ for $t=1,\ldots,s$, where $\calX_{t}$ contains $n_{t}$ input values. The autoregressive cokriging model at any input $\bfx\in\mathcal{X}_{t}$ is 
\begin{align} y_{t}(\bfx)=\gamma_{t-1}(\bfx)y_{t-1}(\bfx)+\delta_{t}(\bfx),\quad t=2,\ldots,s,\label{eqn: AR cokriging} 
\end{align}
where $\delta_{t}(\cdot)$ is the unknown location discrepancy representing the local adjustment from level $t-1$ to level $t$, and $\gamma_{t-1}(\bfx)$ is the scale discrepancy representing the scale change from level $t-1$ to level $t$ at input $\bfx$. This well-interpreted model is induced from the so called Markov assumption: no further information is gained about $y_{t}(\bfx)$ by observing $y_{t-1}(\bfx')$ for any other $\bfx\ne \bfx'$ \citep{Kennedy2000}.

To account for uncertainties in the unknown functions $y_{1}(\cdot)$ and $\delta_{t}(\cdot)$, we assign Gaussian processes
 \begin{align} \begin{split}y_{1}(\cdot)\mid\bfbeta_{1},\sigma_{1}^{2},\bfphi_{1} & \sim\mathcal{GP}(\bfh_{1}^{\top}(\cdot)\bfbeta_{1},\,\sigma_{1}^{2}r(\cdot,\cdot|\bfphi_{1})),\\ \delta_{t}(\cdot) & \sim\mathcal{GP}(\bfh_{t}^{\top}(\cdot)\bfbeta_{t},\,\sigma_{t}^{2}r(\cdot,\cdot|\bfphi_{t})), \end{split} \label{eqn: cokriging model} 
\end{align} 
for $t=2,\ldots,s$. $\bfh_{t}(\cdot)$ is a vector of basis functions and $\bfbeta_{t}$ is a vector of coefficients at fidelity level  $t$. In practice, the basis functions $\bfh_{t}(\cdot)$ along with $\bfw_{t}(\cdot)$ should be determined by exploratory data analysis following standard Gaussian process modeling procedure. $\gamma_{t-1}(\bfx)$ can be modeled as a basis function representation, i.e., $\gamma_{t-1}(\bfx)=\bfw_{t-1}(\bfx)^{\top}\bfxi_{t-1}$, where $\bfw_{t-1}(\bfx)$ is a vector of known basis functions and $\bfxi_{t-1}$ is a vector of unknown coefficients. Here, $r(\cdot,\cdot|\bfphi_{t})$ is a correlation function with correlation parameters $\bfphi_t:=(\phi_{t,1}, \ldots, \phi_{t,d})^\top$ at fidelity level $t$, where $d$ denotes the number of input parameters. Following the seminar work \citep{Sacks1989}, we use the product form of correlation structure to allow anisotropy in each input dimension, i.e., $r(\bfx,\bfx'\mid\bfphi_{t})=\prod_{i=1}^{d}r(x_{i},x_{i}'\mid\phi_{t,i})$, where each $r(x_{i},x_{i}'\mid\phi_{t,i})$ can be chosen as the Mat\'ern correlation function although other choice is available \citep[see][]{Ma2019cov}. The Mat\'ern correlation function is 
\begin{align*} r(u\mid\phi)=\frac{2^{1-\upsilon}}{\Gamma(\upsilon)}\left(\frac{\sqrt{2\upsilon}u}{\phi}\right)^{\upsilon}\mathcal{K}_{\upsilon}\left(\frac{\sqrt{2\upsilon}u}{\phi}\right), 
\end{align*}
where $u=|x_i-x_i'|$ is the Euclidean distance, $\mathcal{K}_{\upsilon}$ is the modified Bessel function of the second kind, and $\upsilon$ is the smoothness parameter that controls the mean-square differentiability of random processes. Following the standard practice of computer model emulation, $\nu=2.5$ is chosen because of its closed form expression and twice mean-square differentiability of random processes \citep[e.g.,][]{Ma2020OBayes}. 

One way to implement the autoregressive cokriging model defined in Equations~\eqref{eqn: AR cokriging} and \eqref{eqn: cokriging model} for the storm surge application is to treat spatial coordinates as an additional input parameter. This can avoid dealing with the high-dimensional output, but it requires dealing with large-scale computations in evaluating the likelihood function. In particular, we have $9,284\times200$ model output values for the low-fidelity simulator ADCIRC and $9,284\times60$ ones for the high-fidelity simulator ADCIRC + SWAN. Fitting the autoregressive cokriging model by treating spatial coordinates as an additional input would require about $1.4\times10^{19}=(9284\times 260)^{3}$ flops to evaluate the likelihood and about 46,252 gigabytes memory to store the covariance matrix with double precision. In addition, the computation of the predictive distribution would be infeasible in standard computers as large number of unknown parameters could not be analytically integrated out; this is because the experimental design in our application is not necessarily fully nested across levels of code as required in existing autorigressive cokriging implementations \citep[e.g.,][]{Kennedy2000,Qian2008,Gratiet2013}. Another challenge is to model nonstationarity in the output space, where the storm surges in Cape Coral show strong heterogeneous spatial dependence structures.

\subsection{The Parallel Partial Cokriging Emulator} \label{sec: multivariate model}

We propose the parallel partial cokriging emulator that couples ideas from the parallel partial Gaussian process \citep{Gu2016} with the cokriging model. This approach mitigates the aforementioned challenges in the storm surge application.

We consider that at each level the output functions and their additive discrepancy functions are modeled as multivariate GPs where each dimension corresponds to a spatial location. This induces a cokriging model with high-dimensional output due to the massive spatial locations available.  Two popular ideas for modeling multivariate output in computer models are the nonseparable covariance model between the input space and output space via basis-function representations for outputs \cite{Higdon2008} and the separable covariance model between input space and output space \citep{Conti2010}.  

However, these approaches have not been developed in the multifidelity setting. What follows is to discuss whether they can provide useful ideas to solve the storm surge application. Specifically, the basis-function representation approach \citep{Higdon2008} requires the output to be represented in terms of a few principal components. However, exploratory analysis suggests that this is not possible in the storm surge application, due to the large changes in the storm output from different input simulations. For the separable model \citep{Conti2010}, it is computationally infeasible due to massive output that the storm surge simulators generate for each simulation. To deal with the high-dimensional output in the separable model, \cite{Gu2016}  propose the PP Gaussian process emulator, which assumes conditionally independent Gaussian processes for each spatial coordinate with the same range parameter. This approach is able to borrow information across the data-poor input dimensions when the data are rich in the spatial dimension with linear computational cost in terms of  the number of simulator outputs $N$.  However the adoption of the PP ideas in the multifidelity setting with possibly non-nested designs is not straightforward. What follows is to introduce the proposed emulator methodology.

\subsubsection{The Parallel Partial Cokriging Model} \label{sec: PP cokriging model}

In the storm surge application, we have $s=2$ fidelity levels for computer models: ADCIRC and ADCIRC + SWAN. Each simulation of these computer models generates output values at the same $N=9,284$ spatial locations. Let $n_{t}$ denote the number of computer simulations at fidelity level $t$, where $n_1=200$ and $n_2=60$ in the application. Let $\bfy_{t,j}$ be a vector of output values over all inputs in $\calX_{t}$ at coordinate $j$ and fidelity level  $t$. Let $\bfy^{t,\mathscr{D}}:=[\bfy_{t,1}, \ldots, \bfy_{t,N}]$ be a $n_t\times N$ matrix of output values across all input values and all spatial locations at level $t$. Let $\bfy_{j}^{\mathscr{D}}:=(\bfy_{1,j}^{\top},\ldots,\bfy_{s,j}^{\top})^{\top}$ be a vector of output values at coordinate $j$ over all inputs in $\calX$ and all fidelity levels.

For each coordinate $j=1,...,N$ and fidelity level $t=2,\ldots,s$, we specify the cokriging model for any input $\bfx \in \mathcal{X}_t$ as  
\begin{align} \label{eqn: AR in multivariate model}
y_{t,j}(\bfx)=\gamma_{t-1,j}(\bfx)y_{t-1,j}(\bfx)+\delta_{t,j}(\bfx).
\end{align}
Here, for $y_{t-1,j}(\bfx)$ and $\delta_{t,j}(\bfx)$, we assign Gaussian processes priors:
\begin{align}
\begin{split}y_{1,j}(\cdot)\mid\bfbeta_{1,j},\sigma_{1,j}^{2},\bfphi_{1} & \sim\mathcal{GP}(\bfh_{1}^{\top}(\cdot)\bfbeta_{1,j},\,\sigma_{1,j}^{2}r(\cdot,\cdot|\bfphi_{1})),\\
\delta_{t,j}(\cdot) & \sim\mathcal{GP}(\bfh_{t}^{\top}(\cdot)\bfbeta_{t,j},\,\sigma_{t,j}^{2}r(\cdot,\cdot|\bfphi_{t})),
\end{split}
\label{eqn: model for coordinate j}
\end{align}
where $\bfh_t(\cdot)$ is a vector of common fixed basis functions across all $N$ spatial locations. In the storm surge application, these basis functions are assumed to be constant functions based on exploratory analysis. For $\gamma_{t-1,j}(\bfx)$, it is represented as a function of $\bfx$, but it will be assumed as unknown constants in the real application, i.e., $\gamma_{t-1,j}(\bfx) := \gamma_{t-1,j}$, since the physics knowledge suggests that the coupled computer model ADCIRC+SWAN tends to generate higher storm surges than ADCIRC due to the wave effects.
For each coordinate $j$, we assume different regression parameters $\bfbeta_j:=\{\bfbeta_{1,j},$ $\ldots,\bfbeta_{s,j}\}$, different variance parameters $\bfsigma^2_j:=\{\sigma_{1,j}^{2},\dots,$ $\sigma_{s,j}^{2}\}$, and different scale discrepancy parameters $\bfgamma_j:=\{\gamma_{1,j},\ldots,\gamma_{t-1,j}\}$. The correlation parameters $\bfphi_t$ are assumed to be the same across different spatial coordinates to simplify computations following \citep{Gu2016}. 

We call the proposed model defined in \eqref{eqn: AR in multivariate model} and \eqref{eqn: model for coordinate j} the \emph{parallel partial cokriging} model. The ``parallel partial'' reflects the fact that our model can be thought of as involving different  autoregressive cokriging models at each coordinate which share common input correlation parameters. As mentioned in Section~\ref{sec: univariate model}, likelihood-based inference requires that  the collection of input runs at each level is nested in order to have closed-form inference, i.e., $\calX_{t}\subset\calX_{t-1}$. In the next section, we present a {data-augmentation} technique to deal with possibly non-nested design so that statistical inference based on the PP cokriging model can be carried out. 

\subsubsection{Data Augmentation for Non-Nested Design} \label{sec: DA}
The specification of convenient priors facilitating the tractability of the  marginal likelihood and the analytic integration of the unknown parameters require the available experimental design to be fully hierarchical nested; i.e. $\mathcal{X}_{t+1}\subset \mathcal{X}_{t}$. This   restrictive requirement can be found by examining the likelihood that results from \eqref{eqn: AR in multivariate model} and \eqref{eqn: model for coordinate j} and it is inherited from the cokriging model. As this requirement is not satisfied in our application, to address this issue we propose a suitable data-augmentation remedy that imputes the data with missing values to create a fully nested design.

We replace the original input design $\calX_{t}$ by $\tilde{\calX}_{t}=\calX_{t}\cup\mathring{\calX}_{t}$ such that $\mathring{\calX}_{t}:=\calX_{(t+1):s}\setminus\calX_{t}$ represents a collection of missing inputs that have not been run by the simulator at fidelity level $t$ in order to form a nested design, where $\calX_{(t+1):s}:=\cup_{k=t+1}^{s}\calX_{k}$ represents the collection of observed inputs from fidelity level $t+1$ up to the highest fidelity level $s$. Let $\mathring{\bfy}_{t,j}$ be a vector of missing output values all inputs in $\mathring{\calX}_{t}$ at coordinate $j$ and fidelity level  $t$. In what follows, we use $\tilde{n}_t$ to denote the number of input values in the augmented set $\tilde{\calX}_t$. Let $\mathring{\bfy}_{j}^{\mathscr{D}}:=(\mathring{\bfy}_{1,j}^{\top},\ldots,\mathring{\bfy}_{s,j}^{\top})^{\top}$ be a vector of missing output values at coordinate $j$ over all fidelity levels, where $\mathring{\bfy}_{s,j}^{\top}$ is defined to be empty for notational convenience. Let $(\mathring{\bfy}_t)_{j=1}^N:=(\mathring{\bfy}_{t,1}^\top, \ldots, \mathring{\bfy}_{t,N}^\top)^\top$ be a vector of missing output at fidelity level  $t$ over all $N$ spatial coordinates. Let $\tilde{\bfy}_{t,j}:=(\bfy_{t,j}^{\top},\mathring{\bfy}_{t,j}^{\top})^{\top}$ be a vector of augmented  output over all inputs at fidelity level  $t$ at the $j$th spatial coordinate and $\tilde{\bfy}_{j}^{\mathscr{D}}:=(\bfy_{j}^{\mathscr{D}\top},\mathring{\bfy}_{j}^{\mathscr{D}\top})^{\top}$ be a vector of augmented output over all inputs at  coordinate $j$. Then the augmented sampling distribution at coordinate $j$
is 
\begin{align}
\begin{split}
L(\tilde{\bfy}_{j}^{\mathscr{D}}\mid\bfbeta_{j},\bfgamma_{j},\bfsigma_{j}^{2},\bfphi)& \propto\pi(\tilde{\bfy}_{1,j}\mid\bfbeta_{1,j},\sigma_{1,j}^{2},\bfphi_{1})
\prod_{t=2}^{s}\pi(\tilde{\bfy}_{t,j}\mid \tilde{\bfy}_{t-1,j}, \bfbeta_{t,j},\gamma_{t,j},\sigma_{t,j}^{2},\bfphi_{t}),
\end{split}
\end{align}
with 
\begin{align*}
    \pi(\tilde{\bfy}_{1,j}\mid\bfbeta_{1,j},\sigma_{1,j}^{2},\bfphi_{1}) &= \mathcal{N}(\widetilde{\bfH}_{1} \bfbeta_{1,j}, \sigma^2_{t,j} \tilde{\bfR}_1), \\ 
    \pi(\tilde{\bfy}_{t,j}\mid \tilde{\bfy}_{t-1,j}, \bfbeta_{t,j},\gamma_{t,j},\sigma_{t,j}^{2},\bfphi_{t}) &= \mathcal{N}( \widetilde{\bfH}_{t} \bfbeta_{t,j} + \widetilde{W}_{t-1,j} \gamma_{t-1,j}, \sigma^2_{t,j} \tilde{\bfR}_{t}),
\end{align*}
where $\widetilde{\bfH}_{t}:=\bfh_t(\tilde{\calX}_t)$, $\tilde{\bfR}_t:=r(\tilde{\calX}_t, \tilde{\calX}_t \mid \bfphi_t)$,   and $\widetilde{W}_{t-1,j}:=\bfy_{t-1,j}(\tilde{\calX}_t)$.  
The proposed augmentation allows the joint likelihood of the complete output values to be factorized into Gaussian likelihood kernels of smaller dimensionality, where we can specify conditional conjugate priors and break the training problem into smaller more tractable ones. 

In Section~\ref{sec: toy example} and Section~\ref{sec: Example with functional output} of the Supplementary Material, artificial examples are used to illustrate the performance of autoregressive cokriging with a \emph{non-nested design} with the proposed data-augmentation technique. These examples show that Incorporating information from a low-fidelity code can improve inferential results on the high-fidelity code as well as demonstrate the good performance of the PP cokriging emulator.

\subsubsection{Empirical Bayesian Inference via an MCEM Algorithm} \label{sec: parameter estimation}

Let $\tilde{\bfy}^{\mathscr{D}}:=(\tilde{\bfy}_{1}^{\mathscr{D}},\ldots,\tilde{\bfy}_{N}^{\mathscr{D}})$ be a $(\sum_{t=1}^sn_t)\times N$ matrix of augmented outputs over all fidelity levels and all spatial locations. We introduce the following notation: $\bfbeta:=(\bfbeta_{1}^{\top},\ldots,\bfbeta_{N}^{\top})^{\top}$, $\bfgamma:=(\bfgamma_{1}^{\top},\ldots,\bfgamma_{N}^{\top})^{\top}$, and $\bfsigma^{2}:=(\bfsigma_{1}^{2\top},\ldots,\bfsigma_{N}^{2\top})^{\top}$. The overall augmented sampling distribution across all $N$ spatial locations is the product of each augmented sampling distribution 
\begin{eqnarray}
L(\tilde{\bfy}^{\mathscr{D}}\mid\bfbeta,\bfgamma,\bfsigma^{2},\bfphi)=\prod_{j=1}^{N}L(\tilde{\bfy}_{j}^{\mathscr{D}}\mid\bfbeta_{j},\bfgamma_{j},\bfsigma_{j}^{2},\bfphi).
\end{eqnarray}

We specify the following a priori model  for the unknown  parameters 
\begin{align} \label{eqn: prior model in multivariate model}
    \begin{split}\pi(\bfbeta,\bfgamma,\bfsigma^{2},\bfphi) 
  &=\pi(\bfphi)\prod_{j=1}^{N}\left\{ \pi(\bfbeta_{s,j}, \sigma_{s,j}^2) \prod_{t=1}^{s-1}\pi(\bfbeta_{t,j},\gamma_{t,j},\sigma_{t,j}^{2})\right\} ,
\end{split}
\end{align}
where  independent Jeffreys priors can be assigned on $\bfbeta_{t,j}, \gamma_{t-1,j}, \sigma_{t,j}^{2}$: i.e., $\pi(\bfbeta_{t,j}, \gamma_{t-1,j}, \sigma_{t,j}^{2})$  $\propto \frac{1}{\sigma^2_{t,j}}$ for $t=2, \ldots, s$, and $\pi(\bfbeta_{1,j}, \sigma^2_{1,j}) \propto \frac{1}{\sigma^2_{1,j}}$ at each coordinate $j$. For each level $t$, we assign a jointly robust prior \citep{Gu2019} on $\bfphi_{t}$ which is  a proper prior and has desirable properties for Gaussian process emulation. 

After integrating out model parameters $\{\bfbeta,\bfgamma,\bfsigma^{2}\}$,
the conditional distribution of $\tilde{\bfy}^{\mathscr{D}}$ given $\bfphi$ is 
\begin{align} \label{eqn: aug marginal dist}
\begin{split}
\pi(\tilde{\bfy}^{\mathscr{D}}\mid\bfphi) 
 & =\prod_{j=1}^{N}\pi(\tilde{\bfy}_{1,j}\mid\bfphi_{1})\prod_{t=2}^{s}\pi(\tilde{\bfy}_{t,j}\mid\bfphi_{t},\tilde{\bfy}_{t-1,j}),
\end{split}
\end{align}
where each conditional distribution is given in Section~\ref{sec: MCEM} of the Supplementary Material. The augmented posterior distribution $\pi(\bfphi \mid \tilde{\bfy}^{\mathscr{D}})$ can be obtained via Bayes' theorem: $\pi(\bfphi \mid \tilde{\bfy}^{\mathscr{D}}) \propto \pi(\tilde{\bfy}^{\mathscr{D}}\mid\bfphi) \times \pi (\bfphi)$, where $\pi (\bfphi)$ is a proper prior.

To estimate $\bfphi$, we adopt an empirical Bayesian approach which  maximizes the integrated posterior $\pi(\bfphi\mid \tilde{\bfy}^{\mathscr{D}})$, since empirical Bayes approaches provide faster computational results compared to the fully Bayesian ones which often require Markov chain Monte Carlo methods. As direct maximization of $\pi(\bfphi\mid \tilde{\bfy}^{\mathscr{D}})$ is impossible due to the intractable form of $\pi(\bfphi\mid \tilde{\bfy}^{\mathscr{D}})$, we introduce a Monte Carlo expectation-maximization (MCEM) algorithm \citep{Wei1990, Dempster1977} to tackle this challenge. The MCEM algorithm consists of two important steps: in E-step, the $Q$-function is defined as a function of certain conditional expectation that is approximated by Monte Carlo methods;  in M-step, maximization of this function is performed with respect to the unknown parameters; for detailed development of this MCEM algorithm; see Algorithm~\ref{alg: MCEM in mult} of the Supplementary Material.

\subsubsection{Prediction} \label{sec: prediction}

For any new input $\bfx_0\in {\calX}$, the goal is to make prediction for $\{y_{s,j}(\bfx_0), j=1, \ldots, N\}$ based upon the data $\bfy^{\mathscr{D}}$. With the prior model~\eqref{eqn: prior model in multivariate model}, the predictive distribution of interest is $y_{s,j}(\bfx_0 \mid \bfy^{\mathscr{D}}, \bfphi)$ for $j=1, \ldots, N$. 

In what follows, we derive a new approach to predicting $y_{s,j}(\bfx_0)$ and termed it as a \emph{sequential prediction} approach. The idea is to add the new input $\bfx_0$ to each collection of missing inputs $\mathring{\calX}_t$ such that a hierarchically nested design can be obtained. To fix the notation, we define $\mathring{\calX}_t^0: = \mathring{\calX}_t \cup \{\bfx_0\}$ and $\tilde{\calX}_t^0 := {\calX}_t \cup \mathring{\calX}_t^0$. Hence the collection of inputs $\{ \tilde{\calX}_t^0: t=1, \ldots, s\}$ also forms a nested design with $\tilde{\calX}_t^0 \subset \tilde{\calX}_{t-1}^0$. Let $\bfy(\bfx_0):=(\bfy_1(\bfx_0)^\top, \ldots, \bfy_s(\bfx_0)^\top)^\top$ with $\bfy_j(\bfx_0):=(y_{1,j}(\bfx_0), \ldots, y_{s,j}(\bfx_0))^\top$. The predictive distribution of $\bfy(\bfx_0)$ given $\tilde{\bfy}^{\mathscr{D}}$ and $\bfphi$ is the product of $N$ independent distributions with 
\begin{align*}
   \pi(\bfy(\bfx_0) \mid \tilde{\bfy}^{\mathscr{D}}, \bfphi) = \prod_{j=1}^N \pi(\bfy_j(\bfx_0)\mid \tilde{\bfy}_j^{\mathscr{D}}, \bfphi).
\end{align*}

The following result gives the predictive distribution at each spatial coordinate. Its proof follows from standard kriging theory.

\begin{proposition}[\textbf{Sequential Prediction}] \label{thm: predict}
Given the PP cokriging model and the non-informative priors~\eqref{eqn: prior model in multivariate model}, the predictive distribution across all fidelity levels at spatial coordinate $j$ for $j=1, \ldots, N$ is 
\begin{align} \label{eqn: predictive distribution in multivariate model}
\begin{split}
        \pi(\bfy_j(\bfx_0) \mid \tilde{\bfy}_j^{\mathscr{D}}, \bfphi) &= \pi(y_{1,j}(\bfx_0) \mid \tilde{\bfy}_{1,j}, \bfphi_1) \prod_{t=2}^{s-1} \pi(y_{t,j}(\bfx_0) \mid \tilde{\bfy}_{t,j}, \tilde{\bfy}_{t-1,j},  \\
        &\quad\quad y_{t-1,j}(\bfx_0),  \bfphi_t) \times \pi(y_{s,j}(\bfx_0) \mid y_{s-1,j}(\bfx_0), \bfy_{s,j}, \bfphi_s).
\end{split}
\end{align}
The conditional distributions on the right-hand side of \eqref{eqn: predictive distribution in multivariate model} are Student-$t$ distributions with degrees of freedom, location, and scale parameters given by
\begin{align*}
\begin{split}\nu_{t,j} & :=n_{t} - q_t,\\
\bfmu_{t,j} & :=\bfT(\bfx_0)\hat{\bfb}_{t,j} +\bfr_t^\top(\bfx_0)\tilde{\bfR}_t^{-1}(\tilde{\bfy}_{t,j}-\tilde{\bfT}_{t,j}\hat{\bfb}_{t,j}),\\
V_{t,j} & :=\frac{S^2(\bfphi_t, \tilde{\bfy}_{t,j})}{n_{t}-q_t}c_{t,j}(\bfx_0),
\end{split}
\end{align*}
with $\hat{\bfb}_{t,j} :=(\tilde{\bfT}_{t,j}^\top \tilde{\bfR}_t^{-1} \tilde{\bfT}_{t,j})^{-1} \tilde{\bfT}_{t,j}^\top \tilde{\bfR}_t^{-1}\tilde{\bfy}_{t,j}$, $S^2(\bfphi_t, \tilde{\bfy}_{t,j})  := (\tilde{\bfy}_{t,j} - \tilde{\bfT}_{t,j}\hat{\bfb}_{t,j})^\top \tilde{\bfR}_t^{-1} (\tilde{\bfy}_{t,j} - \tilde{\bfT}_{t,j}\hat{\bfb}_{t,j})$,
\begin{align*}
\begin{split}
c_{t,j}(\bfx_0) & :=r(\bfx_0, \bfx_0|\bfphi_t) - \bfr_t^\top(\bfx_0)\tilde{\bfR}_t^{-1}\bfr_t(\bfx_0) \\
&\quad+ [{\bfT}_{t,j}(\bfx_0)-\tilde{\bfT}_{t,j}^{\top}\tilde{\bfR}_{t}^{-1}{\bfr}_t(\bfx_0)]^{\top}(\tilde{\bfT}_{t,j}^{\top}\tilde{\bfR}_{t}^{-1}\tilde{\bfT}_{t,j})^{-1}[{\bfT}_{t,j}(\bfx_0) -\tilde{\bfT}_{t,j}^{\top}\tilde{\bfR}_{t}^{-1}{\bfr}_t(\bfx_0)],
\end{split}
\end{align*}
where $\bfr_t(\bfx_0):=r(\tilde{\calX}_t, \bfx_0\mid \bfphi_t)$ and ${\bfT}_{t,j}(\bfx_0)=[{\bfh}_t(\bfx_0), {y}_{t-1,j}(\bfx_0)]$.
\end{proposition}

Proposition~\ref{thm: predict} shows that a random sample from the predictive distribution can be sequentially drawn from a collection of conditional distributions in an efficient manner, since the total computational cost required for such a simulation is $O(\sum_{t=1}^s \tilde{n}_t^3)$ at each spatial coordinate. As the correlation matrix is the same across all spatial locations at each fidelity level, the total computational cost to obtain one single random sample from the predictive distribution across all spatial locations is $O(\sum_{t=1}^s \tilde{n}_t^3 + N \sum_{t=1}^s \tilde{n}_t^2 )$, which is linear in $N$ when $\sum_{t=1}^s \tilde{n}_t^2 \ll N$. Notice that a sample from $\pi(\bfy_{s}(\bfx_0) \mid \bfy^{\mathscr{D}}, \bfphi)$ can be obtained via the composition sample technique based on $\pi(\bfy_{s}(\bfx_0) \mid \bfy^{\mathscr{D}}, \bfphi) = \int \pi(\bfy_{s}(\bfx_0) \mid \tilde{\bfy}^{\mathscr{D}}, \bfphi) \pi(\mathring{\bfy}^{\mathscr{D}} \mid {\bfy}^{\mathscr{D}}, \bfphi)\,  d \mathring{\bfy}^{\mathscr{D}}$, with the missing data $\mathring{\bfy}^{\mathscr{D}}$ being generated from the distribution  $\pi(\mathring{\bfy}^{\mathscr{D}} \mid {\bfy}^{\mathscr{D}}, \bfphi)$ given in Section~\ref{app: distributions} of the Supplementary Material. In practice, $\bfphi$ needs to be replaced with its maximum a posteriori estimate obtained via the MCEM in Algorithm~\ref{alg: MCEM in mult} of the Supplementary Material. {Vanilla Monte Carlo approximation is used to compute the predictive mean and predictive variance of the predictive distribution $\pi(\bfy_{s}(\bfx_0) \mid \bfy^{\mathscr{D}}, \bfphi)$ when the design is not nested.} For a nested design, closed-form expressions for the posterior predictive mean and posterior variance across all fidelity levels are given in Lemma~\ref{thm: prediction in nested design} of the Supplementary Material. 

In Section~\ref{app: pred} of the Supplementary Material, we derive the \emph{one-step prediction} formula based on the idea in \cite{Kennedy2000} and \cite{Gratiet2013} when the design is nested. The sequential prediction formula in Proposition~\ref{thm: predict} has several advantages over the one-step prediction formula in Section~\ref{app: pred} of the Supplementary Material. First, the high-dimentionality of simulator output makes the one-step prediction formula computationally infeasible in the storm surge application, since this sequential prediction formula has $O(N(\sum_{t=1}^s \tilde{n}_t^3))$ computational cost. Second, to obtain predictive distribution $\pi(y_{s,j}(\bfx_0)\mid \bfy^{\mathscr{D}}, \bfphi)$, model parameters $\{\bfgamma, \bfsigma^2 \}$ have to be numerically integrated out in the one-step prediction formula. Thus, Monte Carlo approximation is required to take account of uncertainty in both model parameters $\{ \bfgamma, \bfsigma^2 \}$ and missing data $\mathring{\bfy}^{\mathscr{D}}$. This will even hinder the practicality of the one-step prediction formula for large number of spatial locations. In contrast, the sequential prediction formula explicitly integrate model parameters $\{ \bfbeta, \bfgamma, \bfsigma^2 \}$ without relying on Monte Carlo approximations.

\subsubsection{Computational Cost} \label{sec: cost}
The PP cokriging model can allow efficient computations in output space due to the following reasons. In parameter estimation, each iteration of the MCEM algorithm in Algorithm~\ref{alg: MCEM in mult} of the Supplementary Material requires the computation of so-called $Q$-function in E-step of MCEM and its numerical optimization with respect to correlation parameters $\bfphi_t$ at each level of code. The evaluation of $Q$-function requires matrix inversions and matrix multiplication of size $\tilde{n}_t\times \tilde{n}_t$. Such an evaluation requires $O(MN \tilde{n}_t^2 + \tilde{n}_t^3)$ computational cost across $N$ spatial locations and $M$ Monte Carlo samples. If the numerical optimization requires $k$ evaluations of $Q$-function to find the optimal value, the overall computational cost in each iteration of the MCEM algorithm is $O(kMN \sum_{t=1}^s $ $\tilde{n}_t^2 + k \sum_{t=1}^s\tilde{n}_t^3)$ without any parallelization. Notice that parallelization across $t$ is possible according to Algorithm~\ref{alg: MCEM in mult} of the Supplementary Material. This is a one-time computational cost.  In the predictive distribution~\eqref{eqn: predictive distribution in multivariate model}, each conditional distribution requires matrix inversions and matrix multiplication of size $\tilde{n}_t\times \tilde{n}_t$. This requires $O(\tilde{n}^3_t $\ $+ N\tilde{n}_t^2)$ computational cost. One random sample generated from the predictive distribution at one new input value requires  $O(\sum_{t=1}^s \tilde{n}^3_t+ N\sum_{t=1}^s \tilde{n}_t^2)$ computational cost. As $\tilde{n}_t$ is typically small (a few hundreds at most) in many real applications, the computational cost in prediction is linear in the number of spatial locations, $N$. This indicates the scalability of the proposed method to handle high-dimensional output for multifidelity computer models.

\subsection{Near Equivalence of PP Cokriging and Separable Cokriging} 

The PP cokriging emulator turns out to have nearly the same marginal predictive distributions at each spatial location as a \emph{separable autoregressive cokriging} model, where the separability refers to the fact that unknown spatial correlation matrices $\bfSigma:=\{\bfSigma_t: t=1, \ldots, s\}$ on the spatial domain are assumed for the Gaussian process that approximates the level 1 code and the Gaussian processes in the location-scale discrepancy function. This result is an analogy of Theorem 6.1 of \cite{Gu2016} for autoregressive cokriging models. In what follows, we will assume a nested design. This leads to the following matrix-normal distribution for the separable autoregressive cokriging with $s$ levels:
\begin{align} \label{eqn: MN likelihood}
\begin{split}
L({\bfy}^{\mathscr{D}} \mid \bfB, \Gamma, \bfSigma, \bfphi) &= \mathcal{MN}_{n_1,N} (\bfy_1\mid \bfH_1 \bfB_1, \bfR_1, \bfSigma_1) \\
&\quad \times \prod_{t=2}^s \mathcal{MN}_{n_t,N} (\bfy_t\mid \bfH_t \bfB_t + W_{t-1} \Gamma_{t-1}, \bfR_t, \bfSigma_t),
\end{split}
\end{align}
where $\mathcal{MN}_{n_t, N}(\cdot, \cdot, \cdot)$ is a $n_t\times N$ matrix normal distribution.  $\bfB:=\{\bfB_1, \ldots, \bfB_s\}$ with $\bfB_t:=[\bfbeta_{t,1}, \ldots, \bfbeta_{t,N}]$ being a $q_t \times N$ matrix of unknown mean parameters. $W_{t-1}:=[y_{t-1,1}(\mathcal{X}_t),$ $\ldots, y_{t-1,N}(\mathcal{X}_t)]$ is an $n_t\times N$ matrix. $\bfrho_{t-1}:=(\gamma_{t-1,1}, \ldots, \gamma_{t-1,N})^\top$ is an $N\times 1$ vector. $\Gamma_{t-1}:=\diag\{\bfrho_{t-1}\}$ is an $N\times N$ diagonal matrix. 

We can show that  under the constant prior on the location parameters and any prior on spatial correlation matrices $\{\bfSigma_t: t=1, \ldots, s\}$, the resulting predictive mean in the separable autoregressive cokriging model is simply the predictive mean in the  PP cokriging model, and the resutling predictive variance in the separable autoregressive cokriging model is almost equal to the predictive variance in the PP cokriging model. Thus, if only the mean and marginal predictive variance are concerned, there is no need to introduce spatial correlation structures in the output space. This fact is established in the next theorem with its proof given in Section~\ref{sec: proof in SEP cokriging} of the Supplementary Material.
\begin{theorem} \label{thm: SEP}
For a separable autoregressive cokriging model with likelihood in~\eqref{eqn: MN likelihood}, given the objective prior 
\begin{align} \label{eqn: constant prior in SEP cokriging}
\pi(\bfB_1, \ldots, \bfB_s, \Gamma_1, \ldots, \Gamma_{s-1} \mid \bfSigma_1, \ldots, \bfSigma_s, \bfphi) \propto 1
\end{align}
for the parameters of mean functions and scale discrepancy functions, the following hold:
\begin{itemize}
\item[1.] The posterior predictive mean at level  $t$ in the separable autoregressive cokriging emulator, at an unobserved input $\bfx_0$ and at spatial coordinate $j$, is identical to the PP cokriging emulator mean. 
\item[2.] The posterior predictive variance at level $t$ in the separable autoregressive cokriging emulator, at an unobserved input $\bfx_0$ and at spatial coordinate $j$, depends on $\bfSigma_t$ through the posterior mean of the $j$th diagonal term $E[\bfSigma_t^{jj} \mid \bfy^{\mathscr{D}}, \bfphi_t] $; it is  identical to the PP cokriging emulator variance, if $E[\bfSigma_t^{jj} \mid \bfy^{\mathscr{D}}, \bfphi_t] = (n_t-q_t)\hat{\sigma}^2_{t,j}/(n_t-q_t-2)$, where $\hat{\sigma}^2_{t,j}$, defined in Lemma~\ref{thm: prediction in nested design} of the Supplementary Material, is an estimator of $\sigma^2_{t,j}$ under nested design. 
\end{itemize}
\end{theorem}

Theorem~\ref{thm: SEP} indicates that the PP cokriging emulator and the separable autoregressive cokriging emulator can have the same predictive mean and marginal predictive variance under nested design when $n_t-q_t$ is large as in practice, since the new posterior expectation of $\bfSigma_{t}^{jj}=\sigma^2_{t,j}$ will be approximately the same as $\hat{\sigma}^2_{t,j}$. Moreover, when the scale discrepancy function $\gamma_{t,j}$ is fixed at one. It is easy to check that the results in Theorem~\ref{thm: SEP} still hold. When the design is not nested, it can  be readily checked that PP cokriging  possesses the properties as in the nested design with the same proof except for notational difference by introducing the missing data. That is, the fact that missing data are available only affects how computation is carried out and does not alter the properties of the PP cokriging emulator. Thus, in practice, the predictive mean and predictive variance in PP cokriging can be used in practice as long as posterior draws from the  predictive distribution over the spatial domain is not concerned. In addition, risk assessment of storm surges in the FEMA report \citep{FEMA2017} only requires computation of the annual exceedance probability at each location, which is a function of predictive mean and predictive variance. 

\section{Analysis of Storm Surge Simulations} \label{sec: application}
In this section, the PP cokriging emulator is used to analyze high-dimensional output from the ADCIRC simulator and the ADCIRC + SWAN simulator.

The analysis of emulation results and numerical comparison is presented to demonstrate the advantage of the parallel partial cokriging model with high-dimensional output. The proposed PP cokriging methodology is implemented in the \textsf{R} package \texttt{ARCokrig} \citep{Ma2019ARCokrig}. The PP cokriging model is trained on 200 inputs from the ADCIRC simulator and 60 inputs from the ADCIRC + SWAN simulator, where 50 inputs from the second fidelity level are nested within the first fidelity level. With such model runs, the proposed method can be applied readily. To measure predictive performance, we run the ADCIRC + SWAN simulator at 166 inputs from the original 226 inputs after excluding 60 inputs as described in Section~\ref{sec: model simulation}.

Moreover, we train the PP kriging emulator via the \textsf{R} package \texttt{RobustGaSP} with the same 60 high-fidelity runs used in the PP cokriging emulator. As the landfall location is along the coastline, we define a distance measure $d_{\Ell}$ to replace the actual longitude and latitude coordinate of the landfall location. Specifically, we first choose a reference location $\Ell_0$ to be the landfall location that is in the most northwest direction along the coastline. Then for any landfall location $\Ell$, $d_{\Ell}$ is defined as the spherical distance between $\Ell$ and $\Ell_0$. As the coastline is unique, the landfall location determines the distance measure $d_{\Ell}$ and vice versa. In the implementation of the PP kriging emulator and the PP cokriging emulator, the input variables are $\Delta P$, $R_p$, $V_f$, $\theta$, $B$, $d_{\Ell}$. Evaluation of predictive performance is based on root-mean-squared-prediction errors (RMSPE), coverage probability of the 95\% equal-tail credible interval (CVG(95\%)), average length of the 95\% equal-tail credible interval (ALCI(95\%)), and continuous rank probability score \citep[CRPS;][]{Gneiting2007}. We also compute the Nash-Sutcliffe model efficiency coefficient (NSME): 
$$\text{NSME}:= 1 - \frac{\sum_{j=1}^N \sum_{\bfx \in A}\{m_j(\bfx) - y_{2,j}(\bfx)\}^2}{\sum_{j=1}^N \sum_{\bfx \in A} \{m_j(\bfx)-\bar{y}_{2,j}\}^2},$$ 
where $m_j(\bfx)$ is the value to predict the high-fidelity simulator $y_{2,j}(\cdot)$ at input $\bfx$ and $j$-th spatial coordinate and $\bar{y}_{2,j}: = \sum_{\bfx \in \calX_2} y_{2,j}(\bfx)/n_2$ is the average of code $y_{2,j}(\cdot)$ at $j$-th spatial coordinate. The NSME computes the residual variance with the total variance, which has a similar meaning as the coefficient of determination. The closer NSME is to 1, the more accurate the model is. If the ADCIRC simulator is used to predict the ADCIRC + SWAN simulator at these 166 new inputs, NSME is -1.089, which indicates that the mean of the training data in the high-fidelity simulator is a much better predictor than the low-fidelity simulator at these inputs.

\subsection{Emulation Accuracy} \label{sec: emulation result for storm surge}

 In the PP cokriging model, we include constant basis functions $\bfh_t(\cdot)$ according to exploratory data analysis. The scale discrepancy function $\gamma_{t-1,j}$ is assumed to be an unknown constant that does not depend on input. This assumption still allows the scale discrepancy function to vary across different fidelity levels and spatial coordinates. For parameter estimation, the MCEM algorithm is initialized with multiple starting values and took about 5 hours to achieve convergence for a pre-specified tolerance on a 2-core Macbook Pro with 8 GB random access memory. The predictive mean and predictive variance is approximated by 30 random draws from the distribution \eqref{eqn: predictive distribution in multivariate model}. Negligible improvement is seen from increasing the number of draws. The estimated range parameters show that the peak surge elevation is highly dependent on the inputs: central pressure deficit ($\Delta P$), Holland's B parameter ($B$), since these two inputs have relatively large range parameters compared to their input ranges in the training sets. The small impact of the landfall location ($\Ell$) is due to our focus on a small coastal region in Cape Coral.

A direct approach to building an emulator for the high-fidelity model run is to use the PP kriging emulator trained only with the high-fidelity simulations. As an additional comparison, we also include the prediction results based on the PP kriging emulator using the 200 low-fidelity runs only. The results in Table~\ref{table: CV for cokriging} show that the PP cokriging emulator gives best prediction results all the three emulators, since the PP cokriging model gives smallest RMSPE and CRPS and largest CVG and NSME. It is also interesting to see that PP cokriging is able to give largest CVG with a modest ALCI. As a baseline, the low-fidelity simulator is directly used to predict the high-fidelity simulator over these 166 new inputs, resulting in the RMSPE at 0.132, which is similar to the RMSPE obtained using the PP cokriging emulator trained on the low-fidelity runs only. The PP kriging emulator trained on high-fidelity runs only gives much better CVG and CRPS than it does when trained on the low-fidelity runs only, while PP kriging emulator trained on the high-fidelity runs only gives larger RMSPE than it does when trained on the low-fidelity runs only. This indicates that the folklore of using 10$d$ runs for emulating the high-fidelity simulator does not perform satisfactorily and that the low-fidelity simulator is a biased surrogate model for the high-fidelity simulator in the storm surge application. Figure~\ref{fig: prediction versus held-out} indicates that the PP cokriging emulator performs much better than the PP kriging emulator  at the input setting $\bfx_1$, since its predicted PSE are scattered around the 45-degree straight line. In contrast, the PP kriging emulator has a comparably worse performance. All these numerical and visual results confirm the advantage of using the PP cokriging emulator in the storm surge application by effectively borrowing information across different fidelity levels. 
 

\begin{figure}
\begin{subfigure}{0.5\textwidth}
  \centering
\makebox[\textwidth][c]{ \includegraphics[width=1.0\linewidth, height=0.18\textheight]{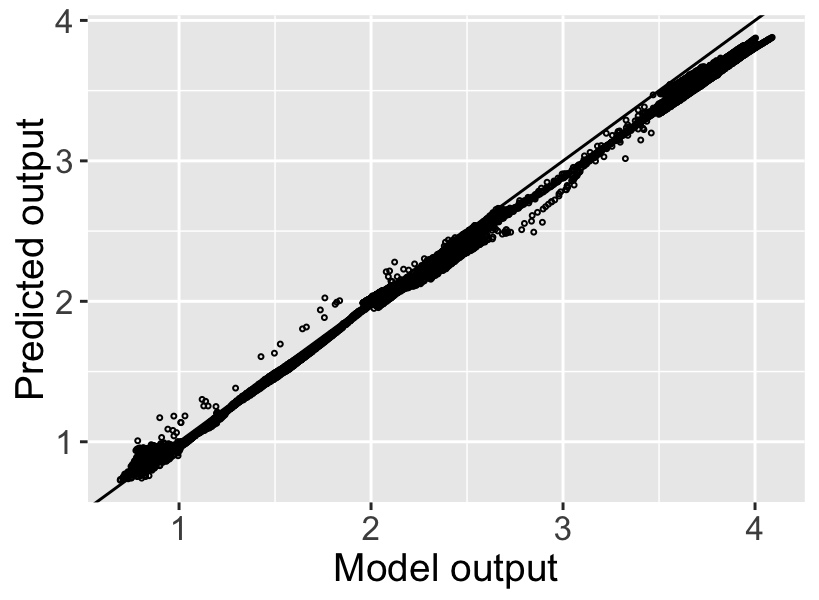}}
  \caption{PP kriging trained on high-fidelity runs only}
\end{subfigure}%
\begin{subfigure}{.5\textwidth}
  \centering
  \includegraphics[width=1.0\linewidth, height=0.18\textheight]{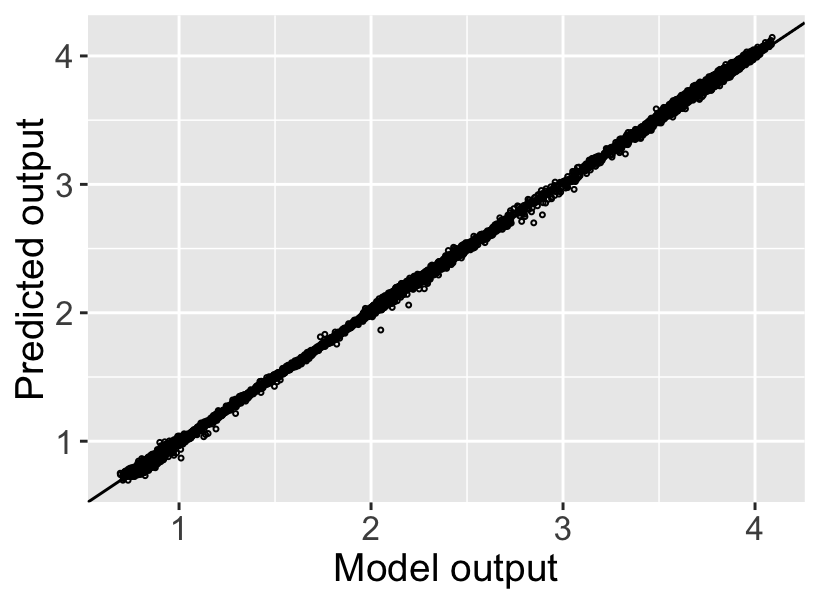}
  \caption{PP cokriging}
\end{subfigure}

 \caption{Scatter plot of predicted PSE against held-out PSE over $N=9,284$ spatial locations at the input setting $\bfx_1$. }
\label{fig: prediction versus held-out}
\end{figure}

In the storm surge application, the high-fidelity simulator is about 10 times slower than the low-fidelity simulator. Increasing the number of model runs in the high-fidelity simulator is therefore computationally prohibitive. The computational cost of predicting a new high-fidelity model run via the PP cokriging emulator is negligible compared to that needed to get a single run from the actual ADCIRC + SWAN simulator. This implies that emulating the high-fidelity simulator by using the our proposed PP cokriging emulator that combines only a small number of high-fidelity runs and a few hundred low-fidelity runs is preferable than using the low-fidelity simulator, in terms of both accuracy and computational cost.   The capability to use the low fidelity simulator, without substantial loss of accuracy through use of the PP cokriging emulator, to explore more  of the parameter space greatly enhances the feasibility of achieving high-precision modeling results without a massive computational budget.

\begin{table}[htbp]
\centering
\normalsize
   \caption{Predictive performance of emulators at $n^*=166$ held-out inputs over all $N=9,284$ spatial locations. PP kriging was trained based on low-fidelity runs  and high-fidelity runs separately. PP cokriging was trained based on both the low-fidelity simulation and high-fidelity simulation. PP = parallel partial.}
     {\resizebox{1.0\textwidth}{!}{%
  \setlength{\tabcolsep}{1.5em}
   \begin{tabular}{l c c c c c} 
   \toprule 
   \noalign{\vskip 1.5pt}
& RMSPE & CVG(95\%)   & ALCI(95\%) & CRPS & NSME \\  \noalign{\vskip 1.5pt}
\noalign{\vskip 1.5pt} \hline \noalign{\vskip 3pt}   \noalign{\vskip 1.5pt}
                 {PP kriging with low-fidelity data} & 0.130    & 0.042 & 0.023 & 0.109 &0.979  \\ \noalign{\vskip 4pt}
				 {PP kriging with high-fidelity data} &0.174     & 0.913  & 0.532 & 0.083 & 0.966 \\ 
				 \noalign{\vskip 4pt}				    
				 {PP cokriging} & 0.040  & 0.992  & 0.257 & 0.024 & 0.998 \\
\noalign{\vskip 1.5pt} \bottomrule
   \end{tabular}%
   }}
   \label{table: CV for cokriging}
\end{table}

\subsection{Uncertainty Analysis}

Cross-validation in the previous section showed that the PP cokriging emulator can provide very accurate predictions when compared to the true high-fidelity surge model runs in an overall sense. Figure~\ref{fig: map of storm surges at two input setting} compares the predicted storm surges against held-out storm surge from the high-fidelity surge model across $N=9,284$ spatial locations at two storm inputs that are used in Figure~\ref{fig: storm surges at two input settings} and Figure~\ref{fig: prediction versus held-out}. At these two inputs, the PP cokriging emulator seems to have large predictive uncertainties in the southeast region of Cape Coral and to have small predictive uncertainties in the Pine Island Sound Aquatic Preserve and the Caloosahatchee River. The largest predictive standard deviation in the PP cokriging emulator across all spatial locations is around 0.2, which is smaller than the difference between the high-fidelity simulator and the low-fidelity simulator. 

\begin{figure}
\begin{subfigure}{.333\textwidth}
  \centering
\makebox[\textwidth][c]{ \includegraphics[width=1.0\linewidth,height=0.12\textheight]{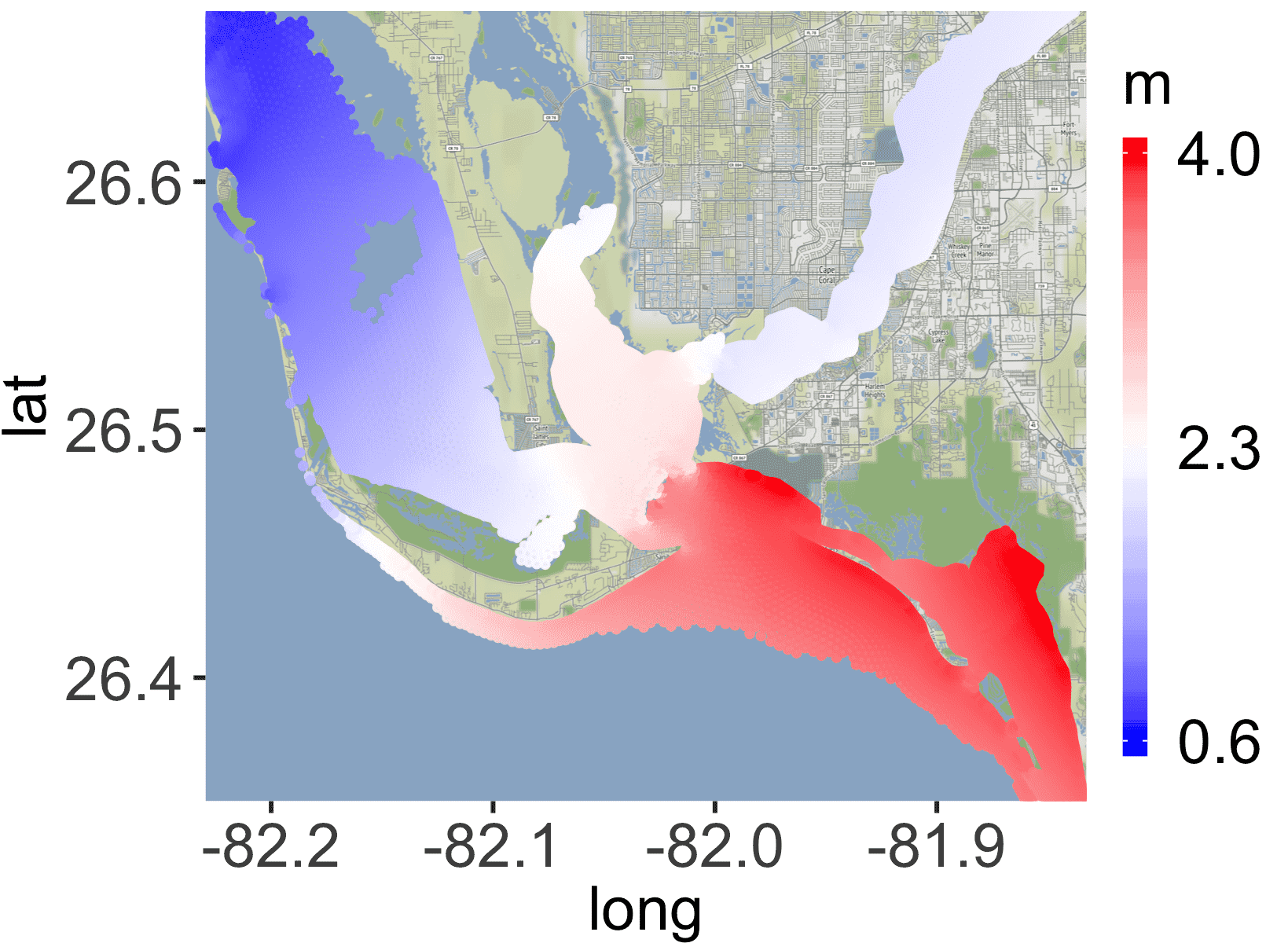}}
\end{subfigure}%
\begin{subfigure}{.333\textwidth}
  \centering
  \includegraphics[width=1.0\linewidth,height=0.12\textheight]{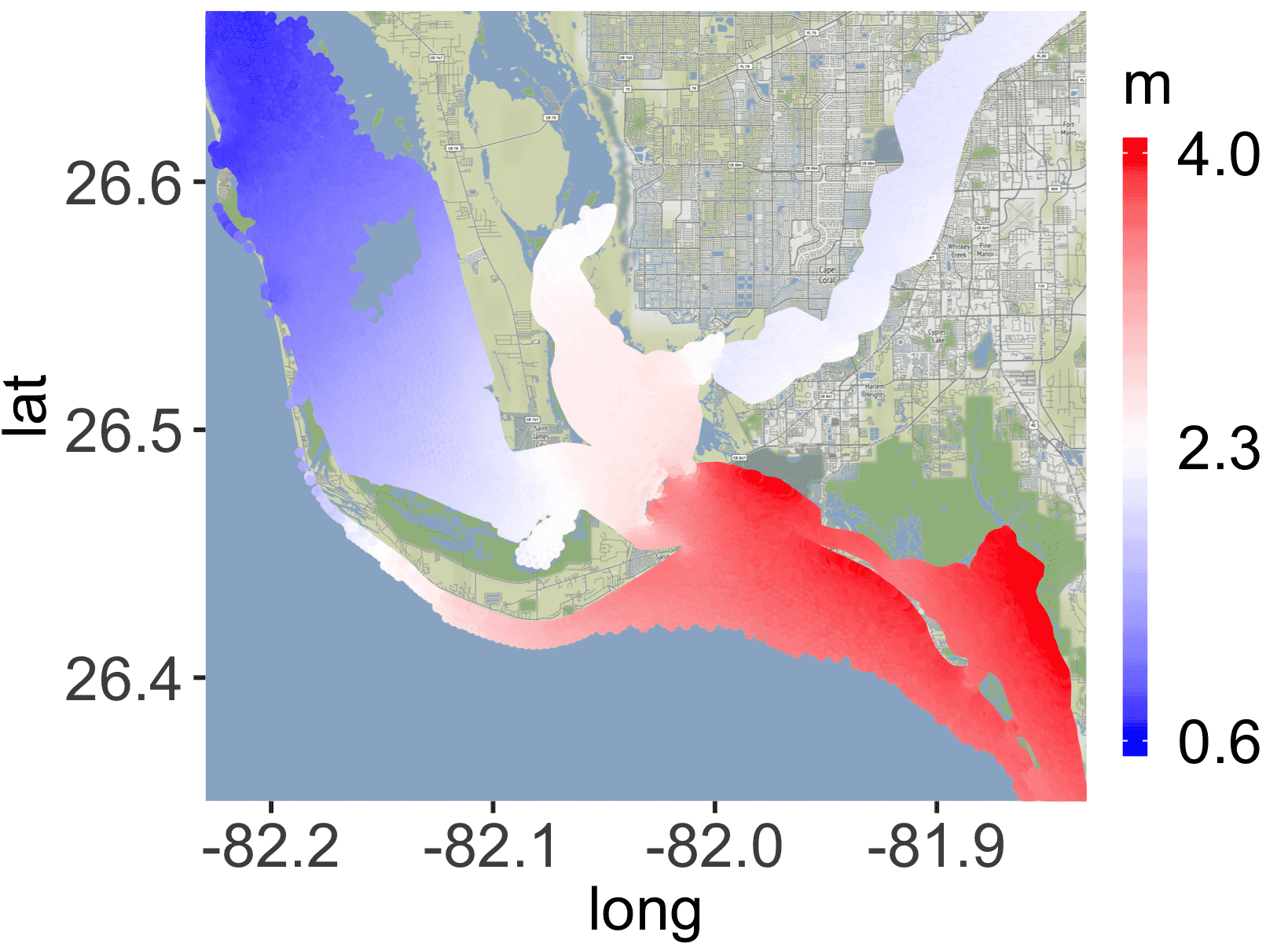}
\end{subfigure}%
\begin{subfigure}{.333\textwidth}
  \centering
  \includegraphics[width=1.0\linewidth,height=0.12\textheight]{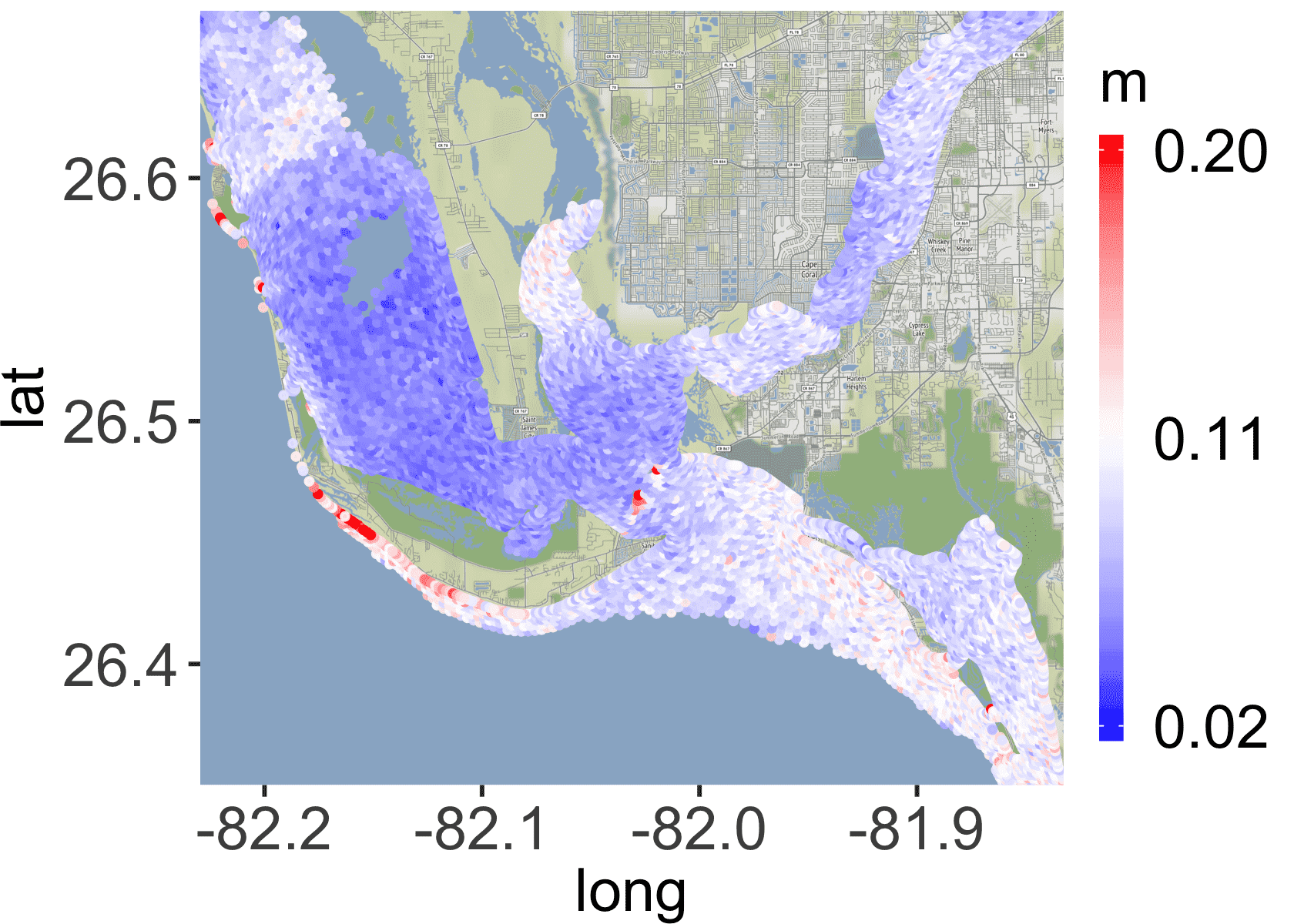}
\end{subfigure}

\begin{subfigure}{.333\textwidth}
  \centering
  \includegraphics[width=1.0\linewidth,height=0.12\textheight]{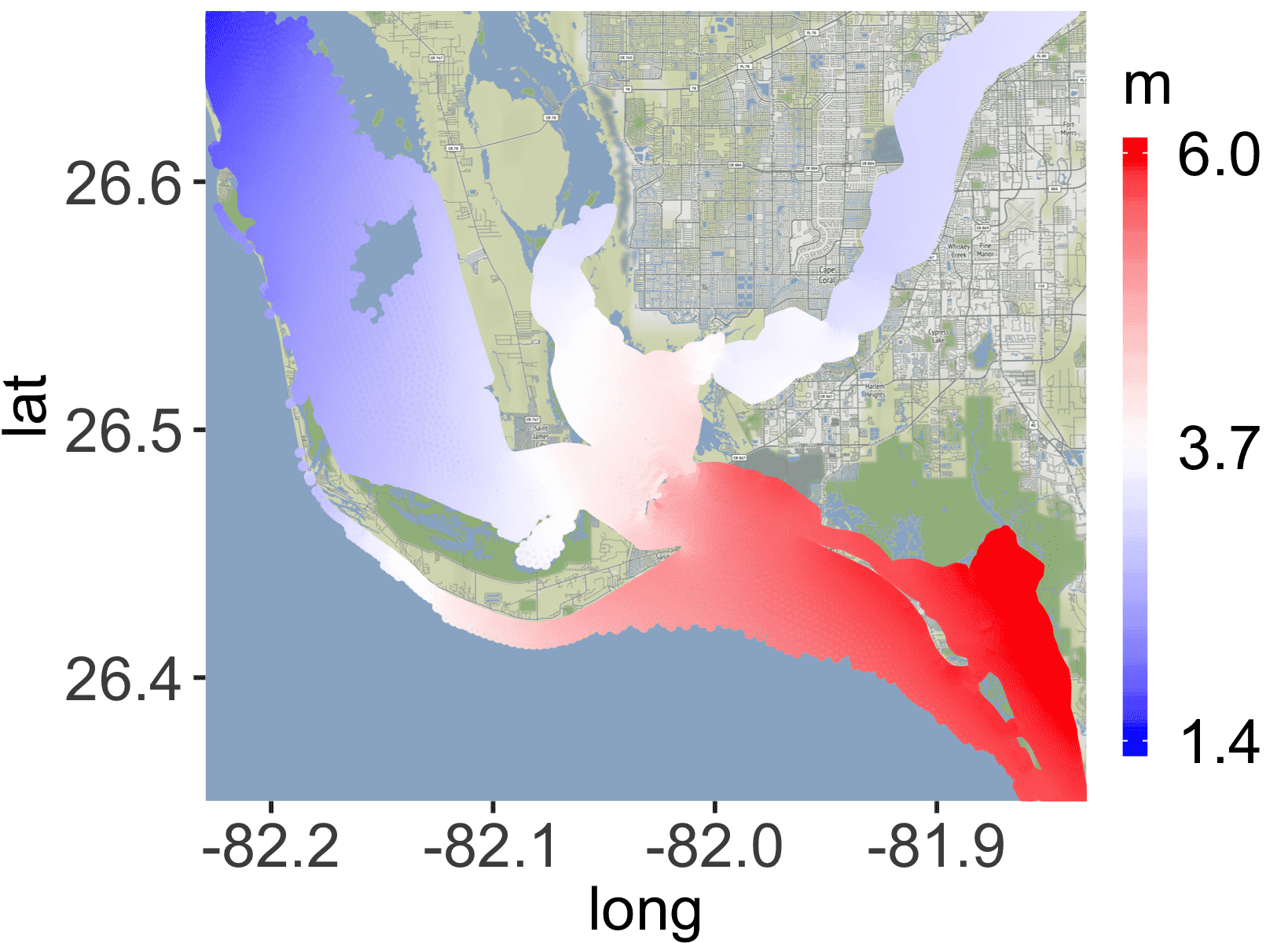}
\end{subfigure}%
\begin{subfigure}{.333\textwidth}
  \centering
  \includegraphics[width=1.0\linewidth,height=0.12\textheight]{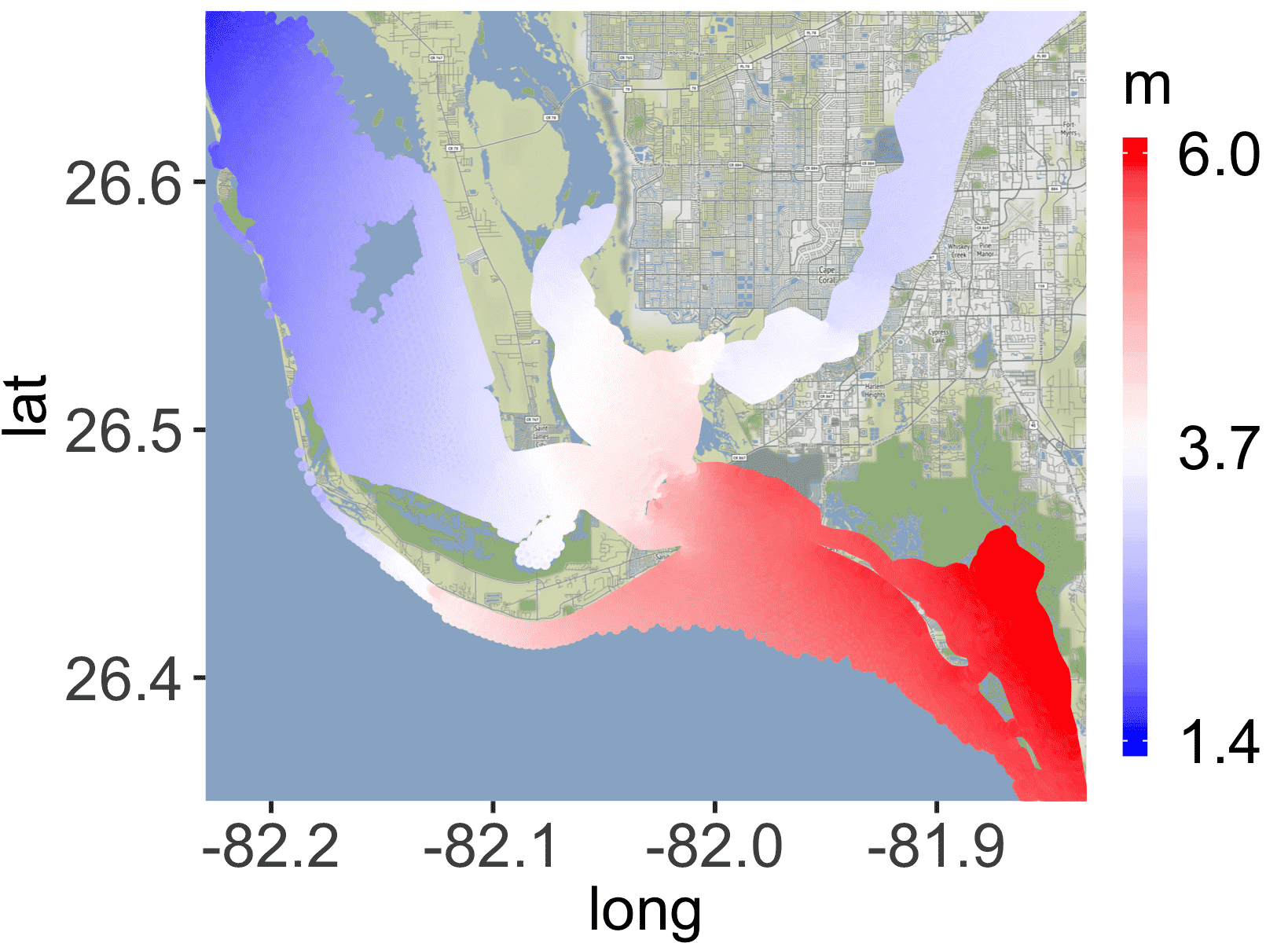}
\end{subfigure}%
\begin{subfigure}{.333\textwidth}
  \centering
  \includegraphics[width=1.0\linewidth,height=0.12\textheight]{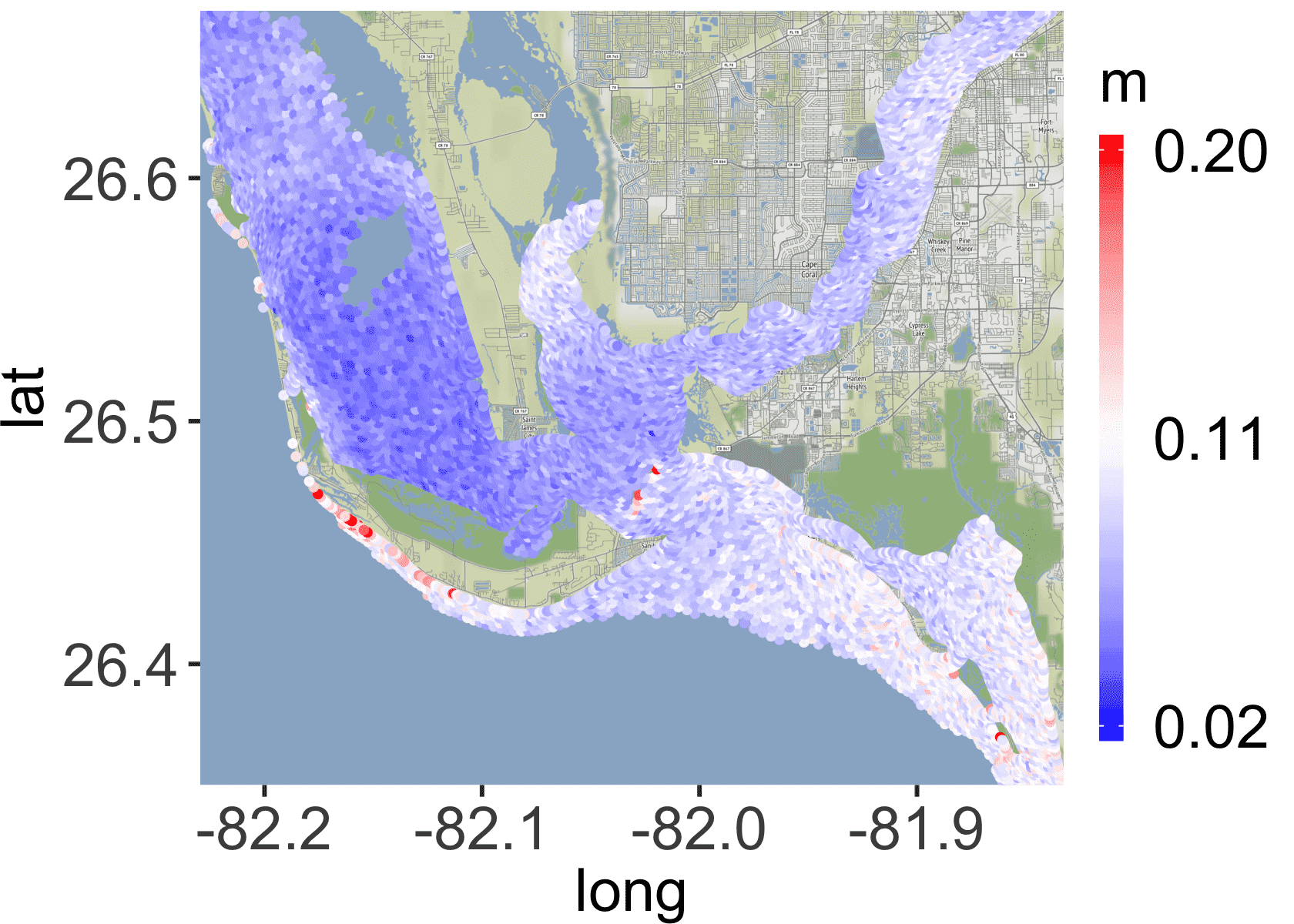}
\end{subfigure}
\caption{High-fidelity runs and predicted peak surge elevations with predictive standard errors at two input settings. The first column shows the high-fidelity runs at two different input settings. The second and third columns show the corresponding predicted PSE and associated predictive standard errors.}
\label{fig: map of storm surges at two input setting}
\end{figure}

Next, we explore the relationship between storm inputs and error structures in the PP cokriging emulator. We compute the prediction errors across all spatial locations at all held-out inputs. Figure~\ref{fig: error analysis} shows that the majority of emulation errors range from -0.5 to 0.5. This indicates that the PP cokriging emulator can capture the input-output relationship quite well. The residuals become larger as the central pressure deficit and the forward speed increase. The scale pressure radius seems to impact the emulation error in an opposite way as central pressure deficit. The residuals across different spatial locations are different as shown in Figure~\ref{fig: error analysis}. This indicates that the current PP cokriging emulator can partially capture the inhomogeneous structures in the output space, and some variations due to inputs are still left; see Section~\ref{sec: Discussion} for discussions on nonstationarity modeling in input space.

Finally, we show the parameter estimates for $\bfbeta_1$, $\bfsigma_1$, $\bfbeta_2$, $\bfgamma_1$, and $\bfsigma_2$ in Figure~\ref{fig: map of parameter estimates}. As we can see, these estimated parameters show strong spatially-varying structures at different regions. The estimated regression parameters $\hat{\bfbeta}_1$ and standard deviation $\hat{\bfsigma}_1$ at the low-fidelity level seem to be smoother than those estimates at fidelity level 2. This is because more variations are captured by the Gaussian process at the low-fidelity level. The remaining variations captured by the discrepancy function $\delta_{2,j}(\cdot)$ are small. This implies that the Gaussian process at the low-fidelity level fits well with model runs from the ADCIRC simulator and the discrepancy between the low-fidelity simulator and the high-fidelity simulator is relatively small. The estimated scale discrepancy parameters $\hat{\bfgamma}_1$ at all locations also show strongly heterogeneous spatial structures with values slightly greater than 1. This indicates that the high-fidelity simulator is more likely to generate higher values of storm surges than the low-fidelity simulator, but this trend is very small. The estimated standard deviations $\hat{\bfsigma}_1$ and $\hat{\bfsigma}_2$ seem to have more local structures than their corresponding regression parameters. This makes sense because we expect the regression trend in Gaussian processes to capture large-scale variations, and covariance structure to capture small-scale variations.

A central task in coastal engineering is to perform risk assessment of storm surges. The proposed emulator can produce predictive mean and predictive variance of storm surges over large number of spatial locations, which are practically useful and advantageous in risk assessment of storm surges in coastal flood hazard studies \citep[e.g.,][]{Cialone2017, niedoroda2010analysis, FEMA2017}, because computation of annual exceedance probabilities at one frequency level can require several thousands of storm surge simulations generated by a surge prediction model, which is prohibitively costly \citep[e.g.,][]{FEMA2017}. Practitioners in ocean engineering can use the proposed PP cokriging emulator to evaluate AEP using Monte Carlo approximations. This would be computationally efficient as running large number of storm surge simulations is not required once the emulator is readily available. The application of an emulator in risk assessment of storm surges is pursued elsewhere. 

\begin{figure}[htbp]
\centering
\makebox[\textwidth][c]{ \includegraphics[width=1.0\textwidth, height=0.30\textheight]{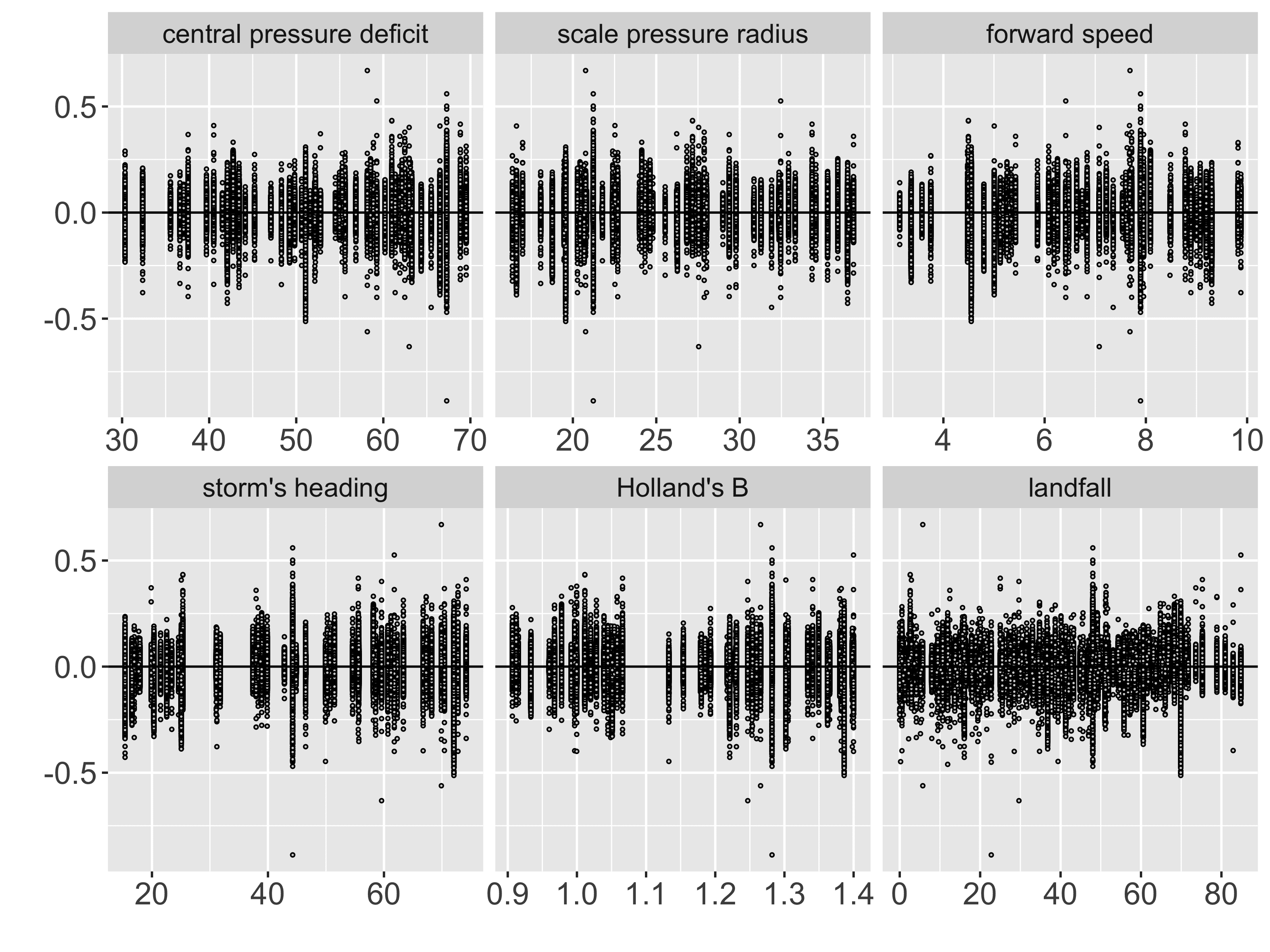}}
\caption{Prediction errors across all $N=9,284$ spatial locations against each storm parameter at all held-out inputs.}
\label{fig: error analysis}
\end{figure}

\begin{figure}
\captionsetup[subfigure]{justification=centering}
\begin{subfigure}{.333\textwidth}
  \centering
\makebox[\textwidth][c]{ \includegraphics[width=1.0\linewidth, height=0.12\textheight]{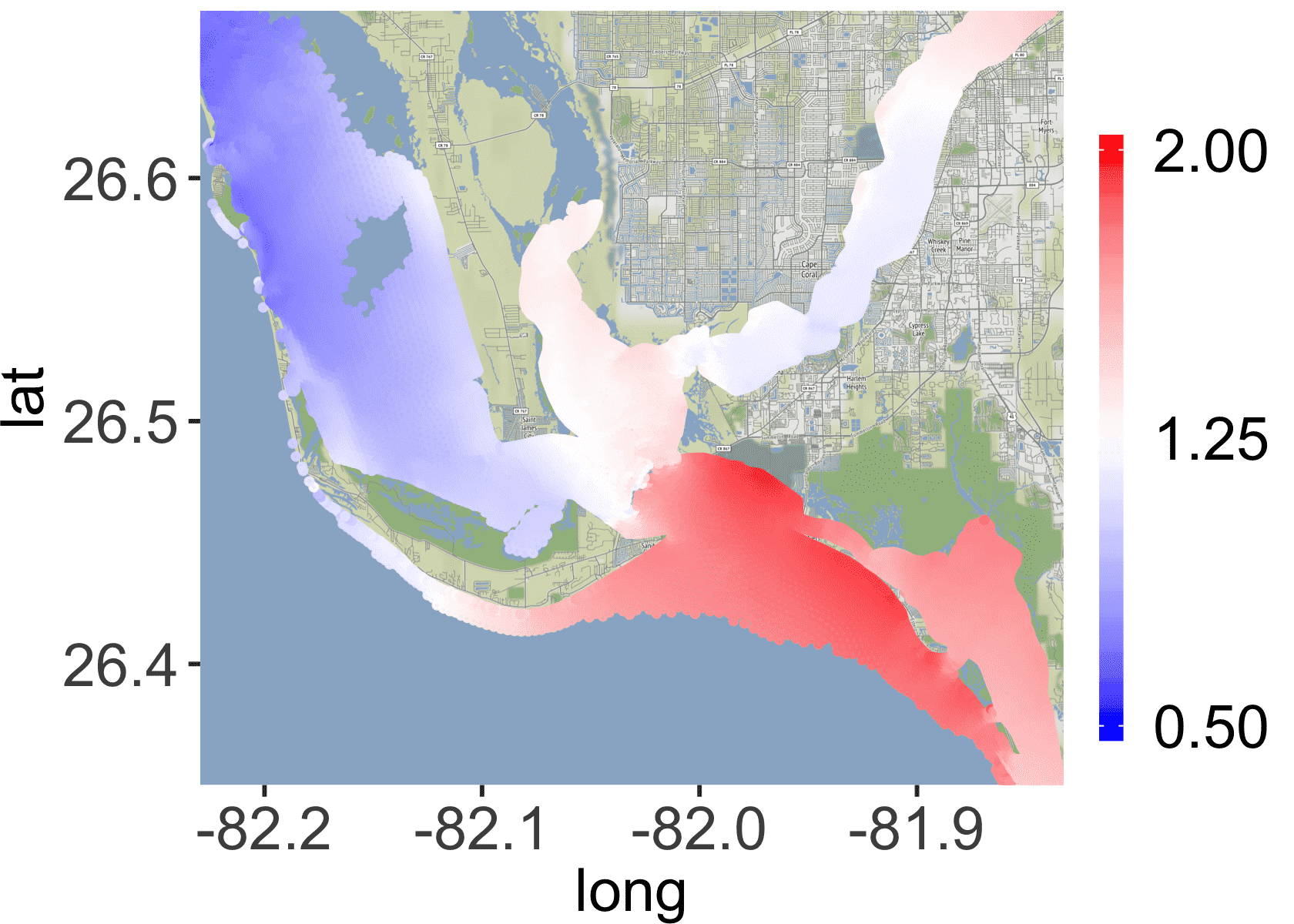}}
  \caption{$\hat{\bfbeta}_1$}
\end{subfigure}%
\begin{subfigure}{.333\textwidth}
  \centering
  \includegraphics[width=1.0\linewidth,height=0.12\textheight]{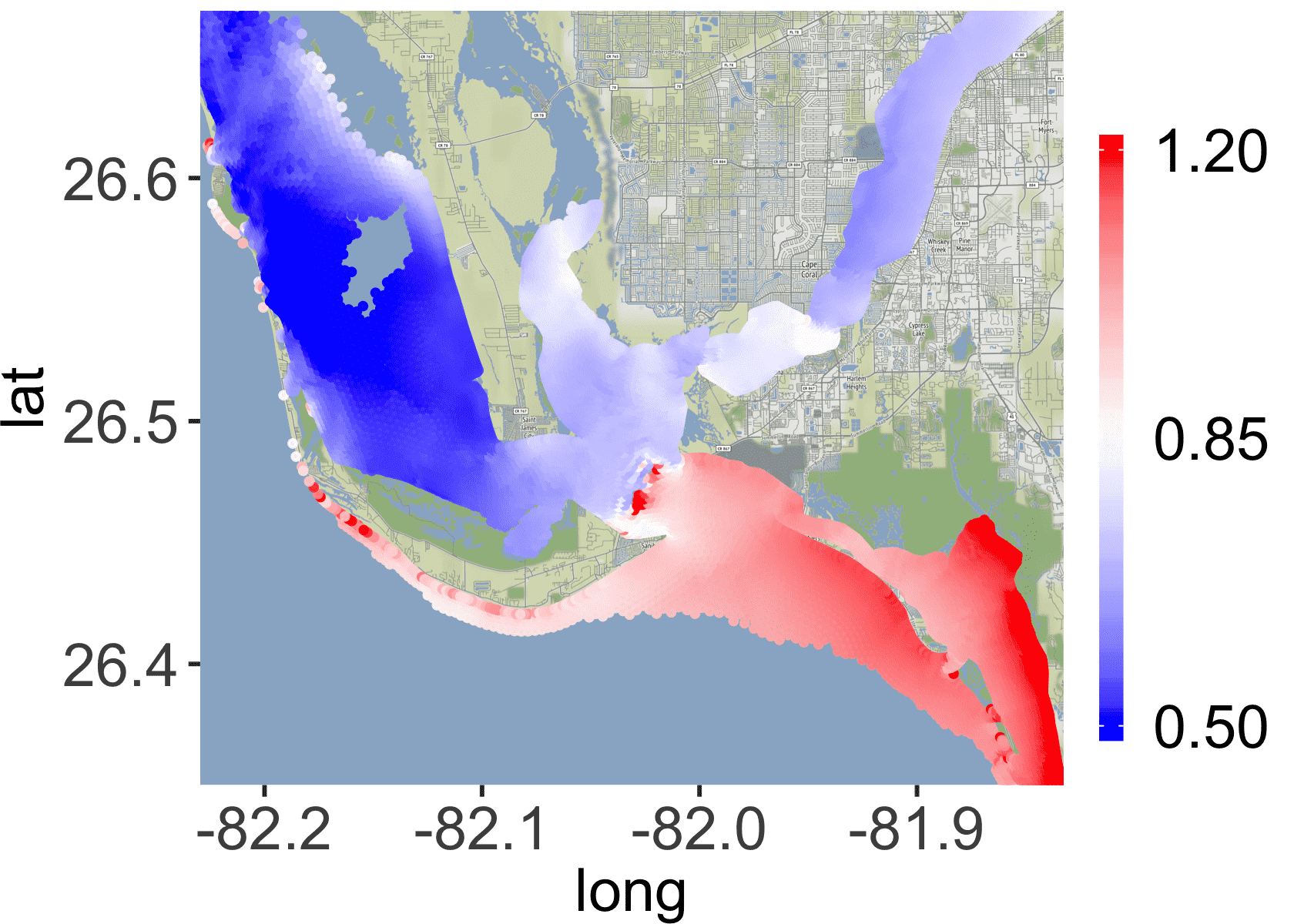}
  \caption{$\hat{\bfsigma}_1$}
\end{subfigure}%
\begin{subfigure}{.333\textwidth}
  \centering
    \includegraphics[width=1.0\linewidth,height=0.12\textheight]{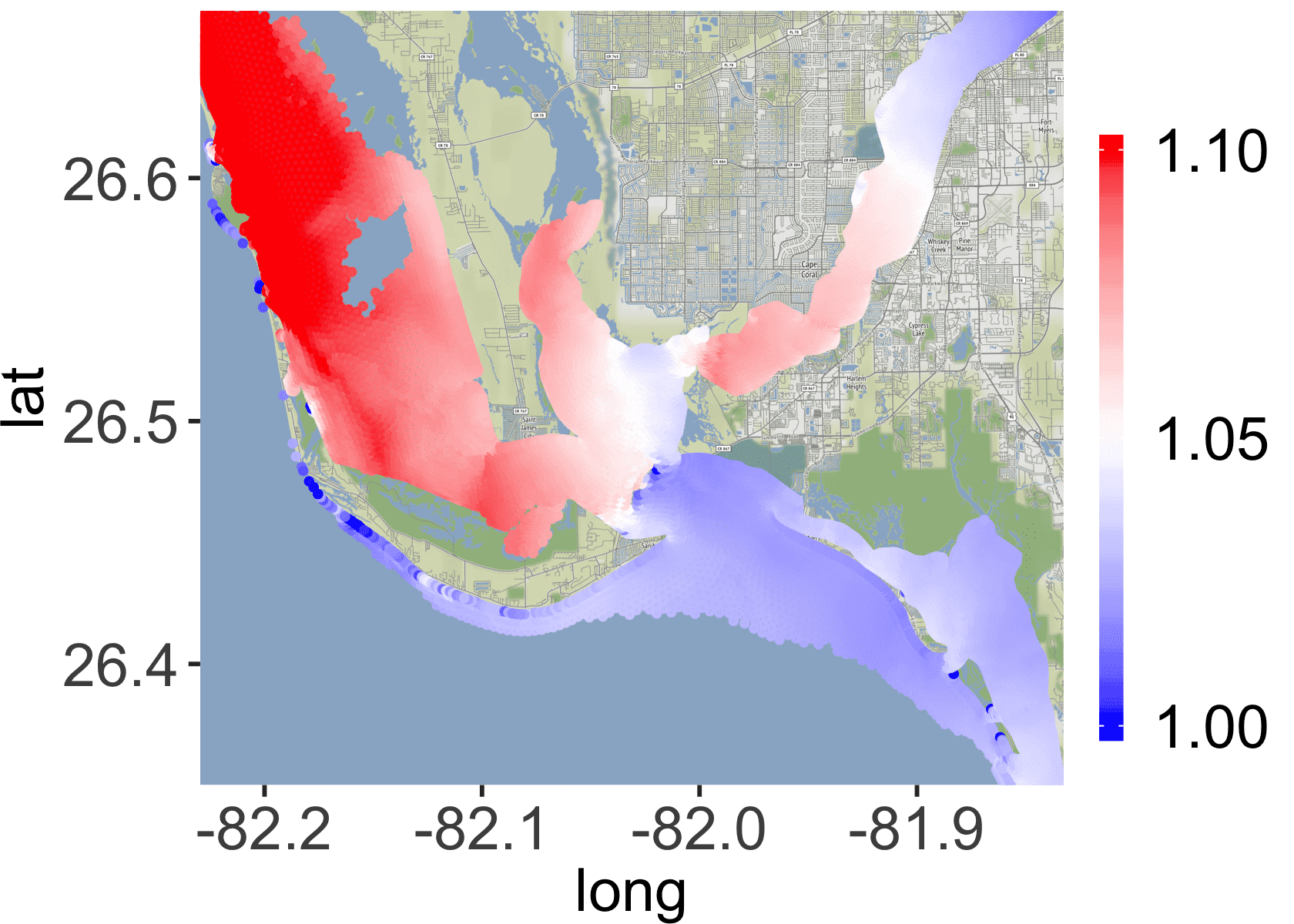}
  \caption{$\hat{\bfgamma}_1$}
\end{subfigure}

\begin{subfigure}{.333\textwidth}
  \centering
  \includegraphics[width=1.0\linewidth,height=0.12\textheight]{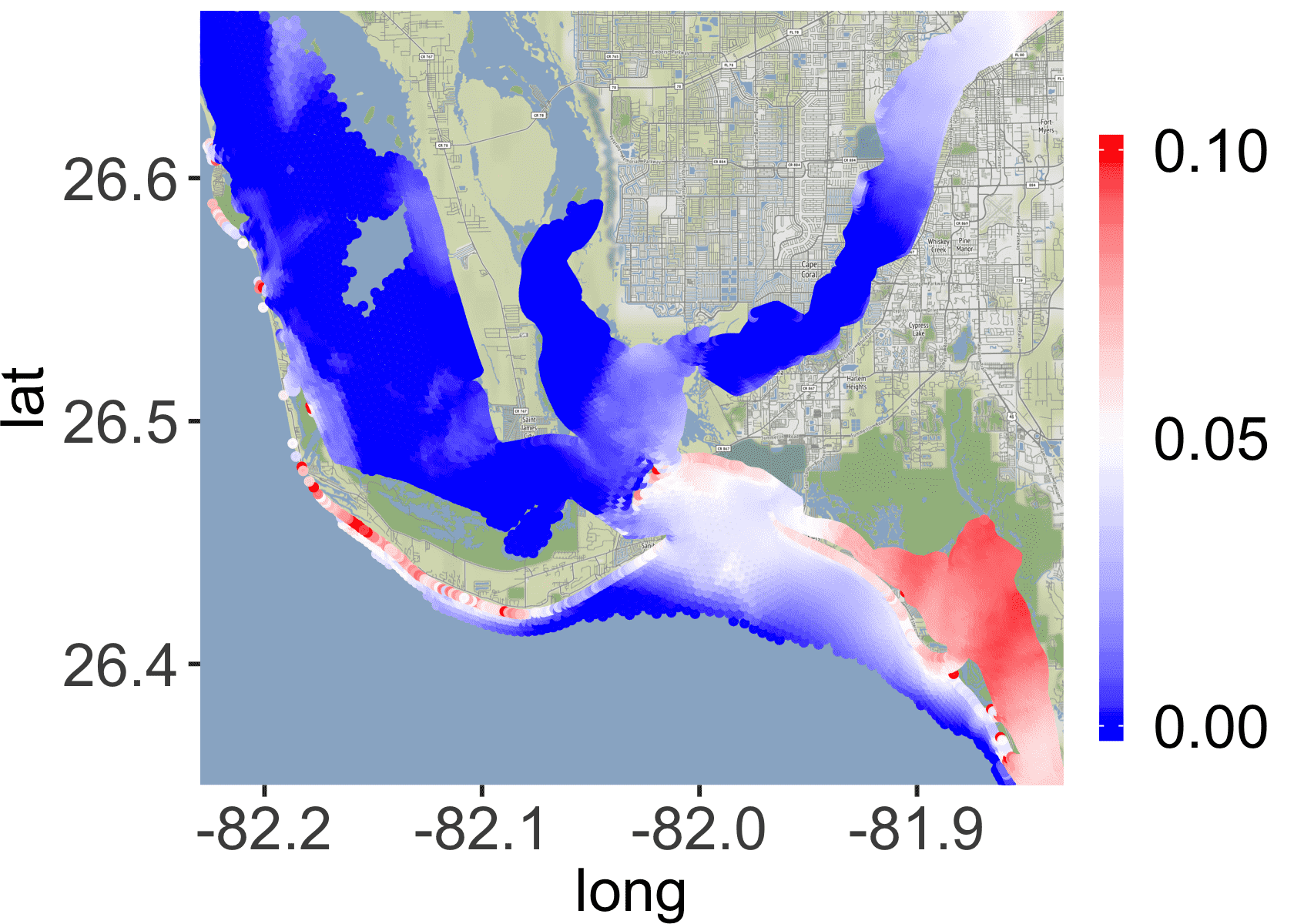}
  \caption{$\hat{\bfbeta}_2$}
\end{subfigure}%
\begin{subfigure}{.333\textwidth}
  \centering
  \includegraphics[width=1.0\linewidth,height=0.12\textheight]{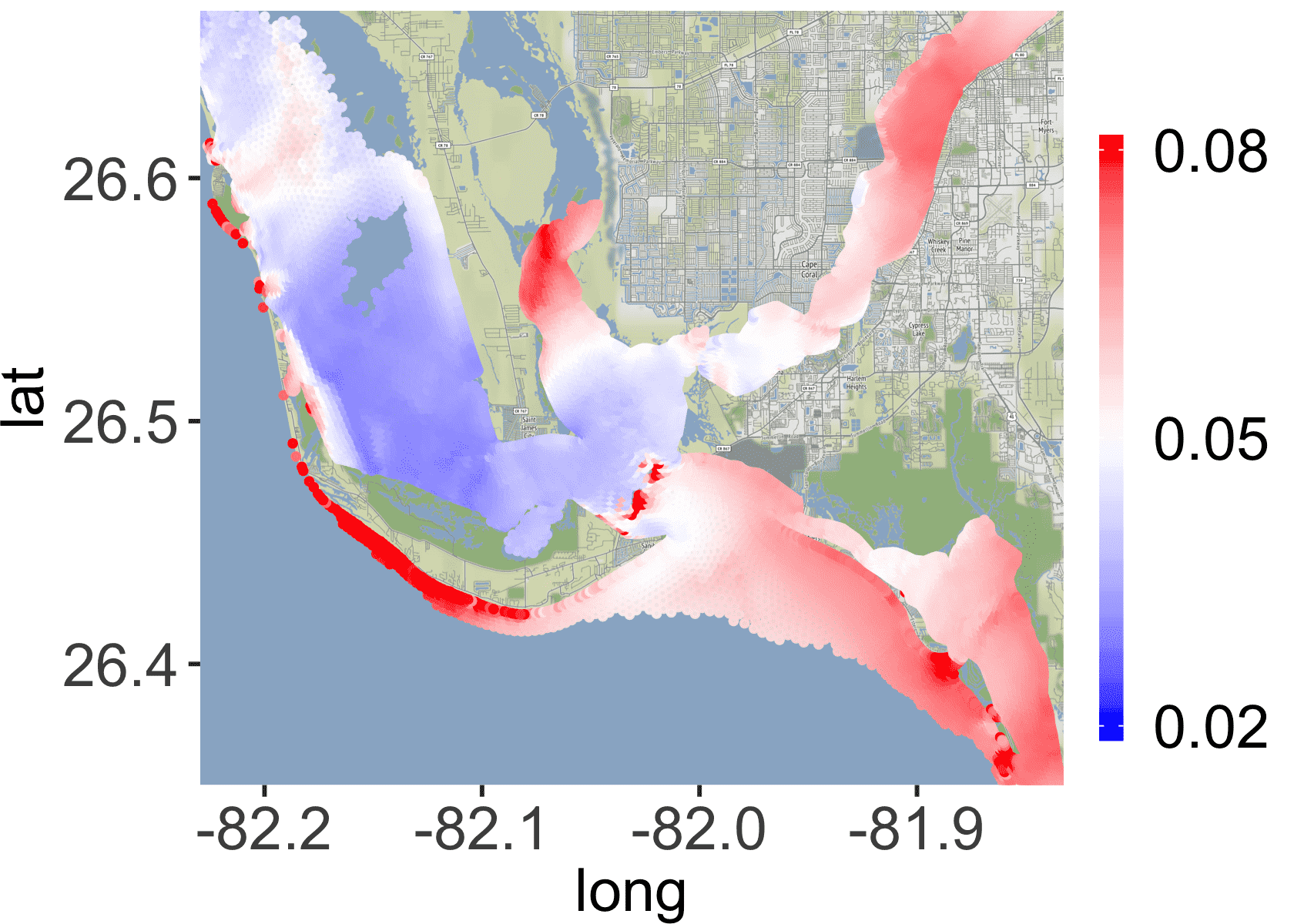}
  \caption{$\hat{\bfsigma}_2$}
\end{subfigure}
\caption{Estimated parameters across all spatial locations. The estimated parameters show strong heterogeneous spatial patterns.}
\label{fig: map of parameter estimates}
\end{figure}

\section{Discussion} \label{sec: Discussion}

Coastal flood hazard studies by FEMA and USACE use ADCIRC + SWAN to quantify the storm surge hazard, where simulation from this computer model is time-consuming and resource-intensive. We have built a parallel partial cokriging emulator to predict storm surges using simulations from both ADCIRC + SWAN and ADCIRC. The PP cokriging emulator effectively provides the marginal predictive distributions of storm surges for any storm characteristics at each spatial location. These marginal predictive distributions are also almost identical to the marginal predictive distributions obtained in situations where any spatial correlation matrix is used. These marginal predictive distributions can be directly used to compute annual exceedance probabilitiies in coastal flood hazard studies, which provide a convenient tool for practitioners in ocean engineering to perform risk assessment and surge forecasting research. The PP cokriging emulator not only has similar prediction accuracy as the high-fidelity simulator ADCIRC + SWAN, but also has a linear computational cost in terms of output values and also induces nonstationarity in output space, which is crucial to capture non-smooth storm surge surface. Field measurements of historical storm data including both observed surges and characteristics of storms that making landfall in Southwest Florida (SWFL) are very limited with only one gague station operating in Fort Myers. If there were sufficient storm data available, the proposed emulator could be further used to model the discrepancy between actual storm surges and ADCIRC + SWAN simulations, since accurate prediction of the actual storm surges can help perform more accurate coastal flood hazard studies. Although this paper focuses on a small region in SWFL, there is also a significant interest in extending the current method from a local region to the entire SWFL region so that the relationship between AEP and storm parameters can be quantified over a much larger coastal region.

The PP cokriging emulator assumes conditional independence across spatial locations that essentially leads to a separable covariance structure between input space and output space to simplify computations. This assumption can help capture nonstationary spatial patterns in the storm surge application. If interest lies in joint modeling across spatial locations, one can choose a spatial window to enable joint modeling. A related concern is the assumption of common correlation parameters at all spatial locations. If correlation parameters differ at each spatial location, the computational cost would not be linear in terms of spatial coordinates. This is a key advantage when the hazard from storm surge is assessed over a large spatial domain. One can potentially partition the domain into a set of subregions and allow different correlation parameters across these subregions.

In the storm surge application, we used a limited number of runs to train the emulator due to computational constraints. To aid coastal flood hazard studies and storm surge forecasting, one may use the proposed PP cokriging emulator to setup the design in a statistical optimal way such as sequential design \citep{LeGratiet2015} or to use large number of model runs. The latter problem can be tackled via computationally efficient Gaussian process approximation approaches \citep[e.g.][]{Gramacy2015}.  
The storm surge output shows quite rough structure across the spatial domain due to hurricane characteristics and heterogeneous topography. One can introduce nonstationarity in input space via treed Gaussian process \citep{Gramacy2008, konomikaragiannisABTCK2019}. Another interesting exploration for the proposed methodology is related to non-nested design on how much gain is obtained by allowing the design is not hierarchically nested over the traditional nested design. There is also an interest in developing  emulators when multifidelity computer codes have different spatial resolutions.  These possible directions could be pursued in future.

\section*{Acknowledgements}

This research began under the auspices and support of the Statistical and Applied Mathematical Sciences Institute (SAMSI) 2018-2019 research program on Model Uncertainty: Mathematical and Statistical (MUMS) in the Storm Surge Working Group that was supported by the U.S. National Science Foundation (NSF) under Grant DMS-1638521. This material is partly based upon work supported by the U.S. Department of Homeland Security (DHS) under Grant Award Number 2015-ST-061-ND0001-01 and by the U.S. Coastal Research Program (USCRP) administered by the U.S. Army Corps of Engineering (USACE), Department of Defense (DoD). The views and conclusions contained in this document are those of the authors and should not be interpreted as necessarily representing the official policies, either expressed or implied, of NSF, DHS, and DoD, and no official endorsement should be inferred. Ma's work used the Extreme Science and Engineering Discovery Environment (XSEDE) Lonestar5 and Stampede2 at the Texas Advanced Computing Center (TACC) through allocation DMS180042.  Ma gratefully acknowledges the support by the postdoctoral fellowship at SAMSI and by the USCRP program under Cooperative Agreement Number W912HZ-2020011 while part of this research was conducted at SAMSI and Duke University. Ma is grateful to Professor Jim Berger for his valuable discussion. Asher gratefully acknowledges the support by the USCRP program under  Cooperative Agreement Number W912HZ-2020011 while part of this research was conducted. Toro gratefully acknowledges in-house research funding provided by Lettis Consultants International, Inc.~for making possible his participation. 

\section*{Supplement Material}

The Supplement Material contains technical details and additional results. Code for the numerical examples can be found in \url{https://github.com/pulongma/PPCokriging}.


\begin{singlespace}
\bibliographystyle{apalike}
\setlength{\bibsep}{5pt}
\bibliography{main}

\begin{thebibliography}{}

\bibitem[Aerts et~al., 2014]{aerts2014evaluating}
Aerts, J.~C., Botzen, W.~W., Emanuel, K., Lin, N., De~Moel, H., and
  Michel-Kerjan, E.~O. (2014).
\newblock Evaluating flood resilience strategies for coastal megacities.
\newblock {\em Science}, 344(6183):473--475.

\bibitem[Asher and Liu, 2022]{Asher2022}
Asher, T.~G. and Liu, Y. (2022).
\newblock Coastal flood modeling database - {S}outhwest {F}lorida.
\newblock PRJ-2433, DesignSafe-CI.
  https://www.designsafe-ci.org/data/browser/public/. In review.

\bibitem[Blanton et~al., 2018]{blanton2018integrated}
Blanton, B., Dresback, K., Colle, B., Kolar, R., Vergara, H., Hong, Y.,
  Leonardo, N., Davidson, R., Nozick, L., and Wachtendorf, T. (2018).
\newblock An integrated scenario ensemble-based framework for hurricane
  evacuation modeling: Part 2{-}hazard modeling.
\newblock {\em Risk Analysis}.

\bibitem[Booij et~al., 1999]{Booij1999}
Booij, N., Ris, R.~C., and Holthuijsen, L.~H. (1999).
\newblock A third-generation wave model for coastal regions: 1. {Model}
  description and validation.
\newblock {\em Journal of Geophysical Research: Oceans}, 104(C4):7649--7666.

\bibitem[Cardone and Cox, 2009]{Cardone2009}
Cardone, V.~J. and Cox, A.~T. (2009).
\newblock Tropical cyclone wind field forcing for surge models: critical issues
  and sensitivities.
\newblock {\em Natural Hazards}, 51(1):29--47.

\bibitem[Cialone et~al., 2017]{Cialone2017}
Cialone, M.~A., Grzegorzewski, A.~S., Mark, D.~J., Bryant, M.~A., and Massey,
  T.~C. (2017).
\newblock Coastal-storm model development and water-level validation for the
  {North Atlantic} coast comprehensive study.
\newblock {\em Journal of Waterway, Port, Coastal, and Ocean Engineering},
  143(5):04017031.

\bibitem[Conti and O'Hagan, 2010]{Conti2010}
Conti, S. and O'Hagan, A. (2010).
\newblock {Bayesian emulation of complex multi-output and dynamic computer
  models}.
\newblock {\em Journal of Statistical Planning and Inference}, 140(3):640--651.

\bibitem[Cressie, 1993]{Cressie1993}
Cressie, N. (1993).
\newblock {\em Statistics for Spatial Data}.
\newblock John Wiley \& Sons, New York, revised edition.

\bibitem[de~Moel and Aerts, 2011]{de2011effect}
de~Moel, H. and Aerts, J. (2011).
\newblock Effect of uncertainty in land use, damage models and inundation depth
  on flood damage estimates.
\newblock {\em Natural Hazards}, 58(1):407--425.

\bibitem[Dempster et~al., 1977]{Dempster1977}
Dempster, A.~P., Laird, N.~M., and Rubin, D.~B. (1977).
\newblock {Maximum likelihood from incomplete data via the EM algorithm}.
\newblock {\em Journal of the Royal Statistical Society: Series B},
  39(1):1--38.

\bibitem[Dietrich et~al., 2011]{Dietrich2011}
Dietrich, J., Zijlema, M., Westerink, J., Holthuijsen, L., Dawson, C.,
  Luettich, R., Jensen, R., Smith, J., Stelling, G., and Stone, G. (2011).
\newblock Modeling hurricane waves and storm surge using integrally-coupled,
  scalable computations.
\newblock {\em Coastal Engineering}, 58(1):45 -- 65.

\bibitem[Dietrich et~al., 2012]{Dietrich2012}
Dietrich, J.~C., Tanaka, S., Westerink, J.~J., Dawson, C.~N., Luettich, R.~A.,
  Zijlema, M., Holthuijsen, L.~H., Smith, J.~M., Westerink, L.~G., and
  Westerink, H.~J. (2012).
\newblock Performance of the unstructured-mesh, {SWAN+ADCIRC} model in
  computing hurricane waves and surge.
\newblock {\em Journal of Scientific Computing}, 52(2):468--497.

\bibitem[{FEMA}, 2006]{FEMA2006}
{FEMA} (2006).
\newblock {Hurricane Katrina in the Gulf Coast}.
\newblock Washington, D.C.: Federal Emergency Management Agency.

\bibitem[{FEMA}, 2008]{FEMA2008}
{FEMA} (2008).
\newblock Mississippi coastal analysis project.
\newblock Project reports prepared by {URS Group Inc.}, (Graithersburg MD and
  Tallahassee FL) under HMTAP Contract HSFEHQ-06-D-0162, Task Order 06-J-0018.

\bibitem[{FEMA}, 2017]{FEMA2017}
{FEMA} (2017).
\newblock Southwest {{Florida Storm Surge Study Intermediate Data Submittal}}
  2.
\newblock Report Prepared by the {Risk Assessment}, {{Mapping}}, and {Planning
  Partners} ({RAMPP}) under {FEMA IDIQ Contract HSFEHQ}-09-{{D}}-0369 and {Task
  Order HSFE04}-13-{{J}}-0097. FEMA Region IV.

\bibitem[Fischbach et~al., 2016]{fischbach2016bias}
Fischbach, J., Johnson, D., and Kuhn, K. (2016).
\newblock Bias and efficiency tradeoffs in the selection of storm suites used
  to estimate flood risk.
\newblock {\em Journal of Marine Science and Engineering}, 4(1):10.

\bibitem[Georgas et~al., 2016]{georgas2016stevens}
Georgas, N., Blumberg, A., Herrington, T., Wakeman, T., Saleh, F., Runnels, D.,
  {Jordi Ballester}, A., Ying, K., Yin, L., Ramaswamy, V., Yakubovskiy, A.,
  Lopez, O., Mcnally, J., Schulte, J., and Wang, Y. (2016).
\newblock The stevens flood advisory system: Operational {H3E} flood forecasts
  for the greater {New York/New Jersey} metropolitan region.
\newblock {\em International Journal of Safety and Security Engineering},
  6(3):648--662.

\bibitem[Gneiting and Raftery, 2007]{Gneiting2007}
Gneiting, T. and Raftery, A.~E. (2007).
\newblock Strictly proper scoring rules, prediction, and estimation.
\newblock {\em Journal of the American Statistical Association},
  102(477):359--378.

\bibitem[Gramacy and Apley, 2015]{Gramacy2015}
Gramacy, R.~B. and Apley, D.~W. (2015).
\newblock Local {Gaussian} process approximation for large computer
  experiments.
\newblock {\em Journal of Computational and Graphical Statistics},
  24(2):561--578.

\bibitem[Gramacy and Lee, 2008]{Gramacy2008}
Gramacy, R.~B. and Lee, H. K.~H. (2008).
\newblock Bayesian treed {Gaussian} process models with an application to
  computer modeling.
\newblock {\em Journal of the American Statistical Association},
  103(483):1119--1130.

\bibitem[Gu, 2019]{Gu2019}
Gu, M. (2019).
\newblock Jointly robust prior for {Gaussian} stochastic process in emulation,
  calibration and variable selection.
\newblock {\em Bayesian Analysis}, 14(3):877--905.

\bibitem[Gu and Berger, 2016]{Gu2016}
Gu, M. and Berger, J.~O. (2016).
\newblock Parallel partial {Gaussian} process emulation for computer models
  with massive output.
\newblock {\em The Annals of Applied Statistics}, 10(3):1317--1347.

\bibitem[Gu et~al., 2019]{Gu2018RobustGP}
Gu, M., Palomo, J., and Berger, J.~O. (2019).
\newblock {RobustGaSP}: Robust {G}aussian stochastic process emulation in {R}.
\newblock {\em {The R Journal}}, 11(1):112--136.

\bibitem[Hesser et~al., 2013]{Hesser2013}
Hesser, T.~J., Cialone, M.~A., and Anderson, M.~E. (2013).
\newblock {Lake St. Clair: Storm wave and water level modeling}.

\bibitem[Higdon et~al., 2008]{Higdon2008}
Higdon, D., Gattiker, J., and Williams, B. (2008).
\newblock {Computer model calibration using high-dimensional output}.
\newblock {\em Journal of the American Statistical Association},
  103(482):570--583.

\bibitem[Jensen et~al., 2012]{Jensen2012}
Jensen, R.~E., Cialone, M.~A., Chapman, R.~S., Ebersole, B.~A., Anderson, M.,
  and Thomas, L. (2012).
\newblock {Lake {Michigan} storm: Wave and water level modeling}.

\bibitem[Kennedy and O'Hagan, 2000]{Kennedy2000}
Kennedy, M. and O'Hagan, A. (2000).
\newblock {Predicting the output from a complex computer code when fast
  approximations are available}.
\newblock {\em Biometrika}, 87(1):1--13.

\bibitem[Konomi and Karagiannis, 2021]{konomikaragiannisABTCK2019}
Konomi, B.~A. and Karagiannis, G. (2021).
\newblock Bayesian analysis of multifidelity computer models with local
  features and non-nested experimental designs: {A}pplication to the {WRF}
  model.
\newblock {\em Technometrics}, 63(4):510--522.

\bibitem[Le~Gratiet, 2013]{Gratiet2013}
Le~Gratiet, L. (2013).
\newblock Bayesian analysis of hierarchical multifidelity codes.
\newblock {\em SIAM/ASA Journal on Uncertainty Quantification}, 1(1):244--269.

\bibitem[Le~Gratiet and Cannamela, 2015]{LeGratiet2015}
Le~Gratiet, L. and Cannamela, C. (2015).
\newblock Cokriging-based sequential design strategies using fast
  cross-validation techniques for multi-fidelity computer codes.
\newblock {\em Technometrics}, 57(3):418--427.

\bibitem[Liu et~al., 2019]{liu2019physical}
Liu, Y., Asher, T.~G., and Irish, J.~L. (2019).
\newblock Physical drivers of changes in probabilistic surge hazard under sea
  level rise.
\newblock {\em Earth's Future}, 7(7):819--832.

\bibitem[Luettich and Westerink, 2004]{Luettich2004}
Luettich, R. and Westerink, J. (2004).
\newblock Formulation and numerical implementation of the {2D/3D ADCIRC} finite
  element model version 44.{XX}.

\bibitem[Ma, 2019]{Ma2019ARCokrig}
Ma, P. (2019).
\newblock {ARCokrig}: Autoregressive cokriging models for multifidelity codes.
\newblock R package version 0.1.2. https://CRAN.R-project.org/package=ARCokrig.

\bibitem[Ma, 2020]{Ma2020OBayes}
Ma, P. (2020).
\newblock {Objective Bayesian Analysis of a Cokriging Model for Hierarchical
  Multifidelity Codes}.
\newblock {\em SIAM/ASA Journal on Uncertainty Quantification},
  8(4):1358--1382.

\bibitem[Ma and Bhadra, 2022]{Ma2019cov}
Ma, P. and Bhadra, A. (2022).
\newblock Beyond {M}at\'ern: On a class of interpretable confluent
  hypergeometric covariance functions.
\newblock {\em Journal of the American Statistical Association, To appear}.
\newblock DOI:10.1080/01621459.2022.2027775.

\bibitem[Marsooli and Lin, 2018]{marsooli2018numerical}
Marsooli, R. and Lin, N. (2018).
\newblock Numerical modeling of historical storm tides and waves and their
  interactions along the {US} east and {Gulf} coasts.
\newblock {\em Journal of Geophysical Research: Oceans}, 123(5):3844--3874.

\bibitem[Niedoroda et~al., 2010]{niedoroda2010analysis}
Niedoroda, A., Resio, D., Toro, G., Divoky, D., Das, H., and Reed, C. (2010).
\newblock Analysis of the coastal {Mississippi} storm surge hazard.
\newblock {\em Ocean Engineering}, 37(1):82--90.

\bibitem[{NOAA National Centers for Environmental Information}, 2019]{NOAA2019}
{NOAA National Centers for Environmental Information} (2019).
\newblock {U.S.} billion-dollar weather and climate disasters.
\newblock URL: https://www.ncdc.noaa.gov/billions/.

\bibitem[Pielke~Jr et~al., 2008]{pielke2008normalized}
Pielke~Jr, R.~A., Gratz, J., Landsea, C.~W., Collins, D., Saunders, M.~A., and
  Musulin, R. (2008).
\newblock Normalized hurricane damage in the {United States}: 1900--2005.
\newblock {\em Natural Hazards Review}, 9(1):29--42.

\bibitem[Qian and Wu, 2008]{Qian2008}
Qian, P. Z.~G. and Wu, C. F.~J. (2008).
\newblock Bayesian hierarchical modeling for integrating low-accuracy and
  high-accuracy experiments.
\newblock {\em Technometrics}, 50(2):192--204.

\bibitem[Rappaport, 2014]{rappaport2014fatalities}
Rappaport, E.~N. (2014).
\newblock {Fatalities in the United States from Atlantic tropical cyclones: New
  data and interpretation}.
\newblock {\em Bulletin of the American Meteorological Society},
  95(3):341--346.

\bibitem[Sacks et~al., 1989]{Sacks1989}
Sacks, J., Welch, W.~J., Mitchell, T.~J., and Wynn, H.~P. (1989).
\newblock {Design and analysis of computer experiments}.
\newblock {\em Statistical Science}, 4(4):409--435.

\bibitem[Tanaka et~al., 2011]{Tanaka2011}
Tanaka, S., Bunya, S., Westerink, J.~J., Dawson, C., and Luettich, R.~A.
  (2011).
\newblock Scalability of an unstructured grid continuous {Galerkin} based
  hurricane storm surge model.
\newblock {\em Journal of Scientific Computing}, 46(3):329--358.

\bibitem[Toro et~al., 2010]{Toro2010}
Toro, G., Niedoroda, A., Reed, C., and Divoky, D. (2010).
\newblock Quadrature-based approach for the efficient evaluation of surge
  hazard.
\newblock {\em Ocean Engineering}, 37(1):114 -- 124.
\newblock A Forensic Analysis of Hurricane Katrina's Impact: Methods and
  Findings.

\bibitem[Towns et~al., 2014]{xsede}
Towns, J., Cockerill, T., Dahan, M., Foster, I., Gaither, K., Grimshaw, A.,
  Hazlewood, V., Lathrop, S., Lifka, D., Peterson, G.~D., Roskies, R., Scott,
  J.~R., and Wilkins-Diehr, N. (2014).
\newblock {XSEDE}: Accelerating scientific discovery.
\newblock {\em Computing in Science \& Engineering}, 16(5):62--74.

\bibitem[Wamsley et~al., 2013]{Wamsley2013}
Wamsley, T., Godsey, E., Bunch, B.~W., Chapman, R.~S., Gravens, M.~B.,
  Grzegorzewski, A.~S., Johnson, B.~D., King, D.~B., Permenter, R.~L., and
  Tillman, D.~L. (2013).
\newblock {Mississippi} coastal improvement program, barrier island restoration
  numerical modeling.

\bibitem[Wei and Tanner, 1990]{Wei1990}
Wei, G. C.~G. and Tanner, M.~A. (1990).
\newblock A {Monte Carlo} implementation of the {EM} algorithm and the poor
  man's data augmentation algorithms.
\newblock {\em Journal of the American Statistical Association},
  85(411):699--704.

\bibitem[Westerink et~al., 2008]{westerink2008basin}
Westerink, J.~J., Luettich, R.~A., Feyen, J.~C., Atkinson, J.~H., Dawson, C.,
  Roberts, H.~J., Powell, M.~D., Dunion, J.~P., Kubatko, E.~J., and Pourtaheri,
  H. (2008).
\newblock A basin-to channel-scale unstructured grid hurricane storm surge
  model applied to southern {Louisiana}.
\newblock {\em Monthly weather review}, 136(3):833--864.

\bibitem[Yin et~al., 2016]{yin2016coupled}
Yin, J., Lin, N., and Yu, D. (2016).
\newblock Coupled modeling of storm surge and coastal inundation: A case study
  in {New York City} during {Hurricane Sandy}.
\newblock {\em Water Resources Research}, 52(11):8685--8699.

\bibitem[Zijlema, 2010]{Zijlema2010}
Zijlema, M. (2010).
\newblock Computation of wind-wave spectra in coastal waters with {SWAN} on
  unstructured grids.
\newblock {\em Coastal Engineering}, 57(3):267 -- 277.

\end{thebibliography}
\end{singlespace}

\clearpage\pagebreak\newpage
\thispagestyle{empty}
\begin{center}
{\LARGE{\bf Web-based supporting materials for {Multifidelity Computer Model Emulation with High-Dimensional Output: An Application to Storm Surge}} } 
\vspace{1cm}
\\ \LARGE {\bf by}
\end{center}

\vskip 1cm
\baselineskip=15pt
\begin{center}
  Pulong Ma
\\
  {\it Clemson University, Clemson, SC, USA}\\
  Georgios Karagiannis \\
  {\it Durham University, Durham, UK}\\
  Bledar A. Konomi\\
  {\it University of Cincinnati, Cincinnati, OH, USA}\\
  Taylor G. Asher \\
  {\it University of North Carolina at Chapel Hill, Chapel Hill, NC, USA}\\
  Gabriel R. Toro \\
  {\it Lettis Consultants International, Inc., USA} \\
  and Andrew T. Cox\\
  {\it Oceanweather, Inc., USA}
  \hskip 5mm\\
\end{center}

\setcounter{equation}{0}
\setcounter{page}{0}
\setcounter{table}{0}
\setcounter{section}{0}
\setcounter{figure}{0}
\clearpage\pagebreak\newpage

\counterwithout{equation}{section}

\renewcommand{\theequation}{S.\arabic{equation}}
\renewcommand{\thesection}{S.\arabic{section}}
\renewcommand{\thesubsection}{S.\arabic{section}.\arabic{subsection}}
\renewcommand{\thepage}{S.\arabic{page}}
\renewcommand{\thetable}{S.\arabic{table}}
\renewcommand{\thefigure}{S.\arabic{figure}}

\section{Nodes in ADCIRC and Selected Points in Cape Coral} \label{sec: nodes}
In the FEMA coastal flood hazard study in the Southwest Florida region, \cite{FEMA2017} used an ADCIRC model to predict storm surges and then computed the annual exceedance probability (AEP) at multiple locations in order to quantify the storm surge hazard. Precisely, AEP was defined as the probability of storm surge exceeding a given flood level across the entire space of storm characteristics at a spatial location over a period of time. The flood levels of interest in coastal flood study are the $10-$, $2-$, $1-$, and $0.2$-percent-annual-chance of exceedance. They  are often used as  annual chance in many FEMA reports and documents in place of the annual chance of exceedance. The calculation of this annual chance requires thousands of simulations from an ADCIRC model, and hence are computationally prohibitive. The methodology developed in the paper provides an approach to alleviate such computational challenges for coastal flood hazard study. 

In the FEMA coastal flood hazard study in the Southwest Florida region \citep{FEMA2017}, ADCIRC is run on a mesh with $148,055$ nodes (spatial points), where numerical computation of the underlying physical models is done and output is returned. This mesh is made up of unstructured triangular grids spanning across coastal areas including Gulf of Mexico and Southwest coast of the United States; see Panel (a) of Figure~\ref{fig: mesh}. Previous study \citep[e.g.,][]{Luettich2004, Dietrich2012, FEMA2017} shows that ADCIRC can only generate accurate output in Southwest Florida region when its output is validated against historical storm data. In this paper, we focus on a subregion of Southwest Florida in Cape Coral. After eliminating spatial locations in the deep ocean and those extending inland (far away from the coastline), we primarily focus on $9,284$ spatial locations shown in the Panel (b) of Figure~\ref{fig: mesh}.

\begin{figure}[htbp]
\begin{subfigure}{.5\textwidth}
{ \includegraphics[width=1.0\textwidth, height=0.25\textheight]{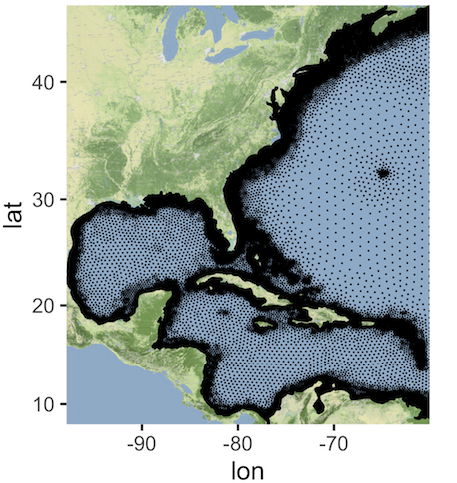}}
\caption{Full nodes of ADCIRC.}
\end{subfigure}%
\begin{subfigure}{.5\textwidth}
  \centering
  \includegraphics[width=1\linewidth, height=0.26\textheight]{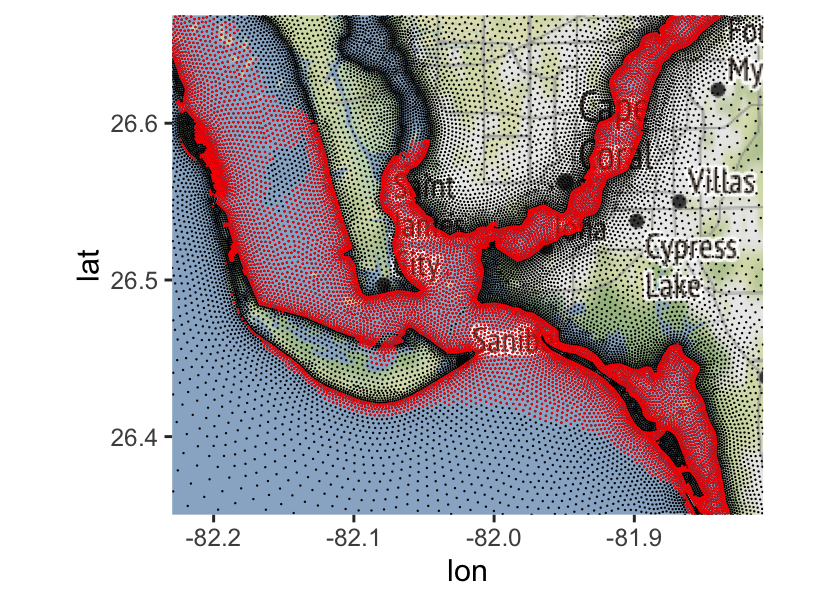}
  \caption{Nodes in Cape Coral.}
\end{subfigure}
\caption{On the left panel: ADCIRC generated output over 148,055 nodes. On the right panel:  the zoom-in view of these nodes in the Cape Coral region with selected 9,284 spatial locations highlighted in red.}
\label{fig: mesh}
\end{figure}

\section{Toy Example with Univariate Output} \label{sec: toy example}

In this section, we illustrate the proposed method with a toy example for univariate output. To demonstrate the cokriging model with the proposed Bayesian estimation method, we adapt an example from \cite{Gratiet2013} and consider a low-fidelity code
$y_1(x) = 0.5(6x-2)^2\sin(12x-4) + 10(x-0.5) - 5,
$ and a high-fidelity code
$y_2(x) = 2y_1(x) - 20x + 20 + { \sin(10\cos(5x))}$, where a constant scale discrepancy and a nonlinear location discrepancy are assumed. The experiment is setup as follows. We choose 20 design points in $[-1, 1]$ with spacing 0.1 for the low-fidelity code. For the high-fidelity code, we consider 10 design points $\{-1, -0.8, -0.55, -0.4, -0.2, 0, 0.2, 0.4, 0.6, 1\}$ for non-nested design. This gives a non-nested design with two points $\{-0.55, -0.2\}$ that are not in the low-fidelity level code. For predictive purpose, we choose 200 points uniformly spaced in the domain $[-1, 1]$. 

We perform the typical kriging with high-fidelity runs using the package \texttt{RobustGaSP} \citep{Gu2018RobustGP}; see the result in the left panel of Figure~\ref{fig: univariate example 1}. Then we use the proposed method to perform autoregressive cokriging using the package \texttt{ARCokrig}. The cokriging result is shown on the right panel of Figure~\ref{fig: univariate example 1}. By combining information from both low-fidelity code and high-fidelity code, the autoregressive cokriging model gives much better result than kriging itself and can approximate the high-fidelity code more accurately.

\begin{figure}[htbp]
\begin{center}
\makebox[\textwidth][c]{ \includegraphics[width=1.0\textwidth, height=0.3\textheight]{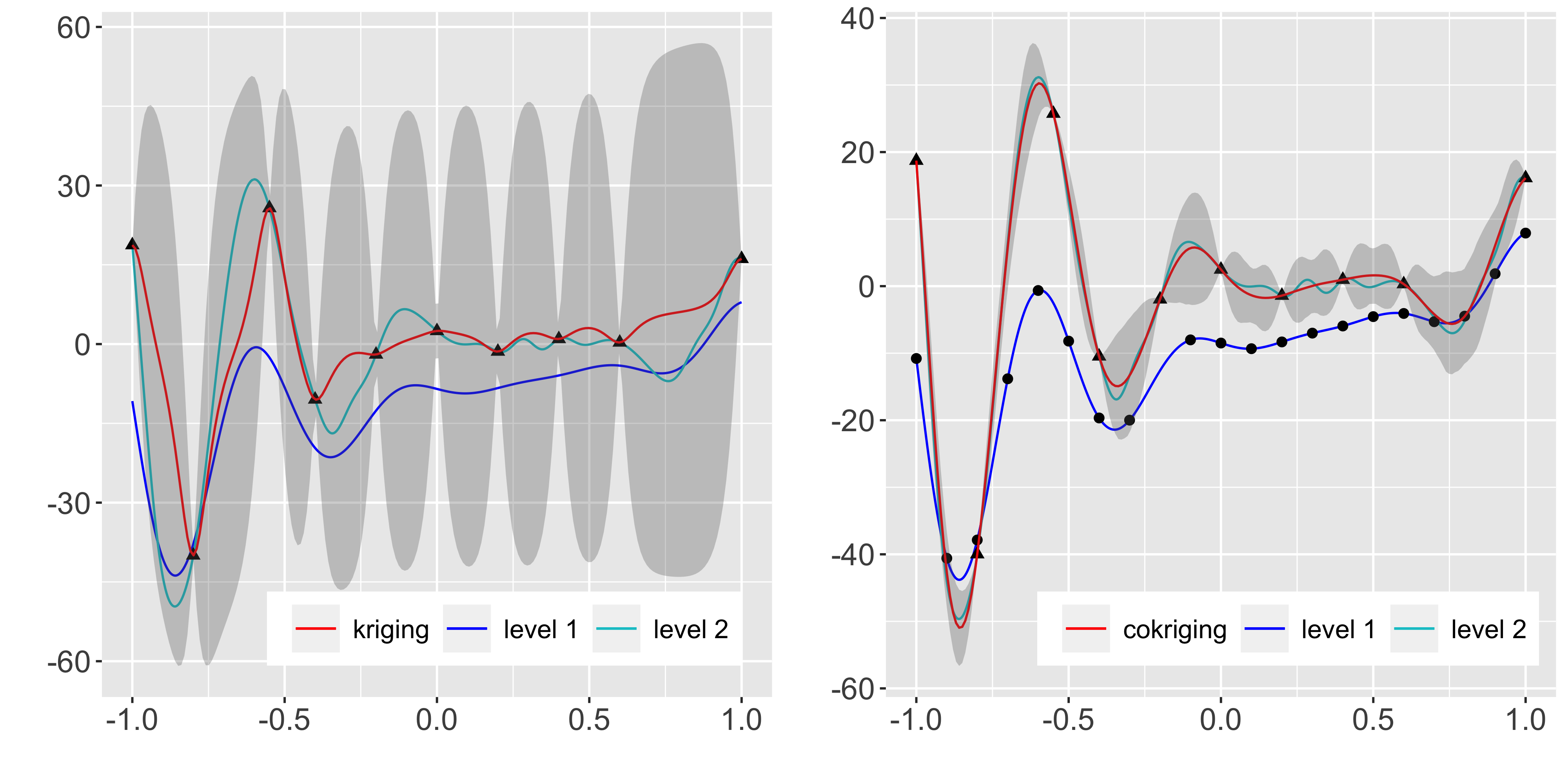}}
\caption{Illustration of the cokriging model under a non-nested design. The triangles represent design points for the high-fidelity code (level 2 code). The dots represent design points for the low-fidelity code (level 1 code). The left panel shows the kriging result based on high-fidelity runs, and the right panel shows the cokriging result based on both low-fidelity runs and high-fidelity runs. The bounds in gray areas are 95\% credible intervals.
}
\label{fig: univariate example 1}
\end{center}
\end{figure}

When the design is not nested, the Monte Carlo EM algorithm in Section~\ref{sec: MCEM} will be used to estimate correlation parameters $\bfphi$. Then the missing data $\mathring{\bfy}^{\mathscr{D}}$ can be estimated based on the posterior $\pi(\mathring{\bfy}^{\mathscr{D}}\mid \bfy^{\mathscr{D}},\hat{\bfphi})$ in Section~\ref{app: distributions} of the Supplementary Material, where $\hat{\bfphi}$ denotes the estimated correlation parameters. The prediction can be obtained by directly generating Monte Carlo samples from the posterior $\pi(\mathring{\bfy}^{\mathscr{D}}\mid \bfy^{\mathscr{D}},\hat{\bfphi})$.  To study the impact of the number of Monte Carlo samples on parameter estimation in the MCEM algorithm and prediction under non-nested design, we also investigated the performance of prediction based on different Monte Carlo sample sizes. In the first situation, we vary the Monte Carlo sample size $M$ in the MCEM algorithm and prediction; in the second situation, we fixed the Monte Carlo sample size $M$ to be 30 to estimate correlation parameters in the MCEM algorithm. Then we generate Monte Carlo samples from the predictive distribution with Monte Carlo sample size varying from $M=10$ to $M=100$. In these two situations, we compared the RMSPEs over 200 input values for the high-fidelity code with results shown in Figure~\ref{fig: MC error}. Clearly, we could see that as $M$ is greater than 25, the RMSPE becomes stable in these two situations. This indicates that the MCEM algorithm shows convergence and it is safe to get prediction based on more than 25 Monte Carlo samples from the predictive distribution. 

\begin{figure}[htbp]
\begin{center}
\makebox[\textwidth][c]{ \includegraphics[width=1.0\textwidth, height=0.3\textheight]{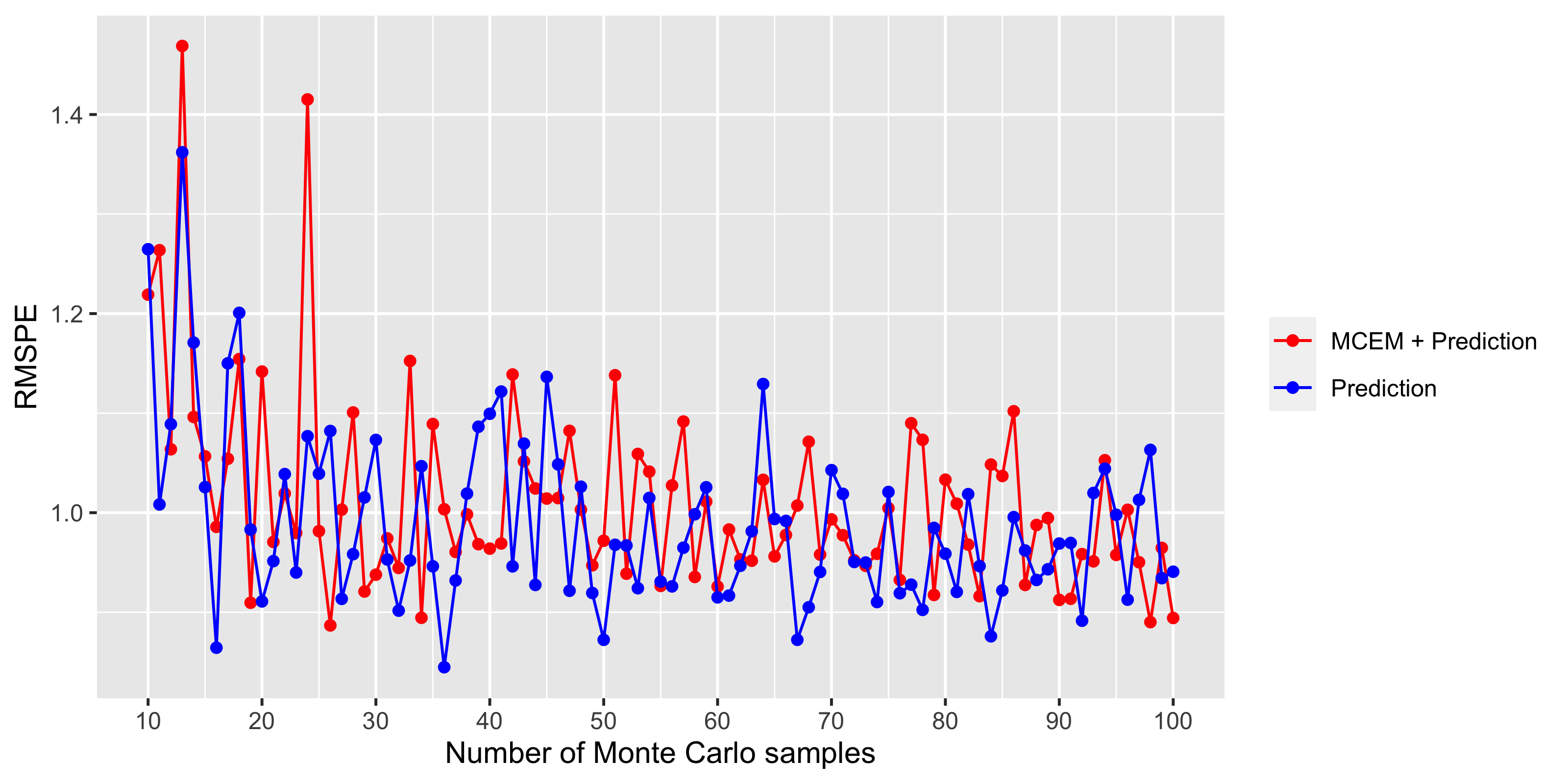}}
\caption{Prediction performance with different Monte Carlo sample sizes under a non-nested design. The horizontal axis represents the Monte Carlo sample size varying from 10 to 100 with increment 1. The vertical axis represents the RMSPE over 200 input values for the high-fidelity code. The blue curve shows the RMSPEs by keeping the Monte Carlo sample size to be the same in the MCEM algorithm and the predictive distribution, where the Monte Carlo sample size varies from 10 to 100. The red curve shows the RMSPEs by fixing the Monte Carlo sample size at 30 in the MCEM algorithm and varying the Monte Carlo sample size in the predictive distribution.
}
\label{fig: MC error}
\end{center}
\end{figure}

\section{Artificial Example with Functional Output} \label{sec: Example with functional output}
In this section, we illustrate the proposed PP cokriging method with an artificial example, where the output is functional. This example is created based on the univariate example in Section~\ref{sec: toy example}. We consider the following low-fidelity code: 
\begin{align}\label{eqn: low}
y_1(x, t) = 0.5(6x-2)^2\sin(12x-4) + 10(x-0.5) - 5 + xt^3 + 2t\exp(-t),
\end{align}
and the high-fidelity code 
\begin{align}\label{eqn: high}
y_2(x, t) = 2 y_1(x,t) - 20x + 20 + \sin(10\cos(5x)) t \cos(t),
\end{align}
where the variable $x$ is assumed to be a physical input parameter and $t$ is considered to be a time point. The goal here is to emulate the functional output from the high-fidelity code for any given input variable $x$. As an illustration, we use the same input design as in Section~\ref{sec: toy example}, and consider the output over 30 time points in $[0, 2]$. Figure~\ref{fig: functional output} indicates that the proposed PP cokriging method gives reasonable results with predictive mean aligned along the 45-degree line. In addition, the RMSPE is 0.4026 and CVG(95\%) is 0.987, and the NSME is 0.999. These numerical measures also suggest that the proposed PP cokriging method gives reasonable prediction results for predicting the high-fidelity code. 

\begin{figure}[htbp]
\begin{center}
\makebox[\textwidth][c]{ \includegraphics[width=1.0\textwidth, height=0.3\textheight]{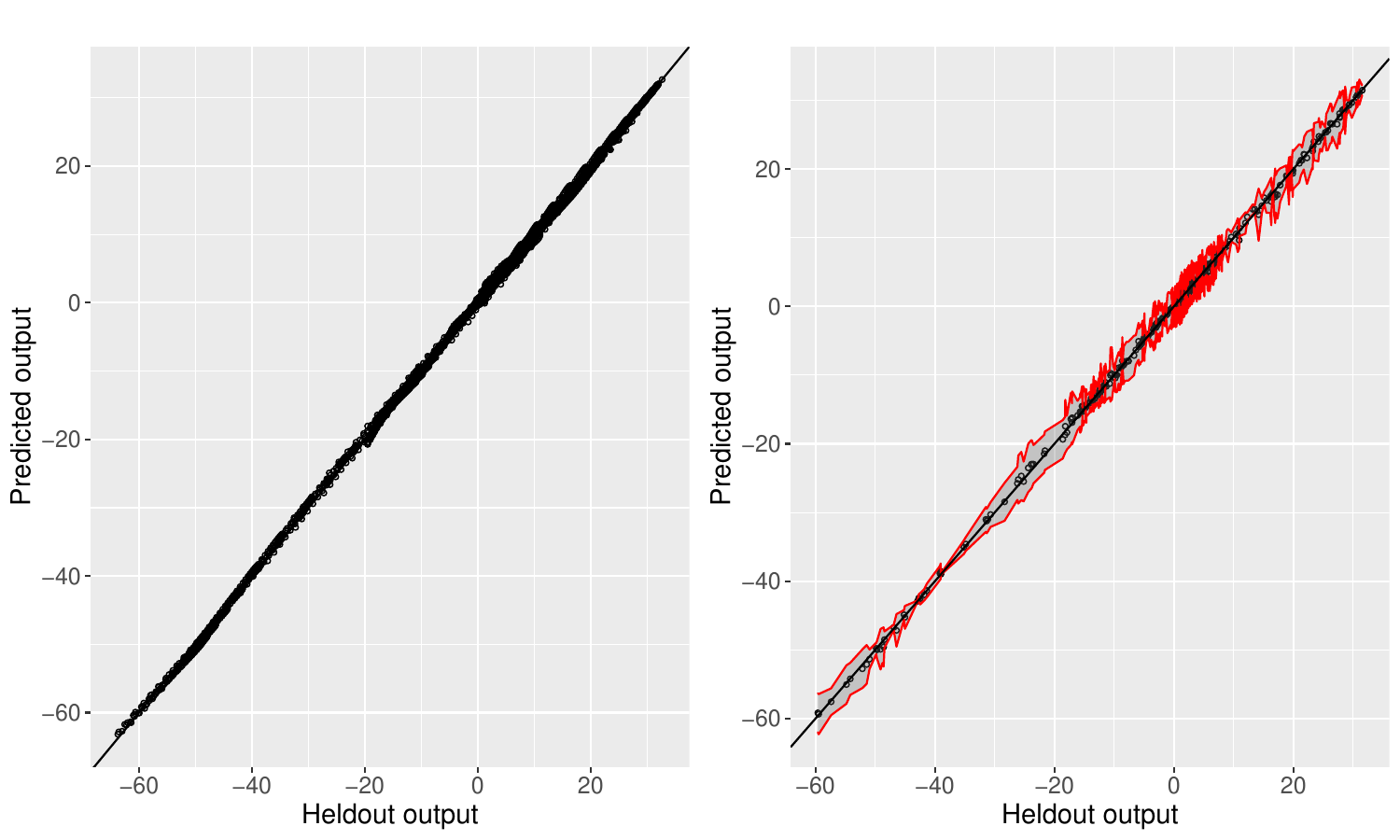}}
\caption{Diagnostics of prediction over 200 input values across 30 time points with PP cokriging. Left panel shows the predictive mean of the high-fidelity code versus the held-out outputs from the high-fidelity code over 200 input values across 30 time points. Right panel shows the predictive mean, 95\% percentile predictive intervals for the high-fidelity code versus the held-out outputs from the high-fidelity code over randomly selected 500 points out of 6000 ($=200\times 30$) points. 
}
\label{fig: functional output}
\end{center}
\end{figure}

\section{The Monte Carlo Expectation-Maximization Algorithm} \label{sec: MCEM}
In the augmented marginal distribution~\eqref{eqn: aug marginal dist}, its conditional distributions are given by  $\pi(\tilde{\bfy}_{t,j}\mid\bfphi_{t},\tilde{\bfy}_{t-1,j}) \propto |\tilde{\bfR}_t|^{-1/2} |\tilde{\bfT}_{t,j}^\top \tilde{\bfR}_t^{-1} \tilde{\bfT}_{t,j}|^{-1/2} \{ {S}^2(\bfphi_t, \tilde{\bfy}_{t,j})\}^{-(\tilde{n}_t - q_t)/2}$, where $\tilde{\bfT}_{1,j}=\tilde{\bfH}_1$ for $j=1, \ldots, N$ and $\tilde{\bfT}_{t,j}=[\tilde{\bfH}_t, \tilde{W}_{t-1,j}]$ for $t=2, \ldots, s$ and $j=1, \ldots, N$. ${S}^2(\bfphi_t, \tilde{\bfy}_{t,j}):=\tilde{\bfy}_{t,j}^\top\tilde{\bfQ}_t\tilde{\bfy}_{t,j}$ with $\tilde{\bfQ}_t:=\tilde{\bfR}_t^{-1}\{\mathbf{I} - \tilde{\bfT}_{t,j}$ $(\tilde{\bfT}_{t,j}^\top \tilde{\bfR}_t^{-1} \tilde{\bfT}_{t,j})^{-1} \tilde{\bfT}_{t,j}^\top \tilde{\bfR}_t^{-1}\}$.

In the EM algorithm \citep{Dempster1977}, we treat $\mathring{\bfy}^{\mathscr{D}}$ as ``missing data'' and $\{\mathring{\bfy}^{\mathscr{D}},\bfy^{\mathscr{D}}\}$ as the complete data. Let $\bfphi^{[\ell]}$ be the parameters in the $\ell$th
iteration of the EM algorithm. The EM algorithm consists of two steps. The first step is to compute the so-called $Q$-function based on a complete-data-log-likelihood $\ln\pi(\tilde{\bfy}^{\mathscr{D}},\bfphi)$ in the E-step, which is given by 
\begin{align}
\begin{split}\ln\pi(\tilde{\bfy}^{\mathscr{D}},\bfphi) 
 & =\sum_{t=1}^{s}g_{t}(\bfphi_{t},(\mathring{\bfy}_t)_{j=1}^N),
\end{split}
\end{align}
where $(\mathring{\bfy}_t)_{j=1}^N: = (\mathring{\bfy}_{t,1}, \ldots, \mathring{\bfy}_{t,N})^\top$ and   $g_{t}(\bfphi_{t},(\mathring{\bfy}_t)_{j=1}^N):= \ln\pi(\bfphi_t)- \sum_{j=1}^N\ln|\tilde{\bfT}_{t,j}^\top \tilde{\bfR}_t^{-1} \tilde{\bfT}_{t,j}|/2$ $-N\ln|\tilde{\bfR}_t|/2- (\tilde{n}_t-q_t) \sum_{j=1}^N\ln {S}^2(\bfphi_t, \tilde{\bfy}_{t,j})/2 .$

Starting with initial values $\bfphi^{[\ell]}$ at the $\ell$-th iteration, we compute the $Q$-function in the E-step:
$Q(\bfphi;\bfphi^{[\ell]}) =E_{\pi(\mathring{\bfy}^{\mathscr{D}}\mid \bfy^{\mathscr{D}},\bfphi^{[\ell]})}\{\ln\pi(\tilde{\bfy}^{\mathscr{D}},\bfphi)\}
$, where the conditional distribution $\pi(\mathring{\bfy}^{\mathscr{D}}\mid \bfy^{\mathscr{D}},\bfphi^{[\ell]})$ is given in Section~\ref{app: distributions} of the Supplementary Material. As this expectation cannot be computed analytically, Monte Carlo samples from the distribution $\pi(\mathring{\bfy}^{\mathscr{D}}\mid \bfy^{\mathscr{D}},\bfphi^{[\ell]})$ can be used to approximate this expectation. Specifically, the $Q$-function is 
\begin{align}
\begin{split}
Q(\bfphi;\bfphi^{[\ell]})  
 & \approx \underset{=\hat{Q}_{t,M}(\bfphi_t\mid\bfphi_t^{[\ell]})}{\sum_{t=1}^{s}\underbrace{\frac{1}{M}\sum_{k=1}^{M}g_{t}(\bfphi_{t}, (\mathring{\bfy}_t^{[k]})_{j=1}^N)}}, 
\end{split}
\end{align}
where $\{(\mathring{\bfy}_t^{[k]})_{j=1}^N:t=1,\ldots,s\}$ is a sample from the distribution $\pi(\mathring{\bfy}^{\mathscr{D}}\mid \bfy^{\mathscr{D}},\bfphi^{[\ell]})$. The second step is to numerically maximize this function with respect to parameters $\bfphi_t$ for $t=1, \ldots, s$. This leads to the so-called Monte Carlo EM (MCEM) algorithm \citep{Wei1990}.

In each iteration of the MCEM algorithm, the computation of the function $g(\bfphi_t, (\mathring{\bfy}_t)_{j=1}^N)$ has computational cost $O(N\tilde{n}_t^3)$. Lemma~\ref{lem: determinant} in Section~\ref{app: useful results} of the Supplementary Material shows that the computational complexity of this function can be reduced to $O(N\tilde{n}_t^2 + \tilde{n}_t^3)$, which substantially reduces computational cost when $N$ is very large as in our real application. In addition, parameters $\{\bfphi_t: t=1, \ldots, s\}$ can be estimated independently in the M-step. This allows great advantage in estimating model parameters with the proposed methodology. As the number of simulations is limited for high-fidelity simulators in real applications, the proposed cokriging model is anticipated to produce more stable estimates than independent kriging models. The procedure of the MCEM algorithm is given in Algorithm~\ref{alg: MCEM in mult}.

\begin{algorithm}
\caption{The MCEM algorithm}
\label{alg: MCEM in mult} 
\begin{raggedright}
\textbf{Input:} Initial values $\bfphi^{[1]}$, data $\bfy$ and $\ell=1$. \\
 \textbf{Output:} Values $\bfphi$ that maximize $\pi(\bfphi\mid \bfy)$. 
\par\end{raggedright}
 \begin{algorithmic}[1] \Repeat 
 \State Generate $M$ samples
from $\pi(\mathring{\bfy}^{\mathscr{D}}\mid \bfy^{\mathscr{D}},\bfphi^{[\ell]})$,
which are denoted by $\{(\mathring{\bfy}_t^{[k]})_{j=1}^N:t=1,\ldots,s; k=1, \ldots, M\}$.
\For{$t=1,\ldots,s$} \Comment{can be run in parallel} \label{alg:step-generatemissingdata}
\State
E-step: 
\begin{eqnarray*}
\hat{Q}_{t,M}(\bfphi_{t}\mid\bfphi_{t}^{[\ell]})=\frac{1}{M}\sum_{k=1}^{M}g_{t}(\bfphi_{t}, (\mathring{\bfy}_t^{[k]})_{j=1}^N).
\end{eqnarray*}
\State M-step: 
\begin{eqnarray*}
\bfphi_{t}^{[\ell+1]}:=\arg\max_{\bfphi_{t}}\hat{Q}_{t,M}(\bfphi_{t}\mid\bfphi_{t}^{[\ell]}).
\end{eqnarray*}
\EndFor \Until{certain stopping convergence criterion is satisfied.}
\end{algorithmic} 
\end{algorithm}

\section{One-Step Prediction} \label{app: pred}
In what follows, we derive the predictive distribution based on the same idea used in \cite{Kennedy2000, Gratiet2013}. It is easy to show that 
\begin{align*}
    \pi(\bfy_{s}(\bfx_0) \mid \bfy^{\mathscr{D}}, \bfphi) 
    & = \prod_{j=1}^N \int \pi(y_{s,j}(\bfx_0) \mid \bfy^{\mathscr{D}}_j, \bfbeta_j, \bfgamma_j, \bfsigma^2_j, \bfphi) \pi(\bfbeta_j, \bfgamma_j, \bfsigma^2_j \mid \bfy^{\mathscr{D}}_j)\\
    &\quad \times d\{\bfbeta_j, \bfgamma_j, \bfsigma^2_j \}.
\end{align*}
Notice that 
\begin{align*}
\begin{split}
    \pi(y_{s,j}(\bfx_0)\mid \bfy^{\mathscr{D}}_j, \bfbeta_j, \bfgamma_j, \bfsigma^2_j, \bfphi) 
    &= \int \pi(y_{s,j}(\bfx_0) \mid \tilde{\bfy}^{\mathscr{D}}_j, \bfbeta_j, \bfgamma_j, \bfsigma^2_j, \bfphi) \\
    &\quad \times \pi(\mathring{\bfy}^{\mathscr{D}}_j \mid \bfy^{\mathscr{D}}_j, \bfbeta_j, \bfgamma_j, \bfsigma^2_j, \bfphi)  \, d\{ \mathring{\bfy}^{\mathscr{D}}_j\}, 
\end{split}
\end{align*}
where $\pi(\mathring{\bfy}^{\mathscr{D}}_j \mid \bfy^{\mathscr{D}}_j, \bfbeta_j, \bfgamma_j, \bfsigma^2_j, \bfphi)$ is a multivariate normal distribution with mean and variance given in Section~\ref{app: distributions} of the Supplementary Material.  $\pi(y_{s,j}(\bfx_0) \mid \tilde{\bfy}_j, \bfbeta_j, \bfgamma_j, \bfsigma^2_j, \bfphi)$ is a normal distribution with mean $m_{s,j}(\bfx_0)$ and variance $v_{s,j}(\bfx_0)$ given in Proposition~\ref{lem: one-step prediction}. 

\begin{proposition}[{\textbf{One-Step Prediction}}] \label{lem: one-step prediction}
The mean $m_{s,j}(\bfx_0)$ and variance $v_{s,j}(\bfx_0)$ in the conditional predictive distribution $\pi(y_{s,j}(\bfx_0) \mid \tilde{\bfy}^{\mathscr{D}}_j, \bfbeta_j, \bfgamma_j,$ $\bfsigma^2_j, \bfphi)$ are given by 
\begin{align*} 
    \begin{split}
    m_{s,j}(\bfx_0) &= \bff^\top_j(\bfx_0) {\bfbeta}_j + \bfc^\top_j(\bfx_0) \bigl(\bfSigma^{j}\bigr)^{-1}(\bfy^{\mathscr{D}}_j - \tilde{\bfF}^j {\bfbeta}_j), \\ 
    v_{s,j}(\bfx_0) &=\sigma^2_{s,j} r(\bfx_0, \bfx_0|\bfphi_s) - \bfc^\top_j(\bfx_0) \bigl(\bfSigma^{j}\bigr)^{-1} \bfc_j(\bfx_0),
    \end{split}
\end{align*}
where the vector $\bff_j(\bfx_0)$ is given by 
\begin{align*}
    \bff_j(\bfx_0) &:= \left( \left(\prod_{i=1}^{s-1} \gamma_{i,j} \right) \tilde{\bfh}_1^\top(\bfx_0), \left(\prod_{i=2}^{s-1} \gamma_{i,j} \right) \tilde{\bfh}_2^\top(\bfx_0), \ldots, \gamma_{s-1,j} \tilde{\bfh}_{s-1}^\top(\bfx_0), \right.\\
    &\quad\quad \left. \tilde{\bfh}_s^\top(\bfx_0) \right)^\top.
\end{align*}
$\bfSigma^j$ is an $s$-by-$s$ block matrix with the $(i,i')$-th block of size $\tilde{n}_t$-by-$\tilde{n}_t$ given by 
\begin{align*}
\bfSigma_{i,i}^j(\tilde{\calX}_t, \tilde{\calX}_t) &:= 
\begin{cases}
\sum_{k=1}^t \left(\prod_{\ell=k}^{t-1}\gamma^2_{\ell,j}\right) \sigma^2_{k,j} r(\tilde{\calX}_t, \tilde{\calX}_t\mid \bfphi_k), & \text{ for } i=i';  \\
\left(\prod_{k=i}^{i'-1}\gamma_{k,j}\right) \bfSigma_{i,i}^j(\tilde{\calX}_i, \tilde{\calX}_i'), & \text{ for } i< i'. 
\end{cases}
\end{align*}
The vector $\bfc_j(\bfx_0)$ is $\bfc_j(\bfx_0):=(\bfc_{1,j}(\bfx_0, \tilde{\calX}_1), \ldots, \bfc_{s,j}(\bfx_0, \tilde{\calX}_s))^\top$ with 
\begin{align*}
    \bfc_{t,j}(\bfx_0, \tilde{\calX}_t) = \gamma_{t-1,j}\bfc_{t-1,j}(\bfx_0, \tilde{\calX}_t) + \left(\prod_{k=t}^{s-1}\gamma_{k,j} \right) \sigma^2_{t,j} r(\bfx_0, \tilde{\calX}_t|\bfphi_t),\, t=2, \ldots, s,
\end{align*}
where $\prod_{k=s}^{s-1}\gamma_{k,j}:=1$ for all $j$'s and $\bfc_{1,j}(\bfx_0, \tilde{\calX}_1) :=\left(\prod_{k=t}^{s-1}\gamma_{k,j}\right) \sigma^2_{1,j} r(\bfx_0, \tilde{\calX}_1|$ $\bfphi_1)$. The $(i,k)$th block in the matrix $\tilde{\bfF}^j$ is $\tilde{\bfF}_{i,k}^j:=\prod_{t=k}^{i-1}\gamma_{t,j} \bfh_k(\tilde{\calX}_i)$ for $i\geq k$, and $\tilde{\bfF}_{i,k}^j:=\mathbf{0}$ for $i<k$. 
\end{proposition}
\begin{proof}
In the univariate setting, this formula is the same as the one in \cite{Kennedy2000, Gratiet2013, Qian2008}. As the predictive distribution is independent given the range parameters, the result follows straightforwardly for spatial coordinate $j$.
\end{proof}

The one-step prediction formula in Proposition~\ref{lem: one-step prediction} requires $O(N(\sum_{t=1}^s \tilde{n}_t)^3)$ flops to obtain predictions over all $N$ spatial locations. It is worth pointing  out that integrating out model parameters $\{\bfgamma_j, \bfsigma^2_j: j=1, \ldots, N\}$ requires Monte Carlo approximations for $N$ spatial locations. This is computationally demanding if we want to account for uncertainties in these model parameters when the number of spatial locations is large as in the storm surge application. Of course, one may generate prediction by plugging the posterior estimates of these model parameters into the predictive distribution, however, doing so would not allow us to account for these model parameters across all the spatial locations. This will underestimate predictive uncertainty. So, the one-step prediction formula should be avoid in practice.

\section{Distributions} \label{app: distributions}
It is straightforward to show that the conditional distribution $\pi(\mathring{\bfy}_j^{\mathscr{D}}\mid \bfy_j^{\mathscr{D}}, \bfbeta_j, \bfgamma_j, \bfsigma^2_j, \bfphi)$
can be factorized as 
\begin{align*} 
\pi(\mathring{\bfy}_j^{\mathscr{D}}\mid \bfy_j^{\mathscr{D}}, \bfbeta_j, \bfgamma_j, \bfsigma^2_j, \bfphi)&=\pi(\mathring{\bfy}_{1}^{\mathscr{D}}\mid \bfy_{1}^{\mathscr{D}},\bfbeta_{1,j}, \sigma^2_{1,j}, \bfphi_{1})\prod_{t=2}^{s-1}\pi(\mathring{\bfy}_{t}^{\mathscr{D}}\mid\tilde{\bfy}_{t-1}^{\mathscr{D}},\bfy_{t}^{\mathscr{D}}, \\
&\quad \bfbeta_{t,j}, \gamma_{t-1,j}, \sigma^2_{t,j}, \bfphi_{t}).
\end{align*}
The conditional distributions are multivariate normal distributions with means and variances given by 
\begin{align*}
\begin{split}
\mu_{\mathring{y}|y}^{t,j} & :=\mathring{\bfT}_{t,j} {\bfb}_{t,j}+r(\mathring{\calX}_{t},{\calX}_{t}\mid\bfphi_{t})\bfR_t^{-1}(\bfy_{t,j}-\bfT_{t,j} {\bfb}_{t,j}),\\
\Sigma_{\mathring{y}|y}^{t,j} & :=\sigma^2_{t,j} \{r(\mathring{\calX}_{t},\mathring{\calX}_{t}\mid\bfphi_{t})-r(\mathring{\calX}_{t},{\calX}_{t}\mid\bfphi_{t})\bfR_t^{-1}r(\calX_{t},\mathring{\calX}_{t}\mid\bfphi_{t})\},
\end{split}
\end{align*}
where $\mathring{\bfT}_{1,j}:=\mathring{\bfH}_{1,j}$, $\mathring{\bfT}_{t,j}:=[\mathring{\bfH}_{1,j}, y_{t-1,j}(\mathring{\mathcal{X}}_t)]$ for $t>1$,  $\bfb_{1,j} = \bfbeta_1$, and $\bfb_{t,j} = (\bfbeta_{t,j}^\top, \gamma_{t-1,j})^\top$ with $t=2, \ldots, s$.

The joint distribution of $\tilde{\bfy}^{\mathscr{D}}$ and $\bfphi$ can be obtained via $\pi(\tilde{\bfy}^{\mathscr{D}},\bfphi)=\pi(\tilde{\bfy}^{\mathscr{D}}\mid\bfphi)\pi(\bfphi)$. Suppose that the missing data $\mathring{\bfy}_t$ only depends on $\bfy_t$ and missing data $\mathring{\bfy}_{t-1}$. It follows that, the conditional distribution $\pi(\mathring{\bfy}^{\mathscr{D}}\mid \bfy^{\mathscr{D}},\bfphi)$
is given by 
\begin{align*}
\begin{split}
\pi(\mathring{\bfy}^{\mathscr{D}}\mid \bfy^{\mathscr{D}},\bfphi) 
= \prod_{j=1}^N \pi(\mathring{\bfy}_{1,j}\mid \bfy_{1,j},\bfphi_{1})\prod_{t=2}^{s-1}\pi(\mathring{\bfy}_{t,j}\mid\tilde{\bfy}_{t-1,j},\bfy_{t,j},\bfphi_{t}),
\end{split}
\end{align*}
where the conditional distributions on the right-hand side are $\mathring{n}_{t}$-dimensional Student-$t$ distribution for $t=1, \ldots, s-1$, whose degrees of freedom is $n_{t} - q_t$, location parameters $\bfmu_{t,j}$ and scale parameters $V_{t,j}$ are 
\begin{align*}
\begin{split}
\bfmu_{t,j} & :=\mathring{\bfT}_t\hat{\bfb}_{t,j} +r(\mathring{\calX}_{t},{\calX}_{t}\mid\bfphi_{t})\bfR_t^{-1}(\bfy_{t,j}-\bfT_{t,j}\hat{\bfb}_{t,j}),\\
V_{t,j} & :=\frac{S^2(\bfphi_t, \bfy_{t,j})}{n_{t}-q_t}\Sigma_{\mathring{y}\mathring{y}|y}^{t},
\end{split}
\end{align*}
with 
\begin{align*}
\begin{split}
\hat{\bfb}_{t,j} &:=(\bfT_{t,j}^\top \bfR_t^{-1} \bfT_{t,j})^{-1} \bfT_{t,j}^\top \bfR_t^{-1}\bfy_{t,j} \\
S^2(\bfphi_t, \bfy_{t,j}) & := \bfy_{t,j}^\top \bfQ_t \bfy_{t,j}\\
\Sigma_{\mathring{y}\mathring{y}|y}^{t} & :=\mathring{\bfR}_t-\bfr_t^\top(\mathring{\calX}_t)\bfR_t^{-1}\bfr_t(\mathring{\calX}_{t}) 
+ [\mathring{\bfT}_{t,j}^{\top}-\bfT_{t,j}^{\top}\bfR_{t}^{-1} \bfr_t(\mathring{\calX}_t)]^{\top} \\
&\cdot (\bfT_{t,j}^{\top}\bfR_{t}^{-1}\bfT_{t,j})^{-1} [\mathring{\bfT}_{t,j}^{\top}-\bfT_{t,j}^{\top}\bfR_{t}^{-1} \bfr_t(\mathring{\calX}_t)],
\end{split}
\end{align*}
where $\bfR_t:=r(\calX_t, \calX_t \mid \bfphi_t)$, $\bfr_t(\mathring{\calX}_t):=r({\calX}_t, \mathring{\calX}_t\mid \bfphi_t)$, and $\mathring{\bfR}_t:=r(\mathring{\calX}_t, \mathring{\calX}_t\mid \bfphi_t)$. $\bfQ_t:=\bfR_t^{-1}\{\mathbf{I} - \bfT_{t,j} (\bfT_{t,j}^\top \bfR_t^{-1} \bfT_{t,j})^{-1} \bfT_{t,j}^\top \bfR_t^{-1}\}$.

\section{PP Cokriging with Nested Design} \label{sec: nested design}
Let $\hat{y}_{t,j}(\bfx_0):=E[y_{t,j}(\bfx_0)\mid \bfy^{\mathscr{D}}, \bfphi]$ be the predictive mean and $\hat{v}_{t,j}(\bfx_0):=Var[y_{t,j}(\bfx_0)\mid \bfy^{\mathscr{D}}, \bfphi]$ be the predictive variance at $j$th coordinate and level $t$. In what follows, $\hat{\bfy}_{t,j}(\bfx_0)$ is called the cokriging predictor and $\hat{\bfv}_{t,j}(\bfx_0)$ is called the cokriging variance at new input $\bfx_0$ for $j$ coordinate and level $t$.

\begin{lemma}  \label{thm: prediction in nested design}
Suppose that $n_t-q_t>2$. For $j=1, \ldots, N$, the cokriging predictor and cokriging variance at fidelity level  $t$ are given by 
\begin{align*} 
    \begin{split}
    \hat{y}_{t,j}(\bfx_0) &= \bfT_{t,j}^\top(\bfx_0) \hat{\bfb}_{t,j} + \bfr_t^\top(\bfx_0) \bfR_t^{-1}(\bfy_t - \bfT_{t,j} \hat{\bfb}_{t,j}), \\
    \hat{v}_{t,j}(\bfx_0) &= \hat{\gamma}_{t-1,j}^2\hat{v}_{t-1,j}(\bfx_0) + \frac{n_t-q_t}{n_t-q_t-2} \hat{\sigma}^2_{t,j} \left\{r(\bfx_0, \bfx_0;\bfphi_t) - \bfr_t^\top(\bfx_0) \bfR_t^{-1} \bfr_t(\bfx_0) + \kappa_{t,j} \right\},
    \end{split}
\end{align*}
where $\bfT_{1,j}(\bfx_0):=\bfh_1(\bfx_0)$, $\bfT_{t,j}(\bfx_0):=[\bfh_t^\top(\bfx_0), \hat{y}_{t-1,j}^\top(\bfx_0)]^\top$ for $t>1$, $\bfT_{1,j}:=\bfH_1$, $\bfT_{t,j}=[\bfH_t, \bfy_{t-1,j}(\mathcal{X}_t)]$, $\hat{v}_{0,j}:=0$, $\hat{\bfb}_{t,j}:=(\bfT_{t,j}^\top \bfR_t^{-1} \bfT_{t,j})^{-1} \bfT_{t,j}^\top \bfR_t^{-1} \bfy_{t,j}$ is the generalized least square estimator for $\bfb_{t,j}:=(\bfbeta_{t,j}^\top, \gamma_{t,j})^\top$ with $\gamma_{0,j}=0$, $\hat{\sigma}_{t,j}^2:=(n_t-q_t)^{-1}(\bfy_{t,j}-\bfT_{t,j}\hat{\bfb}_{t,j})^\top\bfR_t^{-1}(\bfy_{t,j}-\bfT_{t,j}\hat{\bfb}_{t,j})$,   and 
\begin{align*}
    \kappa_{t,j} &:=  [\bfT_{t,j}(\bfx_0)-\bfT_{t,j}^{\top}\bfR_{t}^{-1}\bfr_t(\bfx_0)]^{\top}(\bfT_{t,j}^{\top}\bfR_{t}^{-1}\bfT_{t,j})^{-1}[\bfT_{t,j}(\bfx_0)-\bfT_{t,j}^{\top}\bfR_{t}^{-1}\bfr_t(\bfx_0)] \\
    &\quad + \hat{v}_{t-1,j}(\bfx_0) \left\{ y_{t-1,j}^\top(\calX_t) \bfQ^H_t y_{t-1,j}(\calX_t) \right\}^{-1}.
\end{align*}
with $\bfQ^H_t := \bfR_t^{-1} - \bfR_t^{-1} \bfH_t (\bfH_t^\top \bfR_t^{-1} \bfH_t)^{-1} \bfH_t^\top \bfR_t^{-1}$.
\end{lemma}
\begin{proof}
As independence across spatial locations is assumed, we utilize the predictive formulas independently at each spatial location and then let the input correlation matrix to be the same across each spatial location.  Thanks to Theorem 3.3 of \cite{Ma2020OBayes}, the prediction formula at each fidelity level  is available in a closed form. The results thus follow immediately.   
\end{proof}

\section{Proof of Theorem~\ref{thm: SEP}} \label{sec: proof in SEP cokriging}
\begin{proof}
We first establish the results for predictive mean and then establish the results for predictive variance. 

Let $\bfy^t(\bfx_0):=(y_{t,1}(\bfx_0), \ldots, y_{t,N}(\bfx_0))^\top$ be a vector of outputs at input $\bfx_0$ at level $t$ across $N$ spatial coordinates. 
The joint distribution of $\bfy(\bfx_0):=[\bfy_1(\bfx_0), \ldots, \bfy_N(\bfx_0)]$ and $\bfy^{\mathscr{D}}$ is a product of matrix-normal distributions across $s$ levels,
\begin{align} \label{eqn: joint MN distribution}
\begin{split}
\begin{pmatrix} 
\bfy(\bfx_0) \\
\bfy^{\mathscr{D}} \\
\end{pmatrix} 
\biggr \rvert  \bfB, \Gamma, \bfSigma, \bfphi &\sim 
\mathcal{MN}_{n_1+1, N} \Biggr(   
\begin{pmatrix} 
\bfh^\top_1(\bfx_0) \bfB_1 \\
\bfH_1 \bfB_1 \\
\end{pmatrix}, 
\begin{pmatrix} 
r(\bfx_0, \bfx_0; \bfphi_1) & \bfr^\top_1(\bfx_0)  \\
\bfr_1(\bfx_0)& \bfR_1  \\
\end{pmatrix}, 
\bfSigma_1
\Biggr) \\
&\quad \times \prod_{t=2}^s 
\mathcal{MN}_{n_t+1, N} \left(   
\begin{pmatrix} 
\bfh^\top_t(\bfx_0) \bfB_t + (\bfy^{t-1}(\bfx_0))^\top \Gamma_{t-1} \\
\bfH_t \bfB_t + W_{t-1} \Gamma_{t-1} \\
\end{pmatrix}, \right. \\
&\quad\quad \left.
\begin{pmatrix} 
r(\bfx_0, \bfx_0; \bfphi_t) & \bfr_t^\top(\bfx_0)  \\
\bfr_t(\bfx_0)& \bfR_t  \\
\end{pmatrix}, \bfSigma_t
\right),
\end{split}
\end{align}
where $\bfr_t(\bfx_0):=r(\mathcal{X}_t, \bfx_0; \bfphi_t)$. Thus the  predictive distribution of $\bfy(\bfx_0):=[\bfy_1(\bfx_0),$ $\ldots, \bfy_N(\bfx_0)]$ given $\bfy^{\mathscr{D}}$ and  $\bfB, \Gamma, \bfSigma, \bfphi$ is 
\begin{align} \label{eqn: joint pred dist MN}
\begin{split}
\pi(\bfy(\bfx_0)\mid \bfy^{\mathscr{D}}, \bfB, \Gamma, \bfSigma, \bfphi) &= p(\bfy^1(\bfx_0) \mid \bfy^{1, \mathscr{D}}, \bfB_1, \bfSigma_1, \bfphi_1) \\
&\quad \times \prod_{t=1}^s p(\bfy^t(\bfx_0) \mid \bfy^{t, \mathscr{D}}, \bfy^{t-1, \mathscr{D}}, \bfB_t, \Gamma_t, \bfSigma_t, \bfphi_t),
\end{split}
\end{align}
with 
\begin{align}
\begin{split}
p(\bfy^1(\bfx_0) \mid \bfy^{1, \mathscr{D}}, \bfB_1, \bfSigma_1, \bfphi_1) &=\mathcal{MN}_{1, N} \biggr(E[\bfy^t(\bfx_0) \mid \bfy^{1,\mathscr{D}}, \bfB_1, \bfSigma_1, \bfphi_1], \Lambda_1, \bfSigma_1 \biggr),  \\
p(\bfy^t(\bfx_0) \mid \bfy^{t, \mathscr{D}}, \bfy^{t-1, \mathscr{D}}, \bfB_t, \Gamma_t, \bfSigma_t, \bfphi_t) &= \mathcal{MN}_{1, N}\biggr(E[\bfy^t(\bfx_0) \mid \bfy^{t, \mathscr{D}}, \bfy^{t-1, \mathscr{D}}, \bfB_t, \Gamma_{t-1}, \bfSigma_t, \\
&\quad \quad \bfphi_t], \Lambda_t, \bfSigma_t\biggr),
\end{split}
\end{align}
where 
\begin{align} \label{eqn: predictive quantities}
\begin{split}
E[\bfy^1(\bfx_0) \mid \bfy^{1, \mathscr{D}}, \bfB_1, \bfSigma_1, \bfphi_1] &:= \bfh^\top_1(\bfx_0) \bfB_1 + \bfr_1^\top(\bfx_0) \bfR^{-1}_1\{\bfy^{1, \mathscr{D}} - \bfH_1 \bfB_1  \}, \\
\Lambda_1&:= r(\bfx_0, \bfx_0; \bfphi_1) - \bfr_1^\top(\bfx_0)\bfR_1^{-1} \bfr_1(\bfx_0), \\
E[\bfy^t(\bfx_0) \mid \bfy^{t, \mathscr{D}}, \bfB_t, \Gamma_{t-1}, \bfSigma_t, \bfphi_t] &:= \bfh_t(\bfx_0) \bfB_t  + (\bfy^{t-1}(\bfx_0))^\top\Gamma_{t-1} \\
& \quad  + \bfr_t^\top(\bfx_0) \bfR^{-1}_t\{\bfy^{t, \mathscr{D}} - \bfH_t \bfB_t -W_{t-1}\Gamma_{t-1} \}, \\
\Lambda_t&:= r(\bfx_0, \bfx_0; \bfphi_t) - \bfr_t^\top(\bfx_0)\bfR_t^{-1} \bfr_t(\bfx_0).
\end{split}
\end{align}

Using the objective prior \eqref{eqn: constant prior in SEP cokriging} yields the posterior of $\bfB, \bfrho:=[\bfrho_{1}, \ldots, \bfrho_{s-1}]^\top$ given $\bfy^{\mathscr{D}},\bfSigma, \bfphi$ as a product of posteriors:
\begin{align*}
\pi(\bfB, \bfrho \mid \bfy^{\mathscr{D}},\bfSigma, \bfphi) &= \pi(\bfB_1\mid \bfy^{1, \mathscr{D}},\bfSigma_1, \bfphi_1)  \prod_{t=2}^s \pi(\bfB_t, \bfrho_{t-1}\mid \bfy^{t, \mathscr{D}}, \bfy^{t-1, \mathscr{D}}, \bfSigma_t, \bfphi_t),
\end{align*} 
with 
\begin{align} \label{eqn: posterior t}
\begin{split}
\bfB_1 \mid \bfy^{1, \mathscr{D}}, \bfSigma, \bfphi &\sim \mathcal{MN}_{q_1, N}\biggr(\hat{\bfB}_1, [\bfH_1^\top\bfR_1^{-1}\bfH_1]^{-1}, \bfSigma_1\biggr), \\
\begin{pmatrix}
\bfB_t \\
\bfrho_{t-1}^\top
\end{pmatrix}
\mid \bfy^{t,\mathscr{D}}, \bfy^{t-1,\mathscr{D}}, \bfSigma, \bfphi & \sim \mathcal{MN}_{q_t+1, N}\biggr(
\begin{pmatrix}
\hat{\bfB}_t \\
\hat{\bfrho}_{t-1}^\top
\end{pmatrix}, 
[\bfT_t^\top\bfR_t^{-1}\bfT_t]^{-1}, \bfSigma_t\biggr),
\end{split}
\end{align}
where $\bfT_t:=[\bfH_t, W_{t-1}]$. $\hat{\bfB}_1:=[\bfH_1^\top\bfR_1^{-1}\bfH_1]^{-1} \bfH_1^\top\bfR_1^{-1} \bfy^{1,\mathscr{D}}$. $[\hat{\bfB}_t^\top, \hat{\bfrho}_{t-1}]^\top:=[\bfT_t^\top\bfR_t^{-1}\bfT_t]^{-1}$ $\bfT_t^\top\bfR_t^{-1} \bfy^{t,\mathscr{D}}$. 

\noindent \textbf{(1) Predictive Mean:}\\
For the $j$th coordinate at level $t$, it follows from \eqref{eqn: posterior t} that 
\begin{align} \label{eqn: posterior j}
\begin{split}
\bfbeta_{1,j} \mid \bfy^{1, \mathscr{D}}, \bfSigma, \bfphi &\sim \mathcal{N}_{q_1}\biggr(\hat{\bfbeta}_{1,j}, \bfSigma_1^{jj}[\bfH_1^\top\bfR_1^{-1}\bfH_1]^{-1}\biggr), \\
\begin{pmatrix}
\bfbeta_{t,j} \\
\gamma_{t-1,j}
\end{pmatrix}
\mid \bfy^{t,\mathscr{D}}, \bfy^{t-1,\mathscr{D}}, \bfSigma, \bfphi & \sim \mathcal{N}_{q_t+1}\left(
\begin{pmatrix}
\hat{\bfbeta}_{t,j} \\
\hat{\gamma}_{t-1,j}
\end{pmatrix}, 
\bfSigma_t^{jj}[\bfT_{t,j}^\top\bfR_t^{-1}\bfT_{t,j}]^{-1}\right),
\end{split}
\end{align}
where $\hat{\bfbeta}_{1,j}:=[\bfH_1^\top\bfR_1^{-1}\bfH_1]^{-1} \bfH_1^\top\bfR_1^{-1} \bfy_{1,j}$. $(\hat{\bfbeta}_{t,j}^\top, \hat{\gamma}_{t-1,j})^\top:=[\bfT_{t,j}^\top\bfR_t^{-1}\bfT_{t,j}]^{-1} \bfT_{t,j}^\top\bfR_t^{-1} \bfy_{t,j}$.

For the $j$th coordinate at level $t=1$, it follows from \eqref{eqn: predictive quantities} that the predictive mean in the separable autoregressive cokriging emulator is 
\begin{align} \label{eqn: mean level 1}
E[y_{1,j}(\bfx_0) \mid \bfy^{\mathscr{D}},\bfB, \bfrho, \bfSigma, \bfphi] &= \bfh_1^\top(\bfx_0) \bfbeta_{1,j} + \bfr_1^\top(\bfx_0)\bfR_1^{-1}[\bfy_{1,j} - \bfH_1\bfbeta_{1,j}] 
\end{align}
Taking the expectation over $\bfbeta_{1,j}$ w.r.t.~the posterior \eqref{eqn: posterior j} in \eqref{eqn: mean level 1} yields the same formula for the posterior predictive mean $\hat{y}_{1,j}(\bfx_0)$ in the PP cokriging emulator in Lemma~\ref{thm: prediction in nested design} under nested design. 

For the $j$th coordinate at level $t>1$, it follows from \eqref{eqn: predictive quantities} that the posterior predictive mean in the separable autoregressive cokriging emulator is 
\begin{align} \label{eqn: mean level t}
\begin{split}
E[y_{t,j}(\bfx_0) \mid \bfy^{\mathscr{D}},\bfB, \bfrho, \bfSigma, \bfphi] &= E_{[y_{t-1,j}(\bfx_0) \mid \bfy^{\mathscr{D}},\bfB, \bfrho, \bfSigma, \bfphi]} \{ E_{[y_{t,j}(\bfx_0) \mid \bfy^{\mathscr{D}}, y_{t-1,j}(\bfx_0), \bfB, \bfrho, \bfSigma, \bfphi]}[y_{t}(\bfx_0)] \} \\
&=E_{[y_{t-1,j}(\bfx_0) \mid \bfy^{\mathscr{D}},\bfB, \bfrho, \bfSigma, \bfphi]} \{ \bfh_t^\top(\bfx_0) \bfbeta_{t,j} + y_{t-1,j}(\bfx_0)\gamma_{t-1,j} \\
&\quad + \bfr_t^\top(\bfx_0)\bfR_t^{-1}[\bfy_{t,j} - \bfH_t\bfbeta_{t,j} - y_{t-1,j}(\mathcal{X}_t) \gamma_{t-1,j} ] \} \\
&= \bfh_t^\top(\bfx_0) \bfbeta_{t,j} + \hat{y}_{t-1,j}(\bfx_0)\gamma_{t-1,j} \\
&\quad  + \bfr_t^\top(\bfx_0)\bfR_t^{-1}[\bfy_{t,j} - \bfH_t\bfbeta_{t,j} - y_{t-1,j}(\mathcal{X}_t) \gamma_{t-1,j} ]. 
\end{split}
\end{align}
Taking the expectation over $\bfbeta_{t,j}, \gamma_{t-1,j}$ w.r.t.~the posterior \eqref{eqn: posterior j} in \eqref{eqn: mean level t} yields the same formula for the posterior predictive mean $\hat{y}_{t,j}(\bfx_0)$ in the PP cokriging emulator in Lemma~\ref{thm: prediction in nested design} under nested design. 

\noindent \textbf{(2) Predictive Variance:}\\
The proof of predictive variance consist of two steps: In Step 1, we derive the expression for $\text{Var}\{y_{t,j}(\bfx_0) \mid \bfy^{\mathscr{D}}, \bfSigma, \bfphi \}$; in Step 2, we derive the formula for $\text{Var}\{y_{t,j}(\bfx_0) \mid \bfy^{\mathscr{D}}, \bfphi \}$ by marginalizing out $\bfSigma$.

Note that according to \cite{Conti2010}, the  predictive distribution of $\bfy(\bfx_0):=[\bfy_1(\bfx_0),$ $\ldots, \bfy_N(\bfx_0)]$ given $\bfy^{\mathscr{D}}$ and  $\bfSigma, \bfphi$ is 
\begin{align} \label{eqn: joint pred dist MN}
\begin{split}
\pi(\bfy(\bfx_0)\mid \bfy^{\mathscr{D}}, \bfSigma, \bfphi) &= p(\bfy^1(\bfx_0) \mid \bfy^{1, \mathscr{D}}, \bfphi_1)  \prod_{t=1}^s p(\bfy^t(\bfx_0) \mid \bfy^{t, \mathscr{D}}, \bfy^{t-1, \mathscr{D}},  \bfSigma_t, \bfphi_t),
\end{split}
\end{align}
with 
\begin{align}
\begin{split}
p(\bfy^1(\bfx_0) \mid \bfy^{1, \mathscr{D}}, \bfSigma_1, \bfphi_1) &=\mathcal{MN}_{1, N} \biggr(E[\bfy^t(\bfx_0) \mid \bfy^{1,\mathscr{D}}, \bfSigma_1, \bfphi_1], \Lambda_1^*, \bfSigma_1 \biggr),  \\
p(\bfy^t(\bfx_0) \mid \bfy^{t, \mathscr{D}}, \bfy^{t-1, \mathscr{D}}, \bfSigma_t, \bfphi_t) &= \mathcal{MN}_{1, N}\biggr(E[\bfy^t(\bfx_0) \mid \bfy^{t, \mathscr{D}}, \bfy^{t-1, \mathscr{D}},  \bfSigma_t, \bfphi_t], \Lambda_t^*, \bfSigma_t\biggr),
\end{split}
\end{align}
where 
\begin{align} \label{eqn: marginal predictive quantities}
\begin{split}
E[\bfy^1(\bfx_0) \mid \bfy^{1, \mathscr{D}},  \bfSigma_1, \bfphi_1] &:= \bfh^\top_1(\bfx_0) \hat{\bfB}_1 + \bfr_1^\top(\bfx_0) \bfR^{-1}_1\{\bfy^{1, \mathscr{D}} - \bfH_1 \hat{\bfB}_1  \}, \\
\Lambda_1^*&:= \Lambda_1 + \bigr\{\bfh_1(\bfx_0) - \bfH_1^\top \bfR_1^{-1}\bfr_1(\bfx_0) \bigr\}^\top(\bfH_1^\top \bfR_1^{-1} \bfH_1)^{-1} \\
&\quad\quad \times \bigr\{\bfh_1(\bfx_0) - \bfH_1^\top \bfR_1^{-1}\bfr_1(\bfx_0) \bigr\}, \\
E[\bfy^t(\bfx_0) \mid \bfy^{t, \mathscr{D}}, \bfSigma_t, \bfphi_t] &:= \bfh_t^\top(\bfx_0) \hat{\bfB}_t  + (\bfy^{t-1}(\bfx_0))^\top\hat{\Gamma}_{t-1} \\
& \quad  + \bfr_t^\top(\bfx_0) \bfR^{-1}_t\{\bfy^{t, \mathscr{D}} - \bfT_t \hat{\bfB}_t -W_{t-1}\hat{\Gamma}_{t-1} \}, \\
\Lambda_t^*&:= \Lambda_t + \bigr\{\bfT_t(\bfx_0) - \bfT_t^\top \bfR_t^{-1}\bfr_t(\bfx_0) \bigr\}^\top(\bfT_t^\top \bfR_t^{-1} \bfT_t)^{-1} \\
&\quad\quad \times \bigr\{\bfT_t(\bfx_0) - \bfT_t^\top \bfR_t^{-1}\bfr_t(\bfx_0) \bigr\}, \\
\bfT_t(\bfx_0) &:=[\bfh_t(\bfx_0), (\bfy^{t-1}(\bfx_0))^\top] \\
\bfT_t &:= [\bfH_t, W_{t-1}]\\
\end{split}
\end{align}

\noindent \textbf{Step 1:} The formula for the predictive variance at level $t$ and $j$th spatial coordinate follows from the law of total variance: 
\begin{align} \label{eqn: conditional variance} 
\begin{split}
\text{Var}\{y_{t,j}(\bfx_0) \mid \bfy^{\mathscr{D}}, \bfSigma, \bfphi \}& = \text{Var}\{E[y_{t,j}(\bfx_0) \mid  \bfy^{\mathscr{D}}, y_{t-1,j}(\bfx_0),  \bfSigma, \bfphi ] \mid  \bfy^{\mathscr{D}}, \bfSigma, \bfphi \} \\
&\quad + E\{\text{Var}[y_{t,j}(\bfx_0) \mid \bfy^{\mathscr{D}},  \bfSigma, \bfphi, y_{t-1,j}(\bfx_0)] \mid \bfy^{\mathscr{D}},  \bfSigma, \bfphi \}
\end{split}
\end{align}
with 
\begin{align*}
\text{Var}\{E[y_{t,j}(\bfx_0) \mid  \bfy^{\mathscr{D}}, y_{t-1,j}(\bfx_0),  \bfSigma, \bfphi ] \} &=\text{Var}\{y_{t-1,j}(\bfx_0) \hat{\gamma}_{t-1,j} \mid \bfy^{\mathscr{D}}, \bfSigma, \bfphi\} \\
&=  \hat{\gamma}_{t-1,j}^2 v_{t-1,j}(\bfx_0)   \\
E\{\text{Var}[y_{t,j}(\bfx_0) \mid \bfy^{\mathscr{D}}, \bfSigma, \bfphi, y_{t-1,j}(\bfx_0)] \mid \bfy^{\mathscr{D}},  \bfSigma, \bfphi \} & = 
 E\{\bfSigma_{t}^{jj} \Lambda_t^* \mid \bfy^{\mathscr{D}},  \bfSigma, \bfphi   \} \\
&= \bfSigma_{t}^{jj} \{ \Lambda_t + \kappa_{t,j} \},
\end{align*}
where $\bfSigma_t^{jj} = \sigma_{t,j}^2$.

\noindent \textbf{Step 2:} Applying the law of total variance again yields 
\begin{align*}
\text{Var}\{y_{t,j}(\bfx_0) \mid \bfy^{\mathscr{D}}, \bfphi \} &= \text{Var}_{\bfSigma_t \mid \bfy^{\mathscr{D}}, \bfphi} \bigr\{E[y_{t,j}(\bfx_0) \mid \bfy^{\mathscr{D}}, \bfSigma, \bfphi ] \bigr\} \\
&\quad + E_{\bfSigma_t \mid \bfy^{\mathscr{D}}, \bfphi} \bigr\{\text{Var}[y_{t,j}(\bfx_0) \mid \bfy^{\mathscr{D}}, \bfSigma, \bfphi] \bigr\}.
\end{align*}
Note that the first term above is zero, since the posterior mean does not depend on $\bfSigma_t$ according to~\eqref{eqn: mean level t}. Thus, it follows from ~\eqref{eqn: conditional variance} that
\begin{align*}
\text{Var}\{y_{t,j}(\bfx_0) \mid \bfy^{\mathscr{D}}, \bfphi \} & = \hat{\gamma}_{t-1,j}^2 v_{t-1,j}(\bfx_0) +   E\{\bfSigma_{t}^{jj} \mid \bfy^{\mathscr{D}}, \bfSigma, \bfphi\}  \{ \Lambda_t + \kappa_{t,j} \},
\end{align*}
as desired.

\end{proof} 

\section{Ancillary Results} \label{app: useful results}

\begin{lemma} \label{lem: determinant}
Let $\bfT_{t,j}:=[\bfH_t, W_{t-1,j}]$ and $\bfQ_t:=\bfR_t^{-1}\{\mathbf{I} - \bfT_{t,j} (\bfT_{t,j}^\top \bfR_t^{-1}$\ $ \bfT_{t,j})^{-1} \bfT_{t,j}^\top \bfR_t^{-1}\}$ for $t>1$. Then the following result holds: 
\begin{align*}
    \begin{split}
    S^2(\bfphi_t, \bfy_{t,j}) & = \bfy_{t,j}^\top \bfQ_t^H \bfy_{t,j} - \bfy_{t,j}^\top \bfK_t \bfOmega_{t-1,j} \bfK_t \bfy_{t,j} 
     + 2\bfy_{t,j}^\top \bfR_t^{-1} \bfOmega_{t-1,j} \bfK_t \bfy_{t,j} \\
     &- \bfy_{t,j}^\top  \bfR_t^{-1} \bfOmega_{t-1,j} \bfR_t^{-1} \bfy_{t,j}
    \end{split}
\end{align*}
where $\bfQ_t^H:=\bfR_t^{-1}\{\mathbf{I} - \bfH_t (\bfH_t^\top \bfR_t^{-1} \bfH_t)^{-1} \bfH_t^\top \bfR_t^{-1}\}$, $\bfK_t:=\bfR_t^{-1}\bfH_t(\bfH_t^\top \bfR_t^{-1}$\ $\bfH_t)^{-1} \bfH_t^\top \bfR_t^{-1}$, and $\bfOmega_{t-1,j}:=W_{t-1,j} (W_{t-1,j}^{\top} \bfQ_t^H W_{t-1,j})^{-1} W_{t-1,j}^{\top}$. These two identities show that the computational cost to compute $\sum_{j=1}^N \ln|\bfT_{t,j}^\top \bfR_t^{-1} \bfT_{t,j}|$ and $\sum_{j=1}^N \ln S^2(\bfphi_t, \bfy_{t,j})$ is reduced from $O(Nn_t^3)$ to $O(n_t^3+Nn_t^2)$.
\end{lemma}

\begin{proof}[\textbf{Proof of Lemma~\ref{lem: determinant}}]
Notice that 
\begin{align*}
    \bfT_{t,j}^\top \bfR_t^{-1} \bfT_{t,j} &= \begin{pmatrix}
    \bfH_t^\top \bfR_t^{-1} \bfH_t & \bfH_t^\top \bfR_t^{-1} W_{t-1,j} \\
    W_{t-1,j}^{\top} \bfR_t^{-1} \bfH_t & W_{t-1,j}^{\top} \bfR_t^{-1} W_{t-1,j} 
    \end{pmatrix}, \\
    (\bfT_{t,j}^\top \bfR_t^{-1} \bfT_{t,j})^{-1} &=
     \begin{pmatrix}
    \bfA_t^{-1} + \bfA_t^{-1} \bfH_t^\top \bfR_t^{-1} \bfOmega_{t-1,j}\bfR_t^{-1} \bfH_t \bfA_t^{-1} & -\mathbf{M}_t \\
    -\mathbf{M}_t^\top & \omega_{t-1}
    \end{pmatrix}, 
\end{align*}
where $\omega_{t-1}:=(W_{t-1,j}^\top \bfQ_t^H W_{t-1,j})^{-1}$, $\bfA_t := \bfH_t^\top \bfR_t^{-1} \bfH_t$, $\mathbf{M}_t:=\bfA_t^{-1}\bfH_t^\top \bfR_t^{-1}$\ $W_{t-1,j}\omega_{t-1}$.
Using block matrix determinant yields the formula to compute the determinant of $\bfT_{t,j}^\top \bfR_t^{-1} \bfT_{t,j}$. Using the block matrix inverse yields 
\begin{align*}
    \bfT_{t,j}(\bfT_{t,j}^\top \bfR_t^{-1} \bfT_{t,j})^{-1} \bfT_{t,j}^\top &= \bfH_t \bfA_t^{-1} \bfH_t^\top \\
    &+ \bfH_t\bfA_t^{-1}\bfH_t^\top \bfR_t^{-1} \bfOmega_{t-1,j}\bfR_t^{-1}\bfH_t \bfA_t^{-1} \bfH_t^\top \\
    &- W_{t-1,j}\omega_{t-1}W_{t-1,j}^\top \bfR_t^{-1} \bfH_t \bfA_t^{-1} \bfH_t^\top \\
    &-\bfH_t \bfA_t^{-1}\bfH_t^\top \bfR_t^{-1} W_{t-1,j}\omega_{t-1} W_{t-1,j}^\top \\
    &+ W_{t-1,j} \omega_{t-1} W_{t-1,j}^\top. 
\end{align*}
Combining these terms yields the formula to compute $S^2(\bfphi_t, \bfy_{t,j})$. 
\end{proof}

\section{Additional Results} \label{sec: additional results} 
This section gives some figures referenced in Section~\ref{sec: application}. Figure~\ref{fig: prediction versus held-out 1 uncertainty} shows the scatter plot of predicted PSE against held-out PSE at input setting $\bfx_2$ over randomly sampled 1,000 spatial locations with 95\% percentile predictive intervals.
Figure~\ref{fig: prediction versus held-out 2} shows the scatter plot of predicted PSE against held-out PSE at input setting $\bfx_2$. It suggests that the PP cokriging emulator gives better prediction than the PP kriging emulation at input setting $\bfx_2$.  Figure~\ref{fig: prediction versus held-out 2 uncertainty} shows the scatter plot of predicted PSE against held-out PSE at input setting $\bfx_2$ over randomly sampled 1,000 spatial locations with 95\% percentile predictive intervals. All these results also confirm that the PP cokriging performs quite well for the real application in terms of prediction uncertainty. 

\begin{figure}[htbp] 
\begin{subfigure}{0.5\textwidth}
  \centering
\makebox[\textwidth][c]{ \includegraphics[width=1.0\linewidth, height=0.25\textheight]{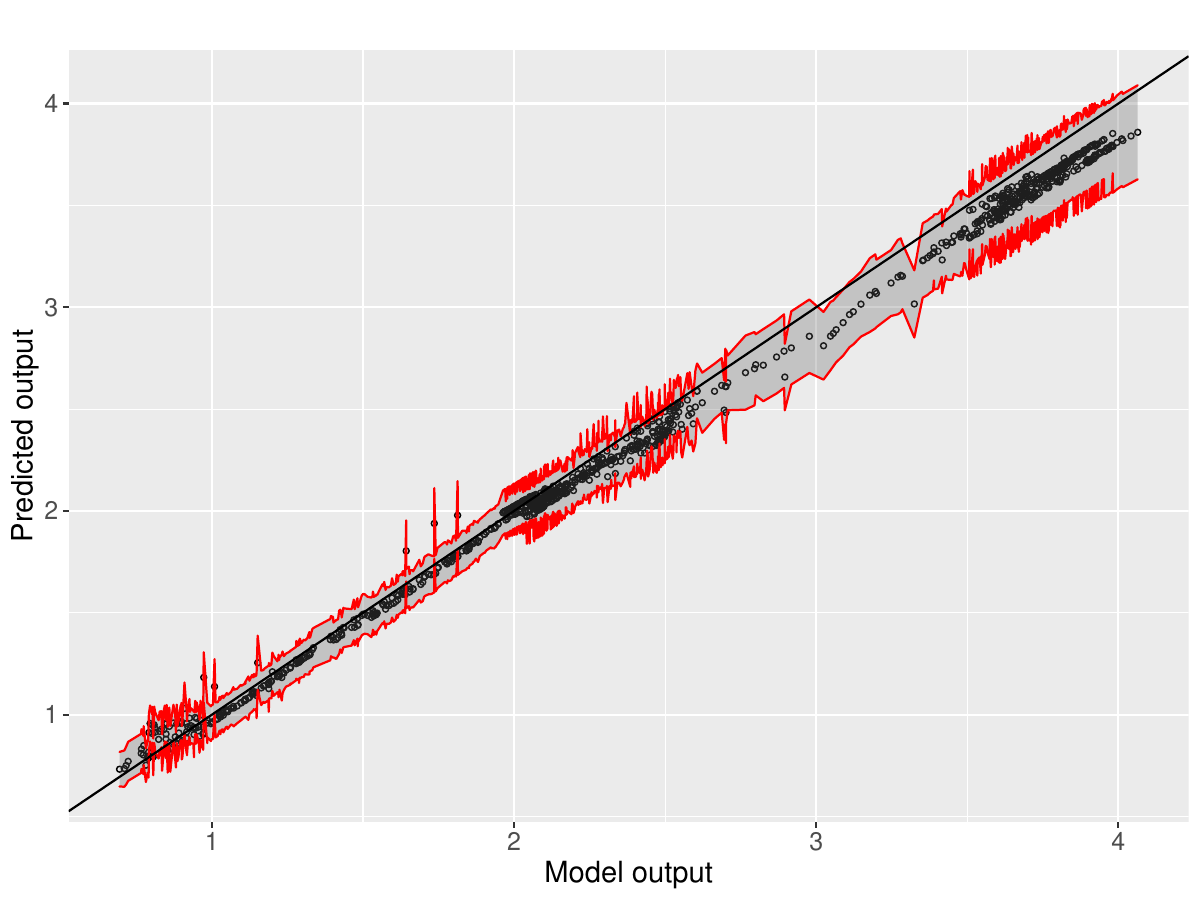}}
  \caption{PP kriging}
\end{subfigure}%
\begin{subfigure}{.5\textwidth}
  \centering
  \includegraphics[width=1.0\linewidth,height=0.25\textheight]{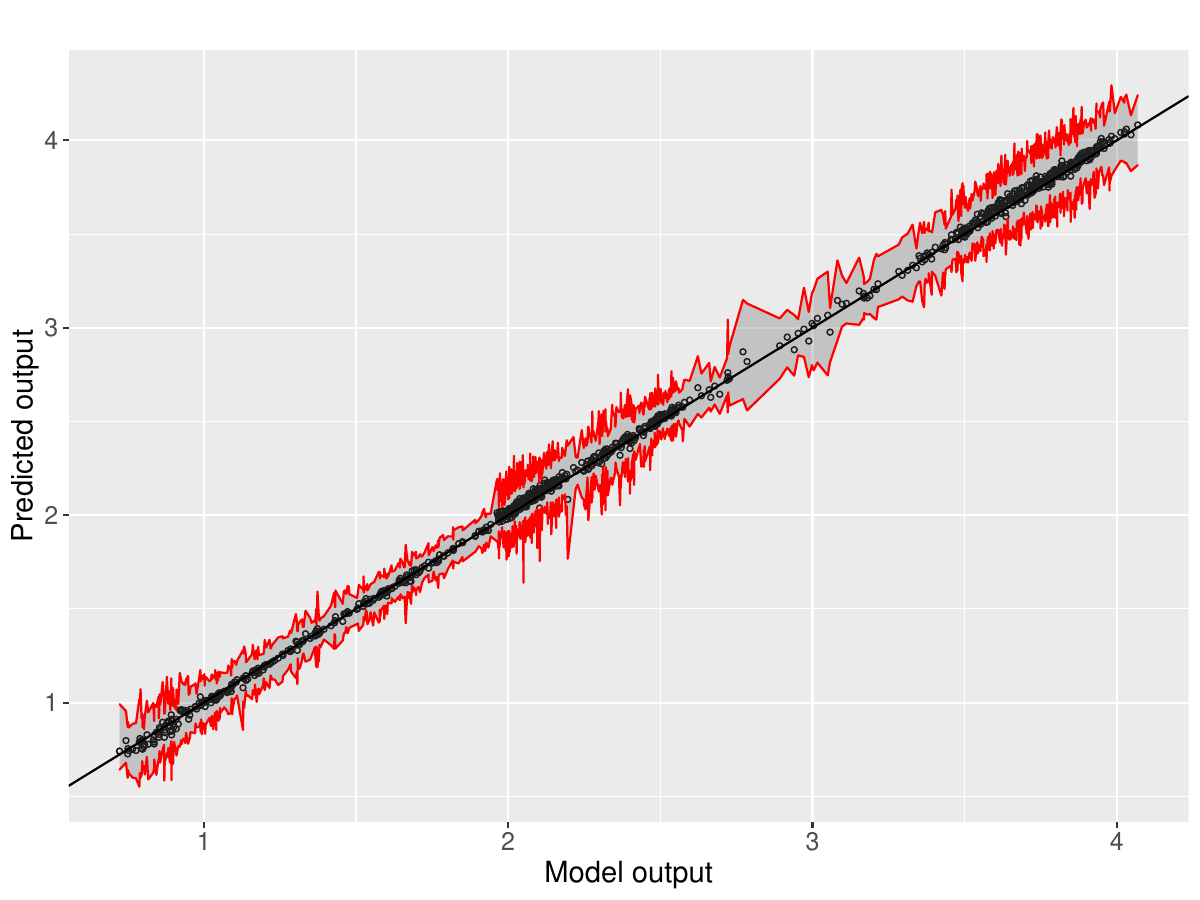}
  \caption{PP cokriging}
\end{subfigure}

 \caption{Scatter plot of predicted PSE against held-out PSE over randomly sampled $N=1,000$ spatial locations at the input setting $\bfx_1$. The red curve shows the 95\% percentile predictive intervals.}
\label{fig: prediction versus held-out 1 uncertainty}
\end{figure}

\begin{figure}[htbp]
\begin{subfigure}{0.5\textwidth}
  \centering
\makebox[\textwidth][c]{ \includegraphics[width=1.0\linewidth, height=0.25\textheight]{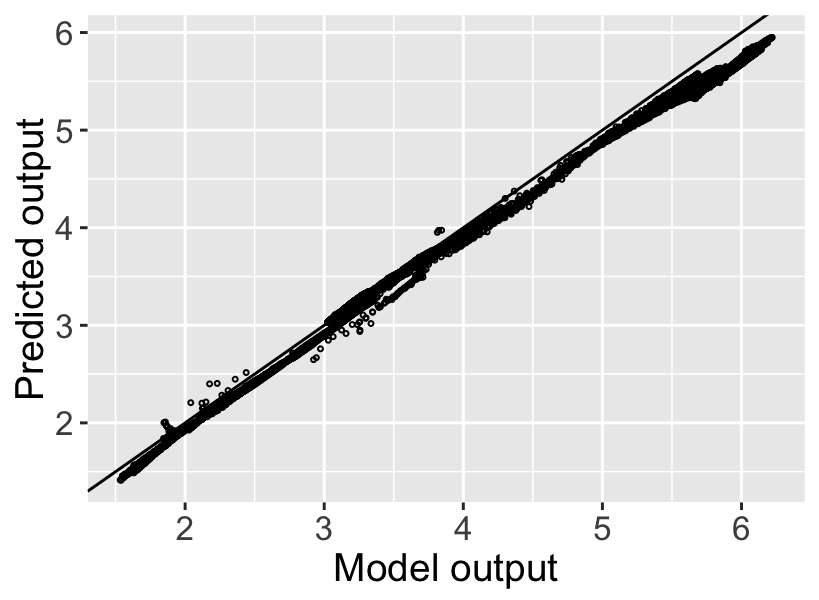}}
  \caption{PP kriging}
\end{subfigure}%
\begin{subfigure}{.5\textwidth}
  \centering
  \includegraphics[width=1.0\linewidth,height=0.25\textheight]{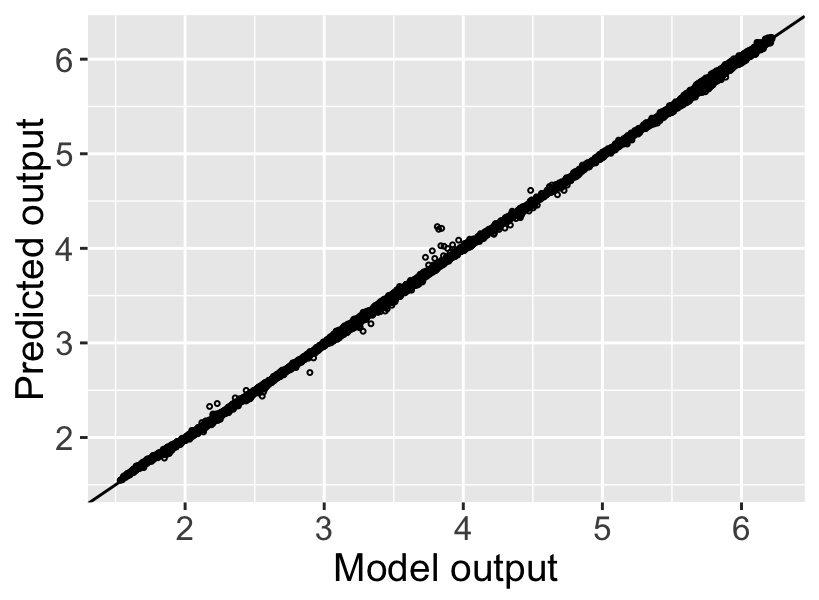}
  \caption{PP cokriging}
\end{subfigure}

 \caption{Scatter plot of predicted PSE against held-out PSE over $N=9,284$ spatial locations at the input setting $\bfx_2$. }
\label{fig: prediction versus held-out 2}
\end{figure}

\begin{figure}[htbp]
\begin{subfigure}{0.5\textwidth}
  \centering
\makebox[\textwidth][c]{ \includegraphics[width=1.0\linewidth, height=0.25\textheight]{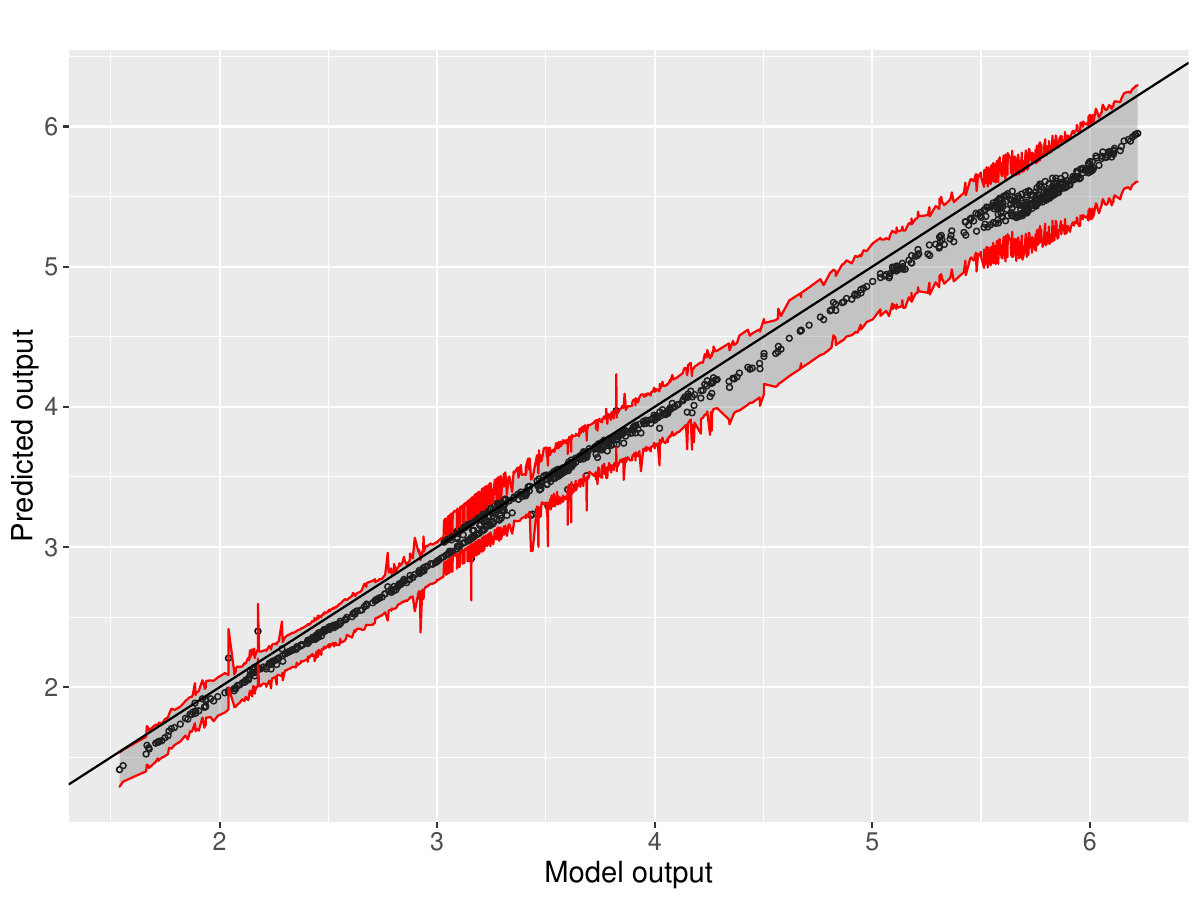}}
  \caption{PP kriging}
\end{subfigure}%
\begin{subfigure}{.5\textwidth}
  \centering
  \includegraphics[width=1.0\linewidth,height=0.25\textheight]{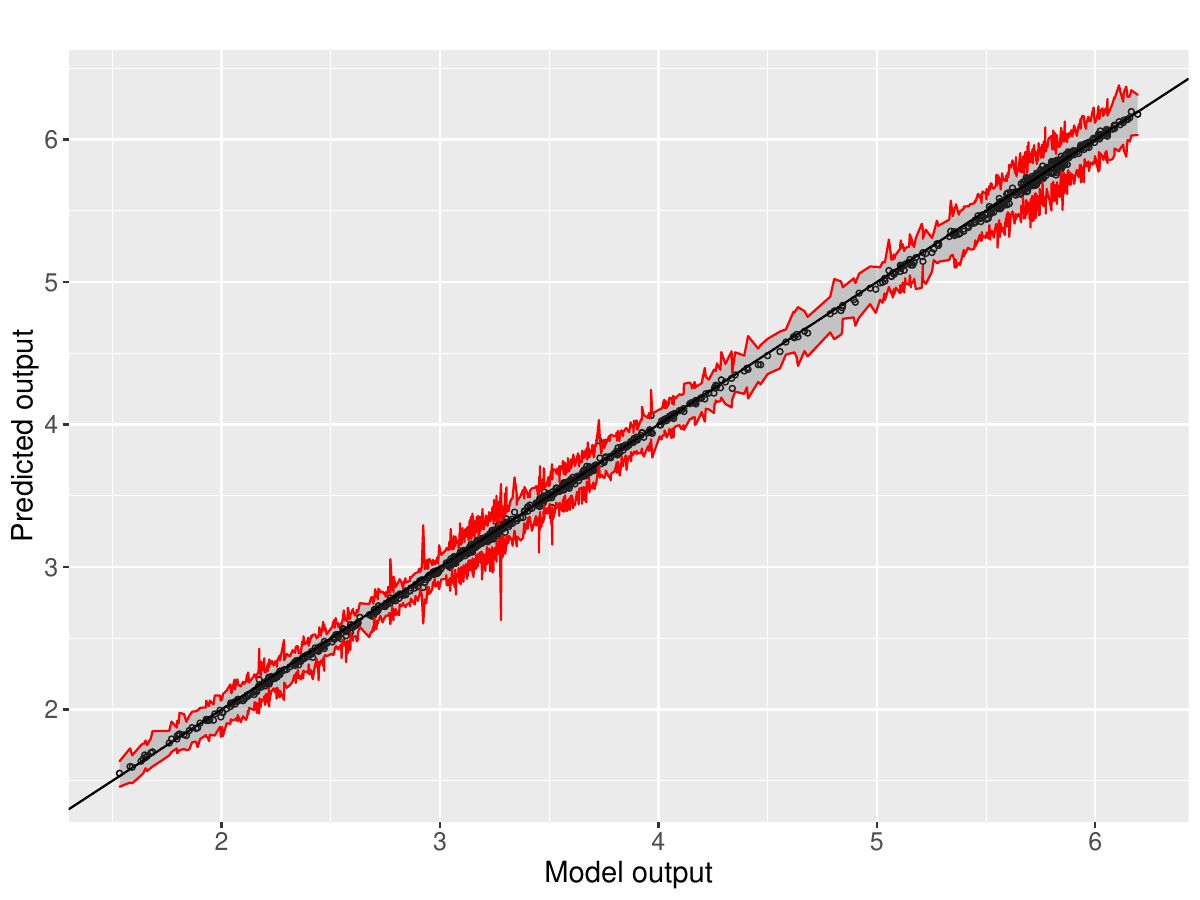}
  \caption{PP cokriging}
\end{subfigure}

 \caption{Scatter plot of predicted PSE against held-out PSE over randomly sampled $N=1,000$ spatial locations at the input setting $\bfx_2$. The red curve shows the 95\% percentile predictive intervals.}
\label{fig: prediction versus held-out 2 uncertainty}
\end{figure}

\end{document}